\documentclass[12pt]{article}
\usepackage[T1]{fontenc}
\usepackage[japanese,english]{babel}
\usepackage{cprotect}
\usepackage{float}
\usepackage{url}
\usepackage{multirow}
\usepackage{amsmath}
\usepackage{amsthm}
\usepackage{amssymb}
\usepackage{graphicx}
\usepackage{geometry}
\usepackage{setspace}
\usepackage[authoryear]{natbib}
\usepackage[unicode=true,pdfusetitle,
 bookmarks=true,bookmarksnumbered=false,bookmarksopen=false,
 breaklinks=false,pdfborder={0 0 1},backref=false,colorlinks=false]
 {hyperref}
\hypersetup{
 pdfborderstyle=}

\makeatletter

\newcommand*\LyXZeroWidthSpace{\hspace{0pt}}
\let\SF@@footnote\footnote
\def\footnote{\ifx\protect\@typeset@protect
    \expandafter\SF@@footnote
  \else
    \expandafter\SF@gobble@opt
  \fi
}
\expandafter\def\csname SF@gobble@opt \endcsname{\@ifnextchar[
  \SF@gobble@twobracket
  \@gobble
}
\edef\SF@gobble@opt{\noexpand\protect
  \expandafter\noexpand\csname SF@gobble@opt \endcsname}
\def\SF@gobble@twobracket[#1]#2{}
\providecommand{\tabularnewline}{\\}
\floatstyle{ruled}
\newfloat{algorithm}{tbp}{loa}
\providecommand{\algorithmname}{Algorithm}
\floatname{algorithm}{\protect\algorithmname}

\theoremstyle{plain}
\newtheorem{prop}{\protect\propositionname}
\theoremstyle{definition}
 \newtheorem{example}{\protect\examplename}
\theoremstyle{plain}
\newtheorem{lem}{\protect\lemmaname}
\theoremstyle{remark}
\newtheorem{rem}{\protect\remarkname}

\usepackage[english]{babel}
\bibpunct{(}{)}{,}{a}{,}{,}
\usepackage{url}
\theoremstyle{plain}

\addto\captionsenglish{\renewcommand{\algorithmname}{Algorithm}}

\usepackage{lscape}
\usepackage[bottom]{footmisc}
\usepackage{multirow}
\usepackage{array}

\newtheorem{assumption}{Assumption}

\newtheorem{ass}{Assumption}

\usepackage{thmtools}

\providecommand{\examplename}{Example}
\providecommand{\lemmaname}{Lemma}
\providecommand{\propositionname}{Proposition}
\providecommand{\remarkname}{Remark}

\makeatother

\addto\captionsenglish{\renewcommand{\examplename}{Example}}
\addto\captionsenglish{\renewcommand{\lemmaname}{Lemma}}
\addto\captionsenglish{\renewcommand{\propositionname}{Proposition}}
\addto\captionsenglish{\renewcommand{\remarkname}{Remark}}
\addto\captionsjapanese{\renewcommand{\algorithmname}{アルゴリズム}}
\addto\captionsjapanese{\renewcommand{\examplename}{例}}
\addto\captionsjapanese{\renewcommand{\lemmaname}{補題}}
\addto\captionsjapanese{\renewcommand{\propositionname}{命題}}
\addto\captionsjapanese{\renewcommand{\remarkname}{所見}}
\providecommand{\examplename}{Example}
\providecommand{\lemmaname}{Lemma}
\providecommand{\propositionname}{Proposition}
\providecommand{\remarkname}{Remark}

\begin{document}
\title{Sequential algorithm for structural estimations with equilibrium constraints\thanks{The current study builds on my prior work titled \textquotedblleft Sequential algorithm for estimating structural models without using analytical derivatives.\textquotedblright{} I thank Victor Aguirregabiria, Adam Dearing, Tianyi Li, Katsumi Shimotsu, and participants at APIOC 2024, the 2025 DSE conference in Hong Kong, 2026 IIOC, the Kansai Econometrics Conference, the Spring Meeting of the Operations Research Society of Japan, the 2026 Spring Meeting of the Japanese Statistical Association, and the 2026 Spring Meeting of the Japanese Economic Association for helpful discussions and comments.}}
\author{Takeshi Fukasawa\thanks{Waseda Institute for Advanced Study, Waseda University, 1-21-1 Nishiwaseda, Shinjuku, Tokyo, Japan. E-mail: fukasawa3431@gmail.com.\protect \protect \\
 This study is supported by JSPS KAKENHI Grant Number JP24K22629.}}
\maketitle
\begin{abstract}
This study examines sequential algorithms with the Zero Jacobian Property (ZJP) for estimating structural models subject to equilibrium constraints. For Maximum Likelihood Estimation (MLE) and the Generalized Method of Moments (GMM), the current study shows that these algorithms attain fast (near-quadratic) local convergence in large samples to the solution of the constrained optimization problem. If consistent initial estimates of the parameters satisfying a suitable rate condition are available, the algorithms yield an asymptotically efficient estimator even after one iteration.

It then proposes a novel algorithm called the Sequential Linearly Constrained (SLC) algorithm, which is applicable to a broader class of structural models than existing methods. A key advantage of the SLC algorithm is that it can be implemented without explicitly computing the Jacobian of the equilibrium constraints and can be several times faster than the Nested Fixed Point (NFXP) approach. The current study illustrates its performance through two numerical experiments: a dynamic discrete game with time-varying unobserved heterogeneity and a dynamic demand model.

{\flushleft{\textbf{Keywords:} Sequential estimation; Zero-Jacobian property; SLC algorithm; NFXP; MPEC}}

\pagebreak{} 
\end{abstract}

\section{Introduction}

Structural estimations with equilibrium constraints are prevalent in economics. Examples include dynamic discrete choice model estimation (henceforth DDC; e.g., \citealp{rust1987optimal}) and demand estimation (e.g., \citealp{berry1995automobile}, henceforth BLP), which are widely applied in industrial organization, labor economics, and many other fields. Structural estimations with equilibrium constraints are also applied to general equilibrium models (e.g., \citealp{ahlfeldt2015economics}).

The traditional approach for estimating structural models with fixed-point constraints is the Nested Fixed Point (NFXP) algorithm. Generally, NFXP aims to solve the following constrained continuous optimization problem:

\begin{eqnarray*}
\min_{\theta,Y} & Q(\theta;Y)\\
s.t. & G(Y;\theta)=0
\end{eqnarray*}

In the problem above, $\theta$ denotes the structural parameters of the model, and $Y$ denotes the economic variables of the model.\footnote{$Y$ is also known as nuisance parameters.} $G(Y;\theta)=0$ denotes the equilibrium constraint. The idea of the NFXP algorithm is to solve $G(Y;\theta)=0$ for $Y$ given candidate structural parameters $\theta$ in the inner loop, and to search for the optimal $\theta$ that minimizes the objective function $Q\left(\theta;Y(\theta)\right)$ in the outer loop. However, as noted by \citet{su2012constrained}, the computational burden of NFXP can be a major obstacle.

The current study investigates sequential algorithms possessing a property called the ``Zero-Jacobian Property'' (hereafter, ZJP) in previous studies, such as \citet{aguirregabiria2002swapping}, \citet{kasahara2008pseudo} (NPL for single-agent DDC estimation), and \citet{dearing2024efficient} (EPL for dynamic game estimation).\footnote{NPL and EPL denote Nested Pseudo-Likelihood and Efficient Pseudo-Likelihood, respectively. Although not explicitly discussed, the ABLP algorithm proposed by \citet{lee2015computationally} for estimating static BLP (random-coefficient logit demand) models is also a sequential algorithm satisfying the ZJP.} In sequential algorithms, $\theta$ and $Y$ are updated iteratively. These studies discussed convergence properties and demonstrated favorable performance of their algorithms. The current study generalizes these algorithms by developing a more unified framework to theoretically analyze their convergence properties. It then proposes a novel algorithm termed the Sequential Linearly Constrained (SLC) algorithm, which can be easily applied to more general structural models than those considered in previous studies and is often much faster than the NFXP algorithm. Through numerical experiments on dynamic discrete games with time-varying unobserved heterogeneity and on dynamic BLP models, the present study demonstrates that the SLC can reduce computational time by several times relative to the NFXP.

Concerning the convergence properties of the sequential algorithms with the ZJP, this study demonstrates that their local convergence rate is nearly quadratic, provided that the sample size is sufficiently large and the model is correctly specified\footnote{The Supplemental Appendix briefly discusses model misspecification.}---even when a consistent initial estimate of $Y$ is unavailable. Here, ``near-quadratic local convergence in large samples'' means that the local convergence rate is quadratic with probability one as the sample size approaches infinity. Because it is unclear whether general methods for obtaining a consistent initial estimate of $Y$ are available for non-standard structural models, this lack of dependence on the availability of an initial consistent estimate is a notable strength.\footnote{For a discussion of the contribution of the current study relative to \citet{dearing2024efficient}, see Section \ref{subsec:Local-convergence-speed}.} Furthermore, the current study shows that these discussions apply not only to Maximum Likelihood Estimation (MLE) but also to the Generalized Method of Moments (GMM), whereas previous studies focused on MLE. Both MLE and GMM share the property $\nabla_{Y}Q\left(\widehat{\theta},\widehat{Y}\right)=O_{p}\left(N^{-\frac{1}{2}}\right)$, where $\left(\widehat{\theta},\widehat{Y}\right)$ denotes the solution of (\ref{eq:obj}) and $N$ denotes the sample size, which allows for a unified discussion of local convergence rates. Given that many structural estimation applications employ GMM, these results substantially broaden the applicability of sequential algorithms. The current study also shows that the estimator attains the same statistical efficiency as the NFXP estimator even after one iteration when the iteration is started from consistent initial values, in line with previous methods utilizing the ZJP.

Building on the unified framework, the present study proposes a novel SLC algorithm as an alternative to previously developed sequential methods. Each iteration is essentially equivalent to solving a surrogate constrained optimization problem in which the constraint is replaced by a linearization of the original constraint $G\left(Y;\theta\right)=0$.\footnote{We can draw an analogy between the SLC update for constrained optimization and the classical one-step Newton--Raphson update for unconstrained optimization. Both updates can yield asymptotically efficient estimates when supplied with consistent initial estimates for MLE. While the Newton--Raphson update solves a surrogate problem based on a quadratic approximation of the objective function, the SLC update instead solves a surrogate problem based on a linear approximation of the constraint while continuing to use the original objective function.} The current study further discusses that the SLC algorithm can be implemented without explicitly computing the Jacobian of the equilibrium constraint function $G\left(Y;\theta\right)$, $\nabla_{Y}G\left(Y;\theta\right)$, which is advantageous because computing Jacobians can be cumbersome and prone to coding mistakes.\footnote{In recent years, automatic differentiation (AD) packages have become increasingly prevalent across many programming languages, allowing researchers to avoid manually coding the Jacobians of functions. However, this study finds that AD packages can sometimes be significantly slower than manually coded Jacobians, as illustrated by the examples in the replication code. Although such performance drawbacks may be mitigated by future technological advancements, practitioners should remain aware of the current limitations of AD-based implementations.} In addition, the matrix storing $\nabla_{Y}G\left(Y;\theta\right)$ itself may require substantial memory, especially when the dimension of $Y$ is large. Thus, a Jacobian-free approach is highly desirable.

A major advantage of the SLC is that it does not rely on specific structural features of the model, such as the existence of a fixed-point mapping with the ZJP (as in NPL) or the linearity of the equilibrium constraint $G(Y;\theta)$ with respect to the structural parameters $\theta$ (as in EPL). Consequently, SLC is in principle applicable to a wider range of structural models. Furthermore, the SLC possesses an additional desirable property in that each SLC update yields a descent direction for a natural merit function. By exploiting this property, we can also achieve global convergence of the SLC algorithm by introducing a line-search procedure under moderate conditions, as demonstrated in the Supplemental Appendix, although numerical results show that the SLC algorithm typically converges even without this procedure. SLC can be combined with fixed-point iteration acceleration methods (e.g., Spectral), whose effectiveness has been highlighted in recent economic studies. Note that, under the linearity of $G\left(Y;\theta\right)$ with respect to $\theta$, the SLC coincides with EPL, whose effectiveness has been demonstrated in dynamic games without unobserved heterogeneity by \citet{dearing2024efficient}.

Other than the NFXP approach and sequential algorithms, the ``MPEC'' (Mathematical Programming with Equilibrium Constraints) approach, which typically relies on Lagrangian-based constrained optimization solvers, is another approach to structural estimation. The current study shows that sequential algorithms with the ZJP, especially the SLC algorithm, share similarities with the Sequential Quadratic Programming (SQP) algorithm, which forms the basis of many solvers applied in previous studies.\footnote{Most previous studies proposing the MPEC approach (e.g., \citealp{su2012constrained}; \citealp{dube2012improving}; \citealp{egesdal2015estimating}) applied the commercial KNITRO solver, which builds on the Interior Point Method, in their numerical experiments.} Both SQP and SLC sequentially solve constrained optimization problems in which the constraint equation corresponds to a linearization of the equilibrium constraint $G(Y;\theta)=0$. The key difference lies in the objective function used in the constrained optimization problem: SLC uses the original objective function $Q\left(\theta;Y\right)$, whereas SQP uses a quadratic objective function based on the candidate value of the Lagrange multiplier.

Compared with NFXP and ``MPEC'' (Lagrangian-based) approaches, the SLC algorithm offers several advantages. Table \ref{tab:Comparison-of-algorithms} compares NFXP, SLC, non-Lagrangian-based sequential algorithms (e.g., EPL), and Lagrangian-based sequential algorithms (e.g., SQP). Notably, SLC is more robust in the presence of multiple solutions to $G\left(Y;\theta\right)=0$, and it is often significantly faster than NFXP, as demonstrated in the numerical experiments. Additionally, unlike Lagrangian-based methods, SLC can be implemented without explicitly computing the Jacobian $\nabla_{Y}G\left(Y;\theta\right)$.

\begin{table}[H]
\caption{Comparison of algorithms\label{tab:Comparison-of-algorithms}}

\begin{centering}
{\scriptsize{}}{\scriptsize{}%
\begin{tabular}{cccccc}
\hline 
 & {\scriptsize Sols. of} & {\scriptsize NFXP} & \multicolumn{2}{c}{{\scriptsize Sequential (Non-Largrangian-based)}} & {\scriptsize Sequential (Lagrangian-based)}\tabularnewline
 & {\scriptsize$G\left(Y;\theta\right)=0$} &  & {\scriptsize SLC} & {\scriptsize Others (e.g., EPL)} & {\scriptsize (e.g., SQP, IP; known as ``MPEC'')}\tabularnewline
\hline 
\multirow{2}{*}{{\scriptsize Global convergence possible?}} & {\scriptsize Unique} & {\scriptsize Yes} & {\scriptsize Yes} & {\scriptsize ?} & {\scriptsize Yes}\tabularnewline
 & {\scriptsize Multiple} & {\scriptsize ?} & {\scriptsize Yes} & {\scriptsize ?} & {\scriptsize Yes}\tabularnewline
\multicolumn{2}{c}{{\scriptsize Jacobian-free possible?}} & {\scriptsize Yes} & {\scriptsize Yes} & {\scriptsize Yes} & {\scriptsize ?{*}}\tabularnewline
\hline 
\end{tabular}}{\scriptsize\par}
\par\end{centering}
{\footnotesize Notes. ``?'' denotes whether the property holds or not is unclear in the previous studies. Concerning Lagrangian-based approaches, explicit Jacobian computation can be avoided when both the Jacobian--vector product (JVP) and the vector--Jacobian product (VJP) can be evaluated without forming the Jacobian itself; see also the Supplemental Appendix.}{\footnotesize\par}
\end{table}

The contribution of the current paper is twofold. First, it develops a unified theoretical framework for analyzing the convergence properties of sequential algorithms possessing the ZJP, extending prior work in this area. Second, it proposes the SLC algorithm, which is computationally more efficient, applicable to a broader class of structural models than previously proposed methods, and implementable with only minor modifications to the widely used NFXP algorithm.

The rest of this paper is organized as follows. Section \ref{sec:Literature} reviews the related literature. Section \ref{sec:Sequential-algorithms} describes sequential algorithms utilizing the ZJP. Note that the overall discussion in Section \ref{sec:Sequential-algorithms} does not rely on any statistical properties specific to structural estimation applications. Instead, it identifies key factors underlying the favorable behavior of the sequential algorithms, which are asymptotically satisfied in structural estimation applications. Section \ref{sec:structural-estimation} then examines the properties of the sequential algorithms in the context of structural estimation, incorporating formal asymptotic statistical analysis. Section \ref{sec:Numerical-experiments} presents numerical results, and Section \ref{sec:Conclusions} concludes. Appendix \ref{sec:Further-discussions} provides additional discussions. Appendix \ref{sec:Proof} contains the proofs of the main results. The Supplemental Appendix provides further discussions, including the introduction of inequality constraints on $\theta$, such as non-negativity constraints, and global convergence results for the SLC algorithm.

\section{Literature\label{sec:Literature}}

Sequential algorithms with ZJP, including the SLC, are closely related to the NFXP approach, as both approaches aim to solve the same constrained optimization problem. It is important to emphasize that the current study does not intend to claim that the NFXP approach is always inferior to sequential algorithms. Prior research has developed techniques to accelerate the convergence of the inner-loop iterations in specific models,\footnote{In the context of the BLP estimation, \citet{conlon2020best} and \citet{fukasawa2024fast} highlighted that introducing acceleration methods (e.g., SQUAREM, Spectral, Anderson acceleration) significantly improves the convergence speed of BLP contraction-mapping iterations. \citet{fukasawa2024fast} further proposed fast and simple fixed-point mappings for the inner loop of static and dynamic BLP estimation procedures. For dynamic models, \citet{bray2019strong} introduced a relative value function iteration algorithm. \citet{iskhakov2016comment} employed Newton's method to efficiently implement the inner loop of the NFXP algorithm in single-agent DDC models.} and these methods should be considered where applicable. However, such techniques are not always available or well studied, especially for non-standard structural models. The SLC algorithm proposed in the current study is, in principle, applicable to any model and is therefore worth considering as an alternative to the NFXP approach. In Appendix \ref{subsec:Comparison-of-computational-costs}, we discuss the circumstances under which the SLC is likely to outperform the NFXP in terms of computational cost.\footnote{Previous studies have highlighted the computationally inferior performance of the ``MPEC'' approach relative to NFXP in certain models. Notable examples include \citet{lee2015computationally}, \citet{pal2023comparing} for static BLP models, and \citet{sun2019computationally} for dynamic BLP models. \citet{iskhakov2016comment} demonstrated that the NFXP algorithm can be highly competitive with the MPEC approach in single-agent DDC models, provided that the inner and outer loops are appropriately designed.}

Other than the NFXP, MPEC, and sequential algorithms, two-step estimation methods have been proposed for certain standard models. Examples include \citet{hotz1993conditional} and others' Conditional Choice Probability (CCP)-based methods for DDC models, and \citet{lu2023semi} for static BLP models.\footnote{Bayesian-type methods (e.g., \citealp{imai2009bayesian}; \citealt{gallant2022constrained}) have also been proposed to mitigate computational challenges in the literature.} Although convenient to use and typically less computationally demanding, these methods generally require large sample sizes to consistently estimate first-stage parameters. In small samples, the resulting bias can be substantial. In addition, they are not applicable to structural models for which such methods have not yet been developed (e.g., dynamic BLP models, to the best of my knowledge). Note that such methods can be used as initial values in sequential algorithms with the ZJP to obtain a statistically efficient estimator.

The current study is also related to studies of sequential estimation methods that do not necessarily satisfy the ZJP (e.g., NPL for dynamic games; cf. \citealp{aguirregabiria2007sequential}).\footnote{\citet{bugni2021iterated} also proposed a sequential estimation method based on a minimum-distance criterion function and showed that the resulting estimator is efficient if an optimal sequence of weight matrices is used.} Compared with these methods, sequential algorithms with the ZJP offer advantages in terms of statistical efficiency and theoretical guarantees regarding local convergence rates.\footnote{Computing standard errors for sequential estimators that do not exhibit the ZJP can be analytically complex in practice (cf. \citealp{yamaguchi2019effects}), often requiring computationally intensive bootstrap methods. In contrast, sequential estimators with the ZJP attain the same asymptotic efficiency as the NFXP estimator, and their standard errors can be computed in the same manner as those of the NFXP estimator without resorting to bootstrap methods.}\footnote{As discussed in \citet{kasahara2012sequential}, NPL algorithm for dynamic games may not be locally convergent even in large samples. To tackle the problem, \citet{kasahara2012sequential} introduced some techniques such as the relaxation method for fixed-point iterations, and \citet{aguirregabiria2021imposing} proposed the use of the spectral algorithm.} While this paper focuses on ZJP-based algorithms, the insights may also be useful for algorithms without the ZJP, which are sometimes simpler to implement. 

The idea underlying sequential algorithms with the ZJP is closely related to Neyman orthogonality, which has been extensively studied in the recent econometrics and machine learning literature. By constructing Neyman-orthogonal scores or moment conditions for MLE or GMM estimators, we can obtain debiased estimates of $\theta$ when an appropriate consistent estimate of $Y$ is available.\footnote{\citet{chernozhukov2018double} presents a general procedure for constructing Neyman-orthogonal scores and moments. \citet{chernozhukov2022locally} and \citet{sawadogo2025efficient} propose methods for estimating DDC models without unobserved heterogeneity using Neyman orthogonality.} For sequential algorithms with the ZJP, it is straightforward to show that the gradient of the surrogate objective function used in the algorithm satisfies the Neyman orthogonality condition in the MLE setting. Even in the GMM setting, the Neyman orthogonality condition holds asymptotically asymptotically.\footnote{\citet{dearing2024efficient} showed that the EPL asymptotically satisfies the Neyman orthogonality condition for MLE. The current study further advances this discussion by demonstrating that the Neyman orthogonality holds even under non-asymptotic settings for MLE.} This suggests that the first-order condition of the optimization problem solved in the sequential algorithm is less sensitive to the value of the nuisance parameters around true parameter values. Hence, resulting parameter updates are likely to remain close to the true parameter values, even if the candidate nuisance-parameter values $\gamma_{k}$ are not very close to the true values.

Although Neyman orthogonality-based methods are promising, they require consistent initial estimates of the nuisance parameter $Y$, which may be difficult to obtain in complex structural models. In addition, constructing Neyman-orthogonal functions often involves computing second derivatives of the objective function $Q$, which may not be straightforward. In contrast, the sequential algorithms proposed here, especially the SLC algorithm, can be implemented without consistent initial estimates or explicit derivative computations. Thus, they provide a practical alternative to Neyman orthogonality-based approaches for structural estimation with equilibrium constraints.

\section{Sequential algorithms with Zero-Jacobian Property\label{sec:Sequential-algorithms}}

The current section discusses sequential algorithms based on mappings satisfying the zero-Jacobian property (ZJP). Note that the discussion in this section does not rely on statistical properties specific to structural estimation applications. Instead, we discuss key factors underlying the favorable behavior of sequential algorithms. These factors are asymptotically satisfied in structural estimation applications, as discussed in detail in Section \ref{sec:structural-estimation}. 

\subsection{Problem setting\label{subsec:problem-setting}}

Structural estimation with equilibrium constraints basically solves the following continuous constrained optimization problem:

\begin{eqnarray}
\min_{\theta\in\mathbb{R}^{n_{\theta}},Y\in\mathbb{R}^{n_{Y}}} & Q(\theta,Y)\label{eq:obj}\\
\text{s.t.} & G(Y;\theta)=0,\nonumber 
\end{eqnarray}
where $\theta$ denotes the structural parameters of the model, and $Y$ denotes the endogenous economic variables of the model. The scalars $n_{\theta}$ and $n_{Y}$ denote the dimensions of the respective variables. Here, we impose the following assumptions:

\begin{assumption}

A unique solution to the constrained minimization problem (\ref{eq:obj}) exists, and $\nabla_{Y}G\left(Y;\theta\right)$ is nonsingular at the solution $\left(\widehat{\theta},\widehat{Y}\right)$.

\label{as:existence_constrained_opt_sol} \end{assumption}

\begin{assumption}

The functions $Q$ and $G$ are twice continuously differentiable with respect to $\theta$ and $Y$.\label{as:conti_diffble} \end{assumption}

Let $\widehat{\gamma}\equiv\left(\widehat{\theta},\widehat{Y}\right)$ be the solution to the problem. Then, the following statement holds: \begin{prop} \label{prop:FOC}Under Assumptions \ref{as:existence_constrained_opt_sol} and \ref{as:conti_diffble},

$\widehat{\gamma}\equiv\left(\widehat{\theta},\widehat{Y}\right)$ satisfies $G\left(\widehat{Y};\widehat{\theta}\right)=0$ and $\nabla_{\theta}Q\left(\widehat{\theta},\widehat{Y}\right)-\left(\nabla_{Y}Q\left(\widehat{\theta},\widehat{Y}\right)\right)\left(\nabla_{Y}G\left(\widehat{Y};\widehat{\theta}\right)\right)^{-1}\left(\nabla_{\theta}G\left(\widehat{Y};\widehat{\theta}\right)\right)=0$. \end{prop} Here, we define 
\[
\Gamma_{FOC}\equiv\left\{ \gamma\equiv\left(\theta,Y\right):G\left(Y;\theta\right)=0\ \text{and}\ \nabla_{\theta}Q\left(\theta,Y\right)-\left(\nabla_{Y}Q\left(\theta,Y\right)\right)\left(\nabla_{Y}G\left(Y;\theta\right)\right)^{-1}\left(\nabla_{\theta}G\left(Y;\theta\right)\right)=0\right\} .
\]
Then, $\widehat{\gamma}\in\Gamma_{FOC}$ holds under Assumptions \ref{as:existence_constrained_opt_sol} and \ref{as:conti_diffble}.

\subsubsection*{Relationship with NFXP\label{subsubsec:Comparison-with-NFXP}}

The original constrained optimization problem $\min_{\theta,Y}Q\left(\theta,Y\right)\ \text{s.t.}\ G\left(Y;\theta\right)=0$ can be reformulated as the following problem, as discussed in \citet{su2012constrained}:

\begin{eqnarray}
\min_{\theta\in\mathbb{R}^{n_{\theta}}} & \left[\min_{Y(\theta)\ \text{s.t.}\ G(Y(\theta);\theta)=0}Q\left(\theta,Y(\theta)\right)\right]\label{eq:NFXP}
\end{eqnarray}

Let $\mathcal{Y}\left(\theta\right)\equiv\arg\min_{Y(\theta)\ \text{s.t.}\ G(Y(\theta);\theta)=0}Q\left(\theta,Y(\theta)\right)$. If we define $Q_{NFXP}\left(\theta\right)\equiv Q\left(\theta,\mathcal{Y}(\theta)\right)$,\footnote{We cannot exclude the possibility that $\mathcal{Y}\left(\theta\right)$ is empty, especially when $\theta$ is far from $\widehat{\theta}$. In this case, we cannot define $Q_{NFXP}\left(\theta\right)$ at that parameter value.} the problem can be written as the unconstrained optimization problem $\min_{\theta\in\mathbb{R}^{n_{\theta}}}Q_{NFXP}\left(\theta\right)$. Note that, in the neighborhood of $\widehat{\theta}$, the function $\mathcal{Y}\left(\theta\right)$ is of class $C^{2}$ by the implicit function theorem and Assumption \ref{as:conti_diffble}, which implies that $Q_{NFXP}\left(\theta\right)$ is also of class $C^{2}$. Here, let $\theta_{NFXP}$ be the solution to $\min_{\theta}Q_{NFXP}\left(\theta\right)$. In addition, the following statement shows the relationship between the solution to the original constrained optimization problem and that of the NFXP problem, as shown by \citet{su2012constrained}: \begin{prop} \label{prop:NFXP-MPEC}$Q\left(\widehat{\theta},\widehat{Y}\right)=Q_{NFXP}\left(\theta_{NFXP}\right)$ holds. If $Q\left(\widehat{\theta},\widehat{Y}\right)<Q\left(\theta,Y\right)\ \forall\left(\theta,Y\right)\neq\left(\widehat{\theta},\widehat{Y}\right)$ such that $G\left(Y;\theta\right)=0$, then $\theta_{NFXP}=\widehat{\theta}$. \end{prop} Proposition \ref{prop:NFXP-MPEC} and Assumption \ref{as:existence_constrained_opt_sol} imply that $\theta_{NFXP}=\widehat{\theta}$, and $\theta_{NFXP}$ is the unique minimizer of $Q_{NFXP}\left(\theta\right)$.

\subsection{Sequential algorithms based on the zero-Jacobian property}

In this study, we consider sequential algorithms that can solve the constrained optimization problem (\ref{eq:obj}). Algorithm \ref{alg:general_seq} outlines the general steps of these algorithms.

\begin{algorithm}[H]
\caption{Sequential algorithm\label{alg:general_seq}}

Set initial values $\gamma_{0}\equiv\left(\theta_{0},Y_{0}\right)$. Iterate the following until convergence $(k=0,1,2,\cdots)$:
\begin{enumerate}
\item Compute $\theta_{k+1}=\arg\min_{\theta\in\mathbb{R}^{n_{\theta}}}\widetilde{Q}\left(\theta;\gamma_{k}\right)\equiv Q\left(\theta,\Upsilon\left(\theta;\gamma_{k}\equiv\left(\theta_{k},Y_{k}\right)\right)\right)$.
\item Compute $Y_{k+1}=\Upsilon\left(\theta_{k+1};\gamma_{k}\right)$.
\end{enumerate}
\end{algorithm}

Here, $\Upsilon\left(\theta;\gamma_{k}\right)$ denotes a mapping that is a function of $\theta$ and $\gamma_{k}=(\theta_{k},Y_{k})$. For general forms of $\Upsilon\left(\theta;\gamma_{k}\right)$, the set of fixed points of Algorithm \ref{alg:general_seq} may not include the solution to (\ref{eq:obj}). Hence, we impose the following assumption:

\begin{assumption}

$\Upsilon$ is twice continuously differentiable with respect to $\theta$ and $\gamma$, and satisfies the following conditions:

(a). $\Upsilon\left(\theta;\gamma=\left(\theta,Y\right)\right)=Y\Leftrightarrow G\left(Y;\theta\right)=0$

(b). $G\left(Y;\theta\right)=0\Rightarrow\nabla_{\gamma}\Upsilon\left(\theta;\gamma=\left(\theta,Y\right)\right)=0$ (Zero-Jacobian property)

(c). $G\left(Y;\theta\right)=0\Rightarrow\nabla_{\theta}\Upsilon\left(\theta;\gamma=\left(\theta,Y\right)\right)=-\left(\nabla_{Y}G\left(Y;\theta\right)\right)^{-1}\left(\nabla_{\theta}G\left(Y;\theta\right)\right)$

\label{as:mapping_cdns} \end{assumption}

Let $\Gamma_{seq}$ be the set of fixed points of Algorithm \ref{alg:general_seq}. Then, the following statement holds: \begin{prop} \label{prop:fixed-points-KKT-seq}Suppose Assumptions \ref{as:existence_constrained_opt_sol}, \ref{as:conti_diffble}, and \ref{as:mapping_cdns} hold. Then,

(a). $\Gamma_{seq}\subset\Gamma_{FOC}$

(b). $\Gamma_{seq}=\Gamma_{FOC}$ if $\widetilde{Q}\left(\theta;\gamma\right)$ is strictly convex

(c). $\widehat{\gamma}\in\Gamma_{seq}$ if $\nabla_{\theta\theta^{\prime}}\widetilde{Q}\left(\widehat{\theta};\widehat{\gamma}\right)$ is positive definite \end{prop} 

\subsubsection*{Forms of $\Upsilon$ proposed in previous studies}

So far, the following forms of $\Upsilon$ have been proposed in the literature: 
\begin{itemize}
\item EPL (\citealp{dearing2024efficient}): $\Upsilon\left(\theta;\gamma_{k}\right)\equiv Y_{k}-\left(\nabla_{Y}G\left(Y_{k};\theta_{k}\right)\right)^{-1}\left(G\left(Y_{k};\theta\right)\right)$ 
\item ABLP (\citealp{lee2015computationally}; for static BLP estimation\footnote{Let $\delta$ be the vector of mean product utilities, and let $\Psi\left(\delta;\theta\right)\equiv\delta+\ln\left(S\right)-\ln\left(s\left(\delta,\theta\right)\right)$ be the BLP contraction mapping given nonlinear utility parameters $\theta$. Here, $S$ denotes the vector of market shares observed in the data, and $s\left(\delta,\theta\right)$ denotes the vector of market shares predicted by the structural model given $\delta$ and $\theta$. Let $Q\left(\delta\right)$ be the GMM objective function given $\delta$, and define $G\left(\delta;\theta\right)\equiv\Psi\left(\delta;\theta\right)-\delta=\ln\left(S\right)-\ln\left(s\left(\delta,\theta\right)\right)$. The ABLP algorithm proposed by \citet{lee2015computationally} corresponds to the following iterations: $\theta_{k+1}=\arg\min_{\theta}Q\left(\Upsilon\left(\theta;\delta_{k}\right)\right)$ and $\delta_{k+1}=\Upsilon\left(\theta_{k+1};\delta_{k}\right)$, where $\Upsilon\left(\theta;\delta_{k}\right)\equiv\delta_{k}-\left(\nabla_{\delta}G\left(\delta_{k};\theta\right)\right)^{-1}\left(G\left(\delta_{k};\theta\right)\right)$. Note that \citet{dearing2024efficient} also mention this mapping in Appendix C of their paper.}): $\Upsilon\left(\theta;\gamma_{k}\right)\equiv Y_{k}-\left(\nabla_{Y}G\left(Y_{k};\theta\right)\right)^{-1}\left(G\left(Y_{k};\theta\right)\right)$ 
\item NPL (\citealp{aguirregabiria2002swapping}; for single-agent DDC estimation\footnote{\citet{aguirregabiria2002swapping} introduced a fixed-point mapping $\Psi\left(P;\theta\right)$ on CCPs $P$, and showed that $\Psi$ satisfies $\nabla_{P}\Psi\left(P_{*};\theta\right)=0$ at points such that $P=\Psi\left(P;\theta\right)$. Here, $\theta$ denotes the utility parameters. \citet{aguirregabiria2002swapping} proposed sequentially iterating $\theta_{k+1}=\arg\min_{\theta}Q\left(\Psi\left(P_{k};\theta\right)\right)=-\sum_{i=1}^{N}\left.\ln\Psi\left(P_{k};\theta\right)\right|_{x_{i},a_{i}}$ and $P_{k+1}=\Psi\left(P_{k},\theta_{k+1}\right)$ for single-agent DDC models without unobserved heterogeneity.}): $\Upsilon\left(\theta;\gamma_{k}\right)\equiv\Psi\left(Y_{k};\theta\right)$, where $\Psi$ satisfies $\Psi\left(Y;\theta\right)=Y\Rightarrow\nabla_{Y}\Psi\left(Y;\theta\right)=0$. Let $G\left(Y;\theta\right)=Y-\Psi\left(Y;\theta\right)$. 
\end{itemize}
\citet{dearing2024efficient} showed in Lemma 2 of their paper that the EPL mapping $\Upsilon\left(\theta;\gamma_{k}\right)\equiv Y_{k}-\left(\nabla_{Y}G\left(Y_{k};\theta_{k}\right)\right)^{-1}\left(G\left(Y_{k};\theta\right)\right)$ satisfies the three conditions in Assumption \ref{as:mapping_cdns}. They also noted in Appendix C of their paper that the mapping $\Upsilon\left(\theta;\gamma_{k}\right)\equiv Y_{k}-\left(\nabla_{Y}G\left(Y_{k};\theta\right)\right)^{-1}\left(G\left(Y_{k};\theta\right)\right)$ satisfies these three conditions as well. We can also easily show that the NPL mapping satisfies the conditions, as demonstrated in Appendix \ref{subsec:Proof-NPL-ZJP}.

\subsection{Sequential Linearly Constrained (SLC) Algorithm\label{subsec:SLC-Algorithm}}

The current study proposes the following mapping:

\begin{eqnarray}
\Upsilon\left(\theta;\gamma_{k}\right) & \equiv & Y_{k}-\left(\left(\nabla_{Y}G\left(Y_{k};\theta_{k}\right)\right)^{-1}\left(G\left(Y_{k};\theta_{k}\right)\right)\right)\label{eq:SLC_mapping}\\
 &  & \ \ \ -\left(\left(\nabla_{Y}G\left(Y_{k};\theta_{k}\right)\right)^{-1}\left(\nabla_{\theta}G\left(Y_{k};\theta_{k}\right)\right)\right)\left(\theta-\theta_{k}\right)\nonumber 
\end{eqnarray}

As shown in Appendix \ref{subsec:Proof-NPL-ZJP}, this mapping satisfies the three conditions in Assumption \ref{as:mapping_cdns}. It can be regarded as a slight extension of the EPL mapping. In the EPL mapping $\Upsilon\left(\theta;\gamma_{k}\right)\equiv Y_{k}-\left(\nabla_{Y}G\left(Y_{k};\theta_{k}\right)\right)^{-1}\left(G\left(Y_{k};\theta\right)\right)$, the first-order Taylor approximation of $G\left(Y_{k};\theta\right)$ around $\theta_{k}$ is $G\left(Y_{k};\theta_{k}\right)+\left(\nabla_{\theta}G\left(Y_{k};\theta_{k}\right)\right)\left(\theta-\theta_{k}\right)$, and we obtain the representation (\ref{eq:SLC_mapping}) by substituting $G\left(Y_{k};\theta_{k}\right)+\left(\nabla_{\theta}G\left(Y_{k};\theta_{k}\right)\right)\left(\theta-\theta_{k}\right)$ for $G\left(Y_{k};\theta\right)$ in the EPL mapping. Note that the EPL mapping and the mapping (\ref{eq:SLC_mapping}) coincide when $G\left(Y;\theta\right)$ is linear with respect to $\theta$.

The mapping (\ref{eq:SLC_mapping}) offers a computational advantage over the EPL mapping, in that $\left(\nabla_{Y}G\left(Y_{k};\theta_{k}\right)\right)^{-1}\left(G\left(Y_{k};\theta_{k}\right)\right)$ and $\left(\nabla_{Y}G\left(Y_{k};\theta_{k}\right)\right)^{-1}\left(\nabla_{\theta}G\left(Y_{k};\theta_{k}\right)\right)$ can be precomputed prior to solving the optimization problem: 
\[
\min_{\theta\in\mathbb{R}^{n_{\theta}}}Q\left(\theta,Y_{k}-\left(\left(\nabla_{Y}G\left(Y_{k};\theta_{k}\right)\right)^{-1}\left(G\left(Y_{k};\theta_{k}\right)\right)\right)-\left(\left(\nabla_{Y}G\left(Y_{k};\theta_{k}\right)\right)^{-1}\left(\nabla_{\theta}G\left(Y_{k};\theta_{k}\right)\right)\right)\left(\theta-\theta_{k}\right)\right).
\]
In contrast, the EPL mapping in general requires repeated evaluations of $\left(\nabla_{Y}G\left(Y_{k};\theta_{k}\right)\right)^{-1}\left(G\left(Y_{k};\theta\right)\right)$ for each candidate $\theta$, which necessitates solving a system of linear equations for every candidate $\theta$.\footnote{The computational cost in the optimization step using the ABLP mapping is even higher because of the additional need to compute $\nabla_{Y}G\left(Y_{k};\theta\right)$ and solve a linear system to obtain $\left(\nabla_{Y}G\left(Y_{k};\theta\right)\right)^{-1}\left(G\left(Y_{k};\theta\right)\right)$ for each candidate value of $\theta$. In the EPL algorithm, solving a linear system to obtain $\left(\nabla_{Y}G\left(Y_{k};\theta_{k}\right)\right)^{-1}\left(G\left(Y_{k};\theta\right)\right)$ is required for each candidate $\theta$. Nevertheless, when relying on direct methods for solving linear systems, matrix factorization can be performed prior to the optimization step.}

We refer to the proposed mapping (\ref{eq:SLC_mapping}) as the SLC (Sequential Linearly Constrained) mapping, because we can easily show that Algorithm \ref{alg:general_seq} using the SLC mapping is equivalent to sequentially solving the following optimization problem with linear constraints:

\begin{eqnarray}
 & \min_{\theta\in\mathbb{R}^{n_{\theta}},Y\in\mathbb{R}^{n_{Y}}} & Q\left(\theta,Y\right)\label{eq:SLC_constrained_opt}\\
 & \text{s.t.} & \left(\nabla_{\theta}G(Y_{k};\theta_{k})\right)\left(\theta-\theta_{k}\right)+\left(\nabla_{Y}G(Y_{k};\theta_{k})\right)\left(Y-Y_{k}\right)+G(Y_{k};\theta_{k})=0\nonumber 
\end{eqnarray}

This constraint is the linearization of the original constraint $G\left(Y;\theta\right)=0$ around $(\theta_{k},Y_{k})$. This representation is convenient for comparing the SLC algorithm with Lagrangian-based algorithms, such as the SQP algorithm. Specifically, the difference between SLC and SQP lies in the objective function at each iteration: SLC uses the original objective $Q\left(\theta,Y\right)$, whereas SQP uses a quadratic approximation of the Lagrangian function. For details, see Appendix \ref{subsec:Comparison-with-Lagrangian-based}.

To ensure the algorithm is well-defined, we impose the following assumption:

\begin{assumption}

$\nabla_{Y}G\left(Y_{k};\theta_{k}\right)$ is nonsingular.

\label{as:nonsingular_Jacobian} \end{assumption}

Note that $\nabla_{Y}G\left(Y;\theta\right)$ is nonsingular at least in a neighborhood of the solution $\left(\widehat{\theta},\widehat{Y}\right)$ under Assumption \ref{as:existence_constrained_opt_sol}. In models where $G$ is derived from a contractive fixed-point constraint (e.g., single-agent DDC and static BLP random-coefficients logit models), we can theoretically establish nonsingularity at any point.\footnote{If $G\left(Y;\theta\right)=Y-\Phi\left(Y;\theta\right)$, then $\nabla_{Y}G\left(Y;\theta\right)=I-\nabla_{Y}\Phi\left(Y;\theta\right)$ holds. If $\Phi$ is a contraction mapping with respect to $Y$, the spectral radius of $\nabla_{Y}\Phi\left(Y;\theta\right)$ is strictly less than 1, and hence $\nabla_{Y}G\left(Y;\theta\right)=I-\nabla_{Y}\Phi\left(Y;\theta\right)$ is nonsingular (cf. Neumann series).} Furthermore, this assumption can be weaker than the requirement that at least one solution to $G\left(Y;\theta\right)=0$ with respect to $Y$ exists for any given $\theta$, on which the NFXP algorithm relies.\footnote{A nonsingular Jacobian assumption is similarly imposed in globally convergent algorithms for solving systems of nonlinear equations (e.g., \citealp{li2000derivative}).}

\subsubsection*{Descent direction}

In addition to lower computational costs in the optimization step among sequential algorithms satisfying the ZJP, the proposed SLC algorithm has the advantage that its updating direction is a descent direction given an appropriate merit function.

Here, we define the $\ell_{1}$-norm merit function $\phi_{1}\left(\theta,Y;\mu\right)\equiv Q\left(\theta,Y\right)+\mu\left\Vert G\left(Y;\theta\right)\right\Vert _{1}$ for a penalty parameter $\mu>0$. As shown in Theorem 17.3 of \citet{nocedal2006numerical}, $\widehat{\gamma}$ is a local minimizer of $\phi_{1}\left(\gamma;\mu\right)$ if $\mu\geq\left\Vert \widehat{\lambda}\right\Vert _{\infty}$, where $\widehat{\lambda}$ denotes the vector of Lagrange multipliers for the constrained optimization problem (\ref{eq:obj}). Concerning this merit function, the following statement holds: \begin{prop} \label{prop:descent_direction}Let $d_{k}\equiv\left(\begin{array}{c}
d_{\theta,k}\\
d_{Y,k}
\end{array}\right)\equiv\left(\begin{array}{c}
\theta_{k+1}-\theta_{k}\\
Y_{k+1}-Y_{k}
\end{array}\right)$ be the step generated by the SLC iteration. Then, the directional derivative of the merit function $\phi_{1}$ in the direction $d_{k}$ satisfies:

\[
D\left(\phi_{1}\left(\theta_{k},Y_{k};\mu\right);d_{k}\right)=\left(\nabla_{(\theta,Y)}Q\left(\theta_{k},Y_{k}\right)\right)d_{k}-\mu\left\Vert G\left(Y_{k};\theta_{k}\right)\right\Vert _{1}
\]
In addition, there exists $\underline{\mu}>0$ such that $D\left(\phi_{1}\left(\theta_{k},Y_{k};\mu\right);d_{k}\right)<0$ holds for all $\mu\geq\underline{\mu}$. \end{prop} This result implies that $\phi_{1}\left(\theta_{k}+\alpha d_{\theta,k},Y_{k}+\alpha d_{Y,k};\mu\right)<\phi_{1}\left(\theta_{k},Y_{k};\mu\right)$ for sufficiently small $\alpha>0$. This property can be utilized to establish the global convergence of the SLC algorithm by incorporating a line search step, as discussed in the Supplemental Appendix.\footnote{The proof strategy is analogous to that for the SQP algorithm (cf. \citealp{bonnans2006numerical}), which exploits the property $\left(\nabla_{\gamma}G\left(\gamma_{k}\right)\right)\left(\gamma-\gamma_{k}\right)+G\left(\gamma_{k}\right)=0$ at each iteration.}\footnote{Although this study establishes a theoretical global convergence result, its practical recommendation is to introduce a line search step only when the iteration fails to converge. The numerical results in Section \ref{sec:Numerical-experiments} suggest that the SLC algorithm converges reliably without a line search, particularly when combined with the spectral algorithm. Moreover, implementing a line search procedure requires specifying additional hyperparameters. In constrained optimization, it is well known that line search methods can impede fast local convergence due to the Maratos effect. Consequently, selecting an appropriate line search procedure and its associated hyperparameters may not be straightforward in practice.}

\subsection{Convexity of $\widetilde{Q}$}

Here, we examine the convexity of $\widetilde{Q}$ in relation to the NFXP objective function. The following lemma establishes this relationship. \begin{lem} \label{lem:difference-Hessian-NFXP}Under Assumptions \ref{as:existence_constrained_opt_sol}, \ref{as:conti_diffble}, and \ref{as:mapping_cdns}, the following holds:

\begin{eqnarray*}
\nabla_{\theta\theta^{\prime}}\widetilde{Q}\left(\widehat{\theta};\widehat{\gamma}\right)-\nabla_{\theta\theta^{\prime}}Q_{NFXP}\left(\widehat{\theta}\right) & = & \sum_{i=1}^{n_{Y}}\frac{\partial Q\left(\widehat{\theta},\widehat{Y}\right)}{\partial Y_{i}}\left(\nabla_{\theta\theta^{\prime}}\Upsilon_{i}(\widehat{\theta};\widehat{\gamma})-\nabla_{\theta\theta^{\prime}}\mathcal{Y}_{i}(\widehat{\theta})\right)
\end{eqnarray*}

\end{lem} Because $\widehat{\theta}$ is the unique minimizer of $Q_{NFXP}\left(\theta\right)$, $\nabla_{\theta\theta^{\prime}}Q_{NFXP}\left(\widehat{\theta}\right)$ is at least positive semidefinite. As discussed in Section \ref{sec:structural-estimation} regarding applications to structural estimation, $\frac{\partial Q\left(\widehat{\theta},\widehat{Y}\right)}{\partial Y}\approx0$ holds in large samples. Consequently, $\nabla_{\theta\theta^{\prime}}\widetilde{Q}\left(\widehat{\theta};\widehat{\gamma}\right)$ is close to being positive semidefinite when the sample size is large.\footnote{In general, $\nabla_{\theta\theta^{\prime}}\widetilde{Q}\left(\widehat{\theta};\widehat{\gamma}\right)$ may not be positive definite. The Supplemental Appendix investigates conditions under which the function $\widetilde{Q}\left(\theta;\gamma\right)$ is convex.}

\subsection{Local convergence speed\label{subsec:Local-convergence-speed}}

Let $H\left(\gamma\right)\equiv\left(\begin{array}{c}
H_{1}\left(\gamma\right)\\
H_{2}\left(\gamma\right)
\end{array}\right)\equiv\left(\begin{array}{c}
\widetilde{\theta}\left(\gamma\right)\\
\Upsilon\left(\widetilde{\theta}\left(\gamma\right);\gamma\right)
\end{array}\right)$, where $\widetilde{\theta}\left(\gamma\right)\equiv\arg\min_{\theta\in\Theta}\widetilde{Q}\left(\theta;\gamma\right)=Q\left(\theta,\Upsilon\left(\theta;\gamma\right)\right)$. Then, 
\begin{eqnarray*}
\nabla_{\gamma}H\left(\gamma\right) & = & \left[\begin{array}{c}
\nabla_{\gamma}\widetilde{\theta}\left(\gamma\right)\\
\left(\nabla_{\theta}\Upsilon\left(\widetilde{\theta}\left(\gamma\right),\gamma\right)\right)\left(\nabla_{\gamma}\widetilde{\theta}\left(\gamma\right)\right)+\left(\nabla_{\gamma}\Upsilon\left(\widetilde{\theta}\left(\gamma\right),\gamma\right)\right)
\end{array}\right]
\end{eqnarray*}
holds. Here, let $\widetilde{\gamma}\in\Gamma_{seq}$. Note that we allow for $\widetilde{\gamma}\neq\widehat{\gamma}\in\Gamma_{seq}$. Since $\nabla_{\theta}\widetilde{Q}\left(\widetilde{\theta}(\gamma),\gamma\right)=0$, it follows from the implicit function theorem that $\nabla_{\gamma}\widetilde{\theta}\left(\gamma\right)=-\left(\nabla_{\theta\theta^{\prime}}\widetilde{Q}\left(\widetilde{\theta}\left(\gamma\right),\gamma\right)\right)^{-1}\left(\nabla_{\theta\gamma^{\prime}}\widetilde{Q}\left(\widetilde{\theta}\left(\gamma\right),\gamma\right)\right)$, provided that the Hessian matrix $\nabla_{\theta\theta^{\prime}}\widetilde{Q}\left(\widetilde{\theta}\left(\gamma\right),\gamma\right)$ is invertible.\footnote{$\nabla_{\theta\theta^{\prime}}\widetilde{Q}\left(\widetilde{\theta}\left(\gamma\right),\gamma\right)$ is invertible if it is positive definite.} Because $\widetilde{Q}\left(\theta,\gamma\right)=Q\left(\theta,\Upsilon\left(\theta,\gamma\right)\right)$, we have $\nabla_{\gamma}\widetilde{Q}\left(\theta,\gamma\right)=\left(\nabla_{Y}Q\left(\theta,\Upsilon\left(\theta,\gamma\right)\right)\right)\left(\nabla_{\gamma}\Upsilon\left(\theta,\gamma\right)\right)$, and hence:

\begin{eqnarray}
\nabla_{\theta\gamma^{\prime}}\widetilde{Q}\left(\theta,\gamma\right) & = & \left[\nabla_{\theta Y^{\prime}}Q\left(\theta,\Upsilon\left(\theta,\gamma\right)\right)+\left(\nabla_{\theta}\Upsilon\left(\theta;\gamma\right)\right)^{\prime}\left(\nabla_{YY^{\prime}}Q\left(\theta,\Upsilon\left(\theta,\gamma\right)\right)\right)\right]\left(\nabla_{\gamma}\Upsilon\left(\theta;\gamma\right)\right)+\nonumber \\
 &  & \sum_{j=1}^{n_{Y}}\left(\nabla_{Y_{j}}Q\left(\theta,\Upsilon\left(\theta,\gamma\right)\right)\right)\cdot\left(\nabla_{\theta\gamma^{\prime}}\Upsilon_{j}\left(\theta;\gamma\right)\right)\label{eq:d2Q_d2theta_gamma}
\end{eqnarray}

By Taylor's theorem, we have $H\left(\gamma_{k-1}\right)-H\left(\widetilde{\gamma}\right)=\left(\nabla_{\gamma}H\left(\widetilde{\gamma}\right)\right)\left(\gamma_{k-1}-\widetilde{\gamma}\right)+O\left(\left\Vert \gamma_{k-1}-\widetilde{\gamma}\right\Vert ^{2}\right)$. Since $H\left(\gamma_{k-1}\right)=\gamma_{k}$ and $H\left(\widetilde{\gamma}\right)=\widetilde{\gamma}$,

\begin{equation}
\gamma_{k}-\widetilde{\gamma}=\left(\nabla_{\gamma}H\left(\widetilde{\gamma}\right)\right)\left(\gamma_{k-1}-\widetilde{\gamma}\right)+O\left(\left\Vert \gamma_{k-1}-\widetilde{\gamma}\right\Vert ^{2}\right)\label{eq:local_conv_Taylor}
\end{equation}

Because $\nabla_{\gamma}\Upsilon\left(\widetilde{\theta};\widetilde{\gamma}\right)=0$ holds under Assumption \ref{as:mapping_cdns}, (\ref{eq:d2Q_d2theta_gamma}) and $\Upsilon\left(\widetilde{\theta};\widetilde{\gamma}\right)=\widetilde{Y}$ imply: 
\begin{eqnarray}
\nabla_{\theta\gamma^{\prime}}\widetilde{Q}\left(\widetilde{\theta},\widetilde{\gamma}\right) & = & \sum_{j=1}^{n_{Y}}\left(\nabla_{Y_{j}}Q\left(\widetilde{\theta},\widetilde{Y}\right)\right)\left(\nabla_{\theta\gamma^{\prime}}\Upsilon_{j}\left(\widetilde{\theta},\widetilde{\gamma}\right)\right)\label{eq:d2Q_d2theta_gamma_at_sol}
\end{eqnarray}
Hence, we obtain the following statement: \begin{prop} \label{prop:local_conv}Under Assumptions \ref{as:existence_constrained_opt_sol}, \ref{as:conti_diffble}, \ref{as:mapping_cdns}, and \ref{as:nonsingular_Jacobian}, the following holds: 
\begin{eqnarray*}
\gamma_{k}-\widetilde{\gamma} & = & A\left(\widetilde{\gamma}\right)\left(\sum_{j=1}^{n_{Y}}\left(\nabla_{Y_{j}}Q\left(\widetilde{\theta},\widetilde{Y}\right)\right)\left(\nabla_{\theta\gamma^{\prime}}\Upsilon_{j}\left(\widetilde{\theta},\widetilde{\gamma}\right)\right)\right)\left(\gamma_{k-1}-\widetilde{\gamma}\right)+O\left(\left\Vert \gamma_{k-1}-\widetilde{\gamma}\right\Vert ^{2}\right),
\end{eqnarray*}

where $A\left(\widetilde{\gamma}\right)\equiv-\left[\begin{array}{c}
I\\
\nabla_{\theta}\Upsilon\left(\widetilde{\theta},\widetilde{\gamma}\right)
\end{array}\right]\left(\nabla_{\theta\theta^{\prime}}\widetilde{Q}\left(\widetilde{\theta},\widetilde{\gamma}\right)\right)^{-1}$. \end{prop} This proposition implies that $\gamma_{k}-\widetilde{\gamma}\approx O\left(\left\Vert \gamma_{k-1}-\widetilde{\gamma}\right\Vert ^{2}\right)$, which means that $\gamma_{k}$ converges locally to $\widetilde{\gamma}$ at a near-quadratic rate when $\nabla_{Y}Q\left(\widehat{\theta},\widehat{Y}\right)\approx0$.\footnote{Quadratic convergence is defined as follows: $\gamma_{k}$ converges to $\widetilde{\gamma}$ quadratically if $\lim_{k\rightarrow\infty}\frac{\left\Vert \gamma_{k+1}-\widetilde{\gamma}\right\Vert }{\left\Vert \gamma_{k}-\widetilde{\gamma}\right\Vert ^{2}}<\infty$.} As discussed in Section \ref{subsec:dQ_dY}, this condition holds approximately at the global minimum $\widehat{\theta}$ in large samples for structural estimation applications.

\subsection{Computational costs of the SLC algorithm}

The SLC algorithm is often computationally more attractive than the NFXP approach. Here, we clarify this point by comparing the computational costs of the SLC and NFXP algorithms when using Newton's method as the inner-loop solver and a gradient-based optimization method (e.g., BFGS) as the outer-loop algorithm. Note that Newton's method (Newton--Kantorovich iteration) is the standard approach for solving systems of nonlinear equations and is known to achieve local quadratic convergence.

In the case of the SLC algorithm, we need to compute $\left(\nabla_{Y}G\left(Y_{k};\theta_{k}\right)\right)^{-1}\left(G\left(Y_{k};\theta_{k}\right)\right)$ and $\left(\nabla_{Y}G\left(Y_{k};\theta_{k}\right)\right)^{-1}\left(G\left(Y_{k};\theta_{k}\right),\nabla_{\theta}G\left(Y_{k};\theta_{k}\right)\right)$ only once per iteration.\footnote{The computational cost of computing $\left(\nabla_{Y}G\left(Y_{k};\theta_{k}\right)\right)^{-1}\left(G\left(Y_{k};\theta_{k}\right),\nabla_{\theta}G\left(Y_{k};\theta_{k}\right)\right)$ does not differ substantially from that of computing $\left(\nabla_{Y}G\left(Y_{k};\theta_{k}\right)\right)^{-1}G\left(Y_{k};\theta_{k}\right)$ for moderate $n_{\theta}$ when relying on direct methods for solving linear systems. Because the $O(n^{3})$ cost of matrix factorization (e.g., LU decomposition) is the primary bottleneck in solving linear systems, where $n$ is the dimension of the matrix, solving for additional right-hand-side columns does not significantly increase the overall computational burden.} In contrast, in the case of the NFXP algorithm with Newton's method,\footnote{Newton's method solves $G(Y;\theta)=0$ by repeatedly applying $Y_{m+1}\leftarrow Y_{m}-\left(\nabla_{Y}G\left(Y_{m};\theta_{m}\right)\right)^{-1}\left(G\left(Y_{m};\theta_{m}\right)\right)$ until convergence, attaining local quadratic convergence in general model settings. Note that utilizing well-behaved fixed-point mappings together with acceleration methods (e.g., Anderson acceleration) can sometimes outperform Newton's method, because computing the Jacobian $\nabla_{Y}G\left(Y;\theta\right)$ in Newton's method can be costly.} we must evaluate $\left(\nabla_{Y}G\left(Y;\theta\right)\right)^{-1}\left(G\left(Y;\theta\right)\right)$ multiple times to solve for $Y$ satisfying $G\left(Y;\theta\right)=0$ within each outer-loop iteration. In addition, we must compute $\left(\nabla_{Y}G\left(Y_{k};\theta_{k}\right)\right)^{-1}\left(\nabla_{\theta}G\left(Y_{k};\theta_{k}\right)\right)$ once per outer-loop iteration to evaluate the gradient of the objective function $Q_{NFXP}\left(\theta\right)$. Consequently, for the NFXP algorithm using Newton's method in the inner loop, the repeated evaluations of $\left(\nabla_{Y}G\left(Y;\theta\right)\right)^{-1}\left(G\left(Y;\theta\right)\right)$ can result in a higher computational cost per iteration than the SLC algorithm.\footnote{I thank Victor Aguirregabiria for suggesting this insight.}

Furthermore, as demonstrated in the numerical experiments, the current study finds that the SLC algorithm often requires fewer iterations than the NFXP algorithm with first-order gradient-based methods (e.g., BFGS or interior-point methods), which further contributes to the superior computational performance of the SLC algorithm.

Note that the SLC algorithm involves the optimization step $\arg\min_{\theta\in\mathbb{R}^{n_{\theta}}}\widetilde{Q}\left(\theta;\gamma_{k}\right)$ at each iteration, which does not appear in the NFXP algorithm. However, as long as this optimization step is computationally inexpensive, the SLC algorithm achieves superior computational performance relative to the NFXP algorithm. Regarding the inner loop of the NFXP algorithm, Newton's method is sometimes outperformed by efficient fixed-point iterations combined with acceleration methods (e.g., Anderson acceleration), because computing $\nabla_{Y}G\left(Y;\theta\right)$ and solving the linear system $\left(\nabla_{Y}G\left(Y;\theta\right)\right)^{-1}\left(G\left(Y;\theta\right)\right)$ can be costly. Nevertheless, the numerical experiments in Section \ref{sec:Numerical-experiments} suggest that the SLC algorithm can still outperform the NFXP algorithm with a fixed-point-iteration-based inner loop. In Appendix \ref{subsec:Comparison-of-computational-costs}, we discuss the comparison of computational costs in greater detail. Note also that Newton's iteration is known to be unstable without regularization when initialized far from the solution, whereas the SLC algorithm is typically stable, as demonstrated in our numerical experiments. The Supplemental Appendix provides further theoretical analysis along with additional numerical experiments.

\subsection{Jacobian-free implementation}

One advantage of the SLC and EPL algorithms is that they can be implemented without explicitly computing the Jacobian of the equilibrium constraint, $\nabla_{Y}G\left(Y;\theta\right)$.\footnote{As shown in Table \ref{tab:Comparison-of-algorithms}, it is unclear whether Lagrangian-based algorithms (e.g., SQP) share this property.} In general, coding the analytical Jacobian $\nabla_{Y}G\left(Y;\theta\right)$ is error-prone, especially when the function $G\left(Y;\theta\right)$ is complex. Moreover, $\nabla_{Y}G\left(Y;\theta\right)$ is an $n_{Y}\times n_{Y}$ matrix, meaning that the memory required to store $\nabla_{Y}G\left(Y;\theta\right)$ grows quadratically with $n_{Y}$; thus, explicit Jacobian storage may be infeasible for high-dimensional problems when practitioners face memory constraints. In this subsection, we discuss the Jacobian-free implementation of the SLC and EPL algorithms, building on \citet{fukasawa2025jacobian}, who proposed a Jacobian-free implementation for the EPL algorithm.

In the SLC and EPL algorithms, we need to solve a linear system $\left(\nabla_{Y}G\left(Y;\theta\right)\right)x=b$ for $x\in\mathbb{R}^{n_{Y}}$ to compute $\left(\nabla_{Y}G\left(Y;\theta\right)\right)^{-1}b$ for a given vector $b\in\mathbb{R}^{n_{Y}}$. We can avoid explicitly forming $\nabla_{Y}G\left(Y;\theta\right)$ by combining numerical differentiation with Krylov subspace methods for solving linear systems. In Krylov subspace methods (e.g., GMRES), solving $\left(\nabla_{Y}G\left(Y;\theta\right)\right)x=b$ requires only the evaluation of Jacobian-vector products $\left(\nabla_{Y}G\left(Y;\theta\right)\right)v$ for arbitrary vectors $v\in\mathbb{R}^{n_{Y}}$. Using central finite differences, the Jacobian-vector product $\left(\nabla_{Y}G\left(Y;\theta\right)\right)v$ can be approximated by 
\[
\left(\nabla_{Y}G\left(Y;\theta\right)\right)v\approx\frac{G\left(Y+\epsilon v;\theta\right)-G\left(Y-\epsilon v;\theta\right)}{2\epsilon}
\]
for a sufficiently small step size $\epsilon>0$. Consequently, we can solve $\left(\nabla_{Y}G\left(Y;\theta\right)\right)x=b$ for $x\in\mathbb{R}^{n_{Y}}$ without explicitly computing the Jacobian $\nabla_{Y}G\left(Y;\theta\right)$. Note that Krylov subspace solvers, such as GMRES, are readily available in standard scientific computing environments (e.g., MATLAB, Julia, and Python), making implementation straightforward provided that the function $G\left(Y;\theta\right)$ is available. For further details, see \citet{fukasawa2025jacobian}, who introduced the Jacobian-free approach for the EPL algorithm.

\section{Sequential algorithm for structural estimation}\label{sec:structural-estimation}

In this section, we consider sequential algorithms for solving constrained optimization problems in the context of statistical problems, especially structural estimation in economics. In structural estimation applications, $\theta$ denotes structural parameters such as utility parameters, and $Y$ denotes economic variables, such as value functions or choice-specific value functions in DDC estimation. Let $\theta^{*}$ be the true parameter value.

In this section, we consider the setting where $Y$ is a fixed-dimensional deterministic variable. We also assume that the equilibrium constraint $G\left(Y;\theta\right)=0$ does not directly depend on the observed data.\footnote{The former setting applies to standard dynamic discrete choice models (single-agent/games). The latter setting applies to BLP models, because some variables such as unobserved product characteristics generally depend on the observed data and are stochastic.}{} Appendix \ref{subsec:Stochastic-Y} investigates the properties of sequential algorithms with ZJP allowing for such models, and clarifies what additional assumptions are required to maintain the current section's statements.

Let $Y^{*}$ be the true value of $Y$ consistent with $\theta^{*}$, which satisfies $G\left(Y^{*};\theta^{*}\right)=0$. We assume the observed data are generated based on $\left(\theta^{*},Y^{*}\right)$, and the objective function $Q\left(\theta,Y\right)$ depends on the observed data $\left\{ w_{i};i=1,\cdots,N\right\} $. Note that the current section assumes correct specification of the model, as in most studies on sequential estimation. In the Supplemental Appendix, we briefly discuss the case of model misspecification.

In this section, we consider two types of statistical specifications widely applied in the economics literature: Maximum Likelihood Estimation (MLE) and Generalized Method of Moments (GMM):\footnote{Although not explicitly discussed, Minimum Distance (MD) estimation can be treated analogously to GMM estimation.}

\subsubsection*{MLE}

\begin{eqnarray}
 & \min_{\theta,Y} & Q\left(\theta,Y\right)=-\frac{1}{N}\sum_{i=1}^{N}\ln f\left(w_{i}|\theta,Y\right)\label{eq:MLE_problem}\\
 & \text{s.t.} & G\left(Y;\theta\right)=0\nonumber 
\end{eqnarray}

We assume $w_{i}$ is generated from $f^{*}\left(w|\theta^{*},Y^{*}\right)=f^{*}\left(w\right)$. Let $Q^{*}\left(\theta,Y\right)\equiv E\left[Q\left(\theta,Y\right)\right]$.

\subsubsection*{GMM}

\begin{eqnarray}
 & \min_{\theta,Y} & Q(\theta,Y)=\left(\overline{m}\left(\theta,Y\right)\right)^{T}\widehat{W}\left(\overline{m}\left(\theta,Y\right)\right)\label{eq:GMM_problem}\\
 & \text{s.t.} & G(Y;\theta)=0\nonumber 
\end{eqnarray}
where $\overline{m}\left(\theta,Y\right)=\frac{1}{N}\sum_{i=1}^{N}m\left(w_{i};\theta,Y\right)$ and $m\left(w_{i};\theta,Y\right)$ depends on the data of sample $i$. Let $Q^{*}\left(\theta,Y\right)\equiv\left(E\left[m\left(w;\theta,Y\right)\right]\right)^{\prime}W\left(E\left[m\left(w;\theta,Y\right)\right]\right)$.

\medskip{}

Here, we impose the following assumption:

\begin{assumption}

(a). $\Theta$ and $\mathcal{Y}$ are compact and convex, and $\left(\theta^{*},Y^{*}\right)\in int\left(\Theta\times\mathcal{Y}\right)$.

(b). $\nabla_{Y}G\left(Y^{*};\theta^{*}\right)$ is nonsingular.

(c). The observations $\left\{ w_{i}:i=1,\cdots,N\right\} $ are i.i.d.

(d). $Q^{*}\left(\theta,Y\right)$ has a unique minimum in $\Theta\times\mathcal{Y}$ subject to $G\left(Y;\theta\right)=0$. The minimum occurs at $\left(\theta^{*},Y^{*}\right)$.

(e). $\widehat{\gamma}-\gamma^{*}=O_{p}\left(N^{-\frac{1}{2}}\right)$ and $\sqrt{N}\left(\widehat{\theta}-\theta^{*}\right)\rightarrow_{d}N\left(0,\Sigma_{\theta}\right)$.

(f). $\widetilde{Q}\left(\theta;\gamma\right)$ is thrice continuously differentiable with respect to $\left(\theta,\gamma\right)$.

For MLE problems, the following assumptions additionally hold:

(g-MLE): The observations $\left\{ w_{i}:i=1,\cdots,N\right\} $ are generated by a distribution $f^{*}\left(w\right)$.

(h-MLE): $f\left(\cdot|\theta,Y\right)$ is thrice continuously differentiable with respect to $\theta$ and $Y$.

(i-MLE): $Var\left(\nabla_{\theta}\ln f\left(w|\theta^{*},Y^{*}\right)\right)$ is finite.

For GMM problems, the following assumptions additionally hold:

(g-GMM): $\widehat{W}\rightarrow_{p}W$, where matrices $\widehat{W}$ and $W$ are both positive definite.

(h-GMM): $m\left(w;\theta,Y\right)$ is thrice continuously differentiable with respect to $\theta$ and $Y$.

(i-GMM): $Var\left(m\left(w;\theta,Y\right)\right)$ is finite.

(j-GMM): $\left(E\left[m\left(w;\theta,Y\right)\right]=0\ \text{and }G\left(Y;\theta\right)=0\right)$ $\Rightarrow\theta=\theta^{*},Y=Y^{*}$.

\label{as:stat_ass} \end{assumption}

\subsection{Values of $\nabla_{Y}Q\left(\widehat{\theta},\widehat{Y}\right)$ and local convergence speed}\label{subsec:dQ_dY}

As discussed in Section \ref{subsec:Local-convergence-speed}, the value of $\nabla_{Y}Q\left(\widehat{\theta},\widehat{Y}\right)$ largely affects the local convergence speed and the local convexity of the objective function. Hence, this subsection investigates the value of $\nabla_{Y}Q\left(\widehat{\theta},\widehat{Y}\right)$. The following discussion shows that $\nabla_{Y}Q\left(\widehat{\theta},\widehat{Y}\right)\rightarrow_{p}0$, and more specifically $\nabla_{Y}Q\left(\widehat{\theta},\widehat{Y}\right)=O_{p}\left(N^{-\frac{1}{2}}\right)$, holds if the model is correctly specified and we apply MLE or GMM.

\subsubsection*{Intuition}

In the large-sample environment, $\nabla_{Y}Q^{*}\left(\theta^{*},Y^{*}\right)=0$ holds. In MLE and GMM, we can choose $Q^{*}$ such that $Q^{*}\left(\theta,Y\right)\geq0\ \forall\left(\theta,Y\right)$ and $Q^{*}\left(\theta^{*},Y^{*}\right)=0$. Concerning MLE, it is well known that MLE is equivalent to minimizing the Kullback--Leibler Information Criterion (KLIC),\footnote{The log-likelihood $\int\ln f\left(w|\theta,Y\right)f^{*}\left(w\right)dw$ satisfies:

\begin{eqnarray*}
-\int\ln f\left(w|\theta,Y\right)f^{*}\left(w\right)dw & = & \int\left[\ln\frac{f^{*}\left(w\right)}{f\left(w|\theta,Y\right)}f^{*}\left(w\right)\right]dw+\int\left[\left(\ln f^{*}\left(w\right)\right)f^{*}\left(w\right)\right]dw\\
 & = & KLIC\left(f^{*};f\left(\cdot;\theta,Y\right)\right)+\int\left[\left(\ln f^{*}\left(w\right)\right)f^{*}\left(w\right)\right]dw
\end{eqnarray*}
Here, we define $KLIC\left(f^{*};f\left(\cdot;\theta,Y\right)\right)\equiv\int\left[\ln\frac{f^{*}\left(w|\theta^{*},Y^{*}\right)}{f\left(w|\theta,Y\right)}f^{*}\left(w|\theta^{*},Y^{*}\right)\right]dw$. By the property of the KLIC, $KLIC\left(f^{*};f\left(\cdot;\theta,Y\right)\right)\geq0$ holds for all $\left(\theta,Y\right)$, with equality holding when $\left(\theta,Y\right)=\left(\theta^{*},Y^{*}\right)$. Then, because the second term on the right-hand side does not depend on $(\theta,Y)$, the minimizer of the KLIC maximizes the likelihood, as long as $G\left(Y^{*};\theta^{*}\right)=0$ holds.} and the conditions hold if we alternatively choose $Q^{*}$ as the KLIC. Regarding GMM, the conditions hold because $Q^{*}\left(\theta^{*},Y^{*}\right)=\left(E\left[m\left(w;\theta^{*},Y^{*}\right)\right]\right)^{T}W\left(E\left[m\left(w;\theta^{*},Y^{*}\right)\right]\right)=0$, where $E\left[m\left(w;\theta^{*},Y^{*}\right)\right]=0$ and $W$ is a positive definite matrix. Under these conditions, we can easily show that $\nabla_{Y}Q^{*}\left(\theta^{*},Y^{*}\right)=0$ holds because of the first-order condition, as long as $G\left(Y^{*};\theta^{*}\right)=0$ holds.

In finite-sample settings, this equality may not hold, but $\nabla_{Y}Q\left(\widehat{\theta},\widehat{Y}\right)$ is close to 0 when the sample size $N$ is large.

\subsubsection*{MLE\protect\protect\footnote{\citet{aitchison1958maximum} showed that the Lagrange multiplier $\widehat{\lambda}$ associated with the constrained optimization problem (\ref{eq:MLE_problem}) satisfies $\widehat{\lambda}=O_{p}\left(N^{-\frac{1}{2}}\right)$ for MLE. Since $\nabla_{Y}Q\left(\theta,Y\right)=\left(\widehat{\lambda}\right)^{T}\left(\nabla_{Y}G\left(Y;\theta\right)\right)$ holds by Proposition \ref{prop:FOC-KKT}(a) (Appendix \ref{subsec:Comparison-with-Lagrangian-based}), $\widehat{\lambda}=O_{p}\left(N^{-\frac{1}{2}}\right)$ implies $\nabla_{Y}Q\left(\theta,Y\right)=O_{p}\left(N^{-\frac{1}{2}}\right)$. Note that the original result of \citet{aitchison1958maximum} (Theorem 2) implies that the Lagrange multiplier is $O_{p}\left(N^{\frac{1}{2}}\right)$, because their formulation does not normalize the objective function by the sample size $N$.} }

By differentiating both sides of $1=\int f\left(w|\gamma\right)dw$ with respect to $Y$ at $\gamma^{*}\equiv\left(\theta^{*},Y^{*}\right)$, we have $0=\int\left(\nabla_{Y}f\left(w|\gamma^{*}\right)\right)dw=\int\left(\nabla_{Y}\ln f\left(w|\gamma^{*}\right)\right)f\left(w|\gamma^{*}\right)dw$. Then,

\begin{eqnarray}
\nabla_{Y}Q\left(\widehat{\theta},\widehat{Y}\right) & = & -\frac{1}{N}\sum_{i=1}^{N}\nabla_{Y}\ln f\left(w_{i}|\widehat{\gamma}\right)\nonumber \\
 & = & -\left[\frac{1}{N}\sum_{i=1}^{N}\left(\nabla_{Y}\ln f\left(w_{i}|\widehat{\gamma}\right)-\nabla_{Y}\ln f\left(w_{i}|\gamma^{*}\right)\right)\right]\label{eq:dQ_MLE}\\
 &  & -\left[\frac{1}{N}\sum_{i=1}^{N}\nabla_{Y}\left(\ln f\left(w_{i}|\gamma^{*}\right)\right)-\int\left(\nabla_{Y}\ln f\left(w|\gamma^{*}\right)\right)f^{*}\left(w\right)dw\right]\nonumber \\
 &  & -\left[\int\left(\nabla_{Y}\ln f\left(w|\gamma^{*}\right)\right)\left(f^{*}\left(w\right)-f\left(w|\gamma^{*}\right)\right)dw\right]\nonumber 
\end{eqnarray}

Concerning the first term, $\nabla_{Y}\ln f\left(w_{i}|\widehat{\gamma}\right)-\nabla_{Y}\ln f\left(w_{i}|\gamma^{*}\right)=\left(\nabla_{Y\gamma^{\prime}}\ln f\left(w_{i}|\overline{\gamma}\right)\right)\left(\widehat{\gamma}-\gamma^{*}\right)$ holds by the mean value theorem, where there exists $c\in[0,1]$ such that $\overline{\gamma}=c\widehat{\gamma}+(1-c)\gamma^{*}$. Then, by $\widehat{\gamma}-\gamma^{*}=O_{p}\left(N^{-\frac{1}{2}}\right)$ and the boundedness of $\nabla_{Y\gamma^{\prime}}\ln f\left(w_{i}|\overline{\gamma}\right)$, the first term is $O_{p}\left(N^{-\frac{1}{2}}\right)$. The second term is also $O_{p}\left(N^{-\frac{1}{2}}\right)$ by the central limit theorem and the assumption that $Var\left(\nabla_{Y}\left(\ln f\left(w_{i}|\gamma^{*}\right)\right)\right)$ is finite. The third term is 0 if the distribution $f$ is correctly specified. Consequently, $\nabla_{Y}Q\left(\widehat{\theta},\widehat{Y}\right)=O_{p}\left(N^{-\frac{1}{2}}\right)$ holds if the distribution $f$ is correctly specified.

Note that $E\left[\nabla_{Y}Q\left(\theta^{*},Y^{*}\right)\right]=0$ also holds. This follows from $E\left[\frac{1}{N}\sum_{i=1}^{N}\left(\nabla_{Y}\ln f\left(w_{i}|\gamma^{*}\right)-\nabla_{Y}\ln f\left(w_{i}|\gamma^{*}\right)\right)\right]=0$ and $E\left[\frac{1}{N}\sum_{i=1}^{N}\nabla_{Y}\left(\ln f\left(w_{i}|\gamma^{*}\right)\right)-\int\left(\nabla_{Y}\ln f\left(w|\gamma^{*}\right)\right)f^{*}\left(w\right)dw\right]=0$. We can utilize this property in the discussion related to Neyman orthogonality (Section \ref{subsec:Relationship-with-Neyman}).

\subsubsection*{GMM}

\begin{eqnarray*}
\nabla_{Y}Q\left(\widehat{\theta},\widehat{Y}\right) & = & \nabla_{Y}\left[\left(\overline{m}\left(\widehat{\theta},\widehat{Y}\right)\right)^{\prime}\widehat{W}\left(\overline{m}\left(\widehat{\theta},\widehat{Y}\right)\right)\right]\\
 & = & 2\left(\nabla_{Y}\overline{m}\left(\widehat{\theta},\widehat{Y}\right)\right)^{\prime}\widehat{W}\left(\overline{m}\left(\widehat{\theta},\widehat{Y}\right)\right)
\end{eqnarray*}
holds. Here, $\overline{m}\left(\widehat{\theta},\widehat{Y}\right)=\frac{1}{N}\sum_{i=1}^{N}m\left(w_{i};\widehat{\theta},\widehat{Y}\right)=O_{p}\left(N^{-\frac{1}{2}}\right)$ holds if the model is correctly specified.\footnote{More precisely, $\frac{1}{N}\sum_{i=1}^{N}m\left(w_{i};\widehat{\gamma}\right)=\frac{1}{N}\sum_{i=1}^{N}m\left(w_{i};\gamma^{*}\right)+\frac{1}{N}\sum_{i=1}^{N}\left(\nabla_{\gamma}m\left(w_{i};\overline{\gamma}\right)\right)\left(\widehat{\gamma}-\gamma^{*}\right)$ holds for some $\overline{\gamma}$ between $\widehat{\gamma}$ and $\gamma^{*}$ by the mean value theorem. Here, $\frac{1}{N}\sum_{i=1}^{N}m\left(w_{i};\gamma^{*}\right)=O_{p}\left(N^{-\frac{1}{2}}\right)$ holds under the assumption of finite $Var\left(m\left(w;\gamma^{*}\right)\right)$ and the central limit theorem. In addition, $\widehat{\gamma}-\gamma^{*}=O_{p}\left(N^{-\frac{1}{2}}\right)$ and the boundedness of $\nabla_{\gamma}m\left(w_{i};\overline{\gamma}\right)$ imply $\frac{1}{N}\sum_{i=1}^{N}\left(\nabla_{\gamma}m\left(w_{i};\overline{\gamma}\right)\right)\left(\widehat{\gamma}-\gamma^{*}\right)=O_{p}\left(N^{-\frac{1}{2}}\right)$.} Hence, by the boundedness of $\nabla_{Y}\overline{m}\left(\theta,Y\right)$, $\nabla_{Y}Q\left(\widehat{\theta},\widehat{Y}\right)=O_{p}\left(N^{-\frac{1}{2}}\right)$ holds.\footnote{In contrast to the case of MLE, $E\left[\nabla_{Y}Q\left(\theta^{*},Y^{*}\right)\right]=0$ may not hold, because $E\left[\frac{1}{N}\sum_{i}m\left(w_{i};\theta^{*},Y^{*}\right)\right]=0$ does not necessarily imply $E\left[\nabla_{Y}Q\left(\theta^{*},Y^{*}\right)\right]=2E\left[\left(\nabla_{Y}\frac{1}{N}\sum_{i}m\left(w_{i};\theta^{*},Y^{*}\right)\right)^{\prime}\widehat{W}\left(\frac{1}{N}\sum_{i}m\left(w_{i};\theta^{*},Y^{*}\right)\right)\right]=0$.} Note that, in the exactly identified case ($n_{m}=n_{\theta}$), $\overline{m}\left(\widehat{\theta},\widehat{Y}\right)=0$ holds. This implies $\nabla_{Y}Q\left(\widehat{\theta},\widehat{Y}\right)=0$, so the local convergence speed is quadratic regardless of the sample size in this setting.

\medskip{}

\medskip{}

By $\nabla_{Y}Q\left(\widehat{\theta},\widehat{Y}\right)=O_{p}\left(N^{-\frac{1}{2}}\right)$ and Proposition \ref{prop:local_conv}, the following statement holds, which is the counterpart of the results in \citet{dearing2024efficient}.\footnote{Each statement in Proposition \ref{prop: asy_property} corresponds to Theorem 2(2), Theorem 1(3), and Theorem 2(3) of \citet{dearing2024efficient}, respectively.} 
\begin{prop}
\label{prop: asy_property}

Under Assumptions \ref{as:existence_constrained_opt_sol}, \ref{as:conti_diffble}, \ref{as:mapping_cdns}, \ref{as:nonsingular_Jacobian}, and \ref{as:stat_ass}, the following statements hold:

(1). (Asymptotic convergence rate) For all $k\geq1$, $\gamma_{k}-\widehat{\gamma}=O_{p}\left(N^{-1/2}\left\Vert \gamma_{k-1}-\widehat{\gamma}\right\Vert +\left\Vert \gamma_{k-1}-\widehat{\gamma}\right\Vert ^{2}\right)$.

(2). (Large-sample convergence) There exists a neighborhood of $\gamma^{*}$, $\mathcal{B}$, such that $\lim_{k\rightarrow\infty}\gamma_{k}=\widehat{\gamma}$ almost surely for any $\gamma_{0}\in\mathcal{B}$. In other words, $\Pr\left[\lim_{k\rightarrow\infty}\gamma_{k}=\widehat{\gamma}|\gamma_{0}\in\mathcal{B}\right]=1$.

(3). (Contraction property) W.p.a. 1 as $N\rightarrow\infty$, for any $\epsilon>0$ there exists a neighborhood of $\widehat{\gamma}$, $\mathcal{B}$, such that the iterations define a contraction mapping on $\mathcal{B}$ with Lipschitz constant $L<\epsilon$. 
\end{prop}
$\gamma_{k}-\widehat{\gamma}=O_{p}\left(N^{-1/2}\left\Vert \gamma_{k-1}-\widehat{\gamma}\right\Vert +\left\Vert \gamma_{k-1}-\widehat{\gamma}\right\Vert ^{2}\right)$ implies that $N^{-1/2}\left\Vert \gamma_{k-1}-\widehat{\gamma}\right\Vert $ is mostly negligible when the sample size $N$ is large, and $\gamma_{k}-\widehat{\gamma}\approx O_{p}\left(\left\Vert \gamma_{k-1}-\widehat{\gamma}\right\Vert ^{2}\right)$. It implies that the local convergence speed of the algorithm is nearly quadratic in large samples.

Note that this statement does not impose the assumption that the initial value $\gamma_{0}$ is statistically consistent, in contrast to \citet{kasahara2008pseudo} for NPL (single-agent DDC). This result is meaningful because consistent initial estimates of $\gamma$ are not straightforward to obtain for general models. Regarding \citet{dearing2024efficient}, they essentially showed that the asymptotic convergence rate formula holds for EPL with MLE irrespective of the availability of consistent initial estimates.\footnote{Although Algorithm 1 of \citet{dearing2024efficient} mentions that the initial value $\gamma_{0}$ should be $\sqrt{N}$-consistent, this condition is not used in the proof of the local convergence speed formula (Theorem 2(b)). I thank a referee for clarifying this point.}{} Compared to \citet{dearing2024efficient}, the current study shows that the formula holds not only for MLE but also for GMM, and applies to general ZJP-based sequential algorithms including NPL for single-agent DDC. In addition, the current study clarifies how the asymptotic property plays a role through the derivation of the non-asymptotic formula (Proposition 5).\footnote{\citet{dearing2024efficient} relied on asymptotic properties from the beginning of their proof and did not clarify how large samples improve the convergence speed. Note that the non-asymptotic formula helps identify the cases in which the algorithm works well in finite-sample environments and guides the design of better sequential algorithms.}{} It demonstrates that the asymptotic formula can be easily derived from the property $\nabla_{Y}Q\left(\widehat{\theta},\widehat{Y}\right)=O_{p}\left(N^{-\frac{1}{2}}\right)$, which generally holds for correctly specified MLE and GMM.

Besides, Proposition \ref{prop:local_conv} in Section \ref{subsec:Local-convergence-speed} indicates that $\widetilde{\gamma}\in\Gamma_{seq}$, which represents local minima of the problem, satisfies $\gamma_{k}-\widetilde{\gamma}=A\left(\widetilde{\gamma}\right)\left(\nabla_{Y}Q\left(\widetilde{\theta},\widetilde{Y}\right)\right)\left(\nabla_{\theta\gamma^{\prime}}\Upsilon\left(\widetilde{\theta},\widetilde{Y}\right)\right)\left(\gamma_{k-1}-\widetilde{\gamma}\right)+O\left(\left\Vert \gamma_{k-1}-\widetilde{\gamma}\right\Vert ^{2}\right)$. Nevertheless, the discussion above implies that $\nabla_{Y}Q\left(\widetilde{\theta},\widetilde{Y}\right)\approx0$ may not hold even in large samples if $\left(\widetilde{\theta},\widetilde{Y}\right)\neq\left(\widehat{\theta},\widehat{Y}\right)$. Consequently, the modulus of the sequential algorithm iterations near local optima other than the global optimum may not be close to 0, making the algorithm less likely to converge to those points, especially when the modulus is larger than 1. This is a unique feature of the sequential algorithms with ZJP discussed in the current paper, which SQP iterations do not possess.\footnote{SQP obtains the solution of the constrained optimization problem by essentially applying Newton's method to the first-order conditions of the constrained optimization problem. Consequently, the SQP iteration attains quadratic local convergence even at local optima. See Chapter 18 of \citet{nocedal2006numerical} for details of the SQP algorithm.} Although not discussed in previous studies applying sequential algorithms (NPL for single-agent DDC, EPL for dynamic discrete games), this feature can be thought of as an advantage of sequential algorithms over SQP.

\subsection{Asymptotic property under consistent initial values}\label{subsec:Asymptotic-property-consistent-initial-values}

In some models (e.g., DDC models), previous studies have developed methods to estimate initial $\gamma$ consistently. In cases where consistent initial estimates of $\gamma_{0}$ are available, we can derive stronger results. The following proposition states these results, which correspond to Theorem 1 of \citet{dearing2024efficient} in the context of EPL for MLE: 
\begin{prop}
\label{prop: asy_property_consistent_initial_values}

Suppose that $\gamma_{0}\equiv\left(\theta_{0},Y_{0}\right)$ is a strongly consistent initial estimate. In addition, suppose that Assumptions \ref{as:existence_constrained_opt_sol}, \ref{as:conti_diffble}, \ref{as:mapping_cdns}, \ref{as:nonsingular_Jacobian}, and \ref{as:stat_ass} hold. Then, the following statements hold:

(1). (Consistency) $\gamma_{k}\equiv\left(\theta_{k},Y_{k}\right)$ is a strongly consistent estimator of $\left(\theta^{*},Y^{*}\right)$.

(2). (Efficiency) If $\gamma_{0}-\gamma^{*}=o_{p}\left(N^{-\frac{1}{4}}\right)$, $\sqrt{N}\left(\theta_{k}-\theta^{*}\right)\rightarrow_{d}N\left(0,\Sigma_{\theta}\right)$ holds.

(3). If $\gamma_{0}-\gamma^{*}=O_{p}\left(N^{-b}\right)$ for $b\in(1/4,1/2]$, $\left\Vert \gamma_{k}-\widehat{\gamma}\right\Vert =O_{p}\left(N^{-(k-1)/2-2b}\right)$ holds. If $b=\frac{1}{2}$, $\left\Vert \gamma_{k}-\widehat{\gamma}\right\Vert =O_{p}\left(N^{-(k+1)/2}\right)$ holds. 
\end{prop}
Note that, unlike \citet{dearing2024efficient}, we only assume that the initial values $\gamma_{0}$ satisfy $\gamma_{0}-\gamma^{*}=o_{p}\left(N^{-\frac{1}{4}}\right)$,\footnote{I thank a referee for pointing out potential limitations of the proof in the previous draft when we do not impose the $\sqrt{N}$-consistency of $\gamma_{0}$, and the importance of imposing $\gamma_{0}-\gamma^{*}=o_{p}\left(N^{-\frac{1}{4}}\right)$.}{} rather than requiring $\sqrt{N}$-consistency. This weaker condition is still sufficient to obtain $\sqrt{N}\left(\theta_{k}-\theta^{*}\right)\rightarrow_{d}N\left(0,\Sigma_{\theta}\right)$. Hence, the efficiency result $\sqrt{N}\left(\theta_{k}-\theta^{*}\right)\rightarrow_{d}N\left(0,\Sigma_{\theta}\right)$ for sequential algorithms applies even when a kernel-based estimator is used to estimate $\gamma_{0}$.\footnote{Proposition 2 of \citet{kasahara2008pseudo} showed that the $\sqrt{N}$-consistency assumption imposed in \citet{aguirregabiria2002swapping} is not necessary to derive the efficiency of the NPL estimator in the single-agent DDC model. The current study generalizes this result to models beyond single-agent DDC models.}

\subsection{Relationship with Neyman orthogonality}\label{subsec:Relationship-with-Neyman}

Sequential algorithms with ZJP are closely related to Neyman orthogonality, which is recognized in recent studies (e.g., \citealp{chernozhukov2018double}) as an important property for the efficient estimation of parameters in the presence of nuisance parameters. The current section briefly clarifies this point.

First, the derivative of the objective function $\widetilde{Q}\left(\theta;\gamma\right)$ used in the sequential algorithm is $\nabla_{\theta}\widetilde{Q}\left(\theta;\gamma\right)$. By (\ref{eq:d2Q_d2theta_gamma}) and $G\left(Y;\theta\right)=0\Rightarrow\nabla_{\gamma}\Upsilon\left(\theta^{*};\gamma^{*}\right)=0$ under Assumption \ref{as:mapping_cdns}, $\left(\theta,Y\right)$ with $G\left(Y;\theta\right)=0$ satisfies $\nabla_{\gamma}\left[\nabla_{\theta}\widetilde{Q}\left(\theta,\gamma\right)\right]=\sum_{j=1}^{n_{Y}}\left(\nabla_{Y_{j}}Q\left(\theta,Y\right)\right)\left(\nabla_{\theta\gamma^{\prime}}\Upsilon_{j}\left(\theta,\gamma\right)\right)$ by (\ref{eq:d2Q_d2theta_gamma_at_sol}). Hence, $\nabla_{\gamma}\left[\nabla_{\theta}\widetilde{Q}\left(\theta^{*},\gamma^{*}\right)\right]=\sum_{j=1}^{n_{Y}}\left(\nabla_{Y_{j}}Q\left(\theta^{*},Y^{*}\right)\right)\left(\nabla_{\theta\gamma^{\prime}}\Upsilon_{j}\left(\theta^{*},\gamma^{*}\right)\right)$ holds.

As discussed in Section \ref{subsec:dQ_dY}, $E\left[\nabla_{Y}Q\left(\theta^{*},Y^{*}\right)\right]=0$ holds for MLE, which implies $\nabla_{\gamma}\left[E\left[\nabla_{\theta}\widetilde{Q}\left(\theta^{*},\gamma^{*}\right)\right]\right]=0$, i.e., $\nabla_{\gamma}\left[E\left[\nabla_{\theta}\widetilde{Q}\left(\theta^{*},\gamma^{*}\right)\right]\right]\left(\gamma-\gamma^{*}\right)=0\ \forall\gamma$.\footnote{This condition implies that $\widetilde{Q}\left(\theta,\gamma\right)$ is an orthogonal loss function (\citealp{foster2023orthogonal}).} Consequently, the score function $\nabla_{\theta}\widetilde{Q}\left(\theta;\gamma\right)$ is Neyman orthogonal with respect to the nuisance parameters $\gamma$ in the MLE setting. Note that $\nabla_{\gamma}\left[\nabla_{\theta}\widetilde{Q}\left(\theta^{*},\gamma^{*}\right)\right]\rightarrow_{p}0$ holds in the GMM setting, as discussed in Section \ref{subsec:dQ_dY}, although $\nabla_{\gamma}\left[E\left[\nabla_{\theta}\widetilde{Q}\left(\theta^{*},\gamma^{*}\right)\right]\right]=0$ may not hold in general.

Intuitively, if $\nabla_{\gamma}\left[E\left[\nabla_{\theta}\widetilde{Q}\left(\theta^{*},\gamma^{*}\right)\right]\right]=0$ or $\nabla_{\gamma}\left[\nabla_{\theta}\widetilde{Q}\left(\theta^{*},\gamma^{*}\right)\right]\rightarrow_{p}0$ holds, the first-order condition of the optimization problem $\theta_{k+1}=\arg\min_{\theta}\widetilde{Q}\left(\theta;\gamma_{k}\right)\equiv Q\left(\theta,\Upsilon\left(\theta;\gamma_{k}\right)\right)$ is less likely to be affected by the value of the nuisance parameters $\gamma_{k}$ around $\left(\theta^{*},\gamma^{*}\right)$.\footnote{Concerning SLC using $\Upsilon\left(\theta;\gamma_{k}\right)\equiv Y_{k}-\left(\left(\nabla_{Y}G\left(Y_{k};\theta_{k}\right)\right)^{-1}\left(G\left(Y_{k};\theta_{k}\right)\right)\right)-\left(\left(\nabla_{Y}G\left(Y_{k};\theta_{k}\right)\right)^{-1}\left(\nabla_{\theta}G\left(Y_{k};\theta_{k}\right)\right)\right)\left(\theta-\theta_{k}\right)$, we can treat not only $\theta$ and $Y$ but also $\left(\nabla_{Y}G\left(Y_{k};\theta_{k}\right)\right)^{-1}\left(\nabla_{\theta}G\left(Y_{k};\theta_{k}\right)\right)$ as nuisance parameters $\gamma$, and the ZJP $\nabla_{\gamma}\Upsilon\left(\theta,\gamma\right)=0$ at points satisfying $G\left(Y;\theta\right)=0$ holds even under this setting (see also the proof in Appendix \ref{subsec:Proof-SLC-ZJP}). This property suggests that locally $\theta_{k+1}$ is less likely to be affected by a slight numerical error in $\left(\nabla_{Y}G\left(Y_{k};\theta_{k}\right)\right)^{-1}\left(\nabla_{\theta}G\left(Y_{k};\theta_{k}\right)\right)$, which can be computed by iterative Krylov-based methods whose solution may not exactly coincide with the exact solution of a linear system.} Hence, we can obtain updated parameters $\theta_{k+1}$ that are close to the true value, even if the candidate nuisance parameter values $\gamma_{k}$ are not very close to the true values.\footnote{If we apply SLC or EPL, we require consistent initial values not only for $Y$ but also for $\theta$ to obtain an efficient estimate of $\theta$ after one iteration. Although this requirement may seem restrictive, studies that rely on Neyman orthogonality, such as \citet{sawadogo2025efficient} and \citet{belloni2018high}, also require consistent initial values for both the main parameter and the nuisance parameter to construct Neyman-orthogonal score functions.}

\subsection{Acceleration of iterations}\label{subsec:Acceleration}

One potential concern with sequential algorithms with ZJP is that, although they have a theoretical guarantee of near\nobreakdash-quadratic local convergence in large samples---as shown in Section \ref{subsec:dQ_dY} using $\nabla_{Y}Q\left(\widehat{\theta},\widehat{Y}\right)=O_{p}\left(N^{-\frac{1}{2}}\right)$---convergence can be slow when the pseudo\nobreakdash-objective function $\widetilde{Q}\left(\theta;\gamma\right)$ is relatively flat near the solution. Proposition \ref{prop:local_conv} indicates that the local convergence rate can be far from quadratic when the Hessian matrix $\nabla_{\theta\theta^{\prime}}\widetilde{Q}\left(\widetilde{\theta},\widetilde{\gamma}\right)$ is close to singular in a finite-sample environment. In addition, the iteration may not be contractive when far from the solution.

A practical way to mitigate these issues is to employ fixed\nobreakdash-point iteration acceleration methods (e.g., Spectral, Anderson acceleration, SQUAREM). These methods inherit the spirit of Newton\nobreakdash-type algorithms for solving nonlinear equations without requiring the computation of additional derivatives. Recent studies in economics have shown that such methods can substantially accelerate and stabilize fixed\nobreakdash-point iterations (e.g., Spectral for dynamic game NPL; SQUAREM and Anderson acceleration for static BLP NFXP\footnote{See \citet{aguirregabiria2021imposing}, \citet{conlon2020best}, and \citet{fukasawa2024fast}.}).\footnote{As discussed by \citet{pollock2021anderson} in the case of Anderson acceleration for fixed-point iterations, the use of an acceleration method sometimes leads to the convergence of non-contractive fixed-point mappings. This phenomenon is known in the computational physics literature, and \citet{pollock2021anderson} analytically investigates the mechanism behind it, motivated by previous studies.} Because sequential algorithms with ZJP repeatedly apply a fixed\nobreakdash-point mapping $H$, these acceleration techniques can be incorporated with minimal modification. For details on fixed\nobreakdash-point acceleration methods, see the Supplemental Appendix. Regarding SLC, although the original algorithm typically performs well, incorporating acceleration can further improve convergence speed, as demonstrated in Section \ref{sec:Numerical-experiments}.

\section{Numerical experiments}\label{sec:Numerical-experiments}

To demonstrate the performance of SLC, this study conducts a set of numerical experiments. We consider two models: a dynamic discrete game with time-varying unobserved heterogeneity and a dynamic demand model (dynamic BLP). \citet{dearing2024efficient} showed that EPL performs well in dynamic games without unobserved heterogeneity, where the mapping $G$ is linear in $\theta$ and EPL coincides with SLC. Accordingly, the present study focuses on settings in which this linearity does not hold.

Because explicit computation of Jacobians is often not straightforward in practice due to complicated model structures or memory constraints, the current study also presents results based on the Jacobian-free SLC method and compares them with those for NFXP algorithms using numerical derivatives.\footnote{In the dynamic BLP model, we focus on the Jacobian-free method, as deriving the Jacobian can be cumbersome due to complex formulas and memory requirements.} All experiments were run on a laptop computer with an AMD Ryzen 5 6600H 3.30 GHz CPU, 16.0 GB of RAM, Windows 11 64-bit, and MATLAB 2022b.\cprotect\footnote{Replication code for the numerical experiments is available at \url{https://github.com/takeshi-fukasawa/sequential_SLC}.}

In both models, the tolerance level for the inner loop of the NFXP algorithm is set to 1E-12. The tolerance level for the convergence criteria of both SLC and the outer loop of NFXP is set to 1E-6 unless otherwise stated.\footnote{Regarding the tolerance level in the optimization step of SLC, the current study uses 1E-7. When a looser tolerance level is employed, the SLC iterations occasionally stagnate near the solution, causing the convergence criterion with a threshold of 1E-6 to remain unmet.} Concerning numerical derivative-based (Jacobian-free) algorithms, we use central finite differences. SLC is implemented in conjunction with a Krylov-based method (GMRES) using numerical derivatives, and NFXP is implemented by computing the gradient of the objective function with respect to statistical parameters via numerical derivatives. Concerning the inner loop of NFXP, we use fixed-point iterations with Anderson acceleration. A warm-start procedure, in which the previous value of $Y$ is used as the initial value for the next iteration, is applied to NFXP to accelerate convergence. The maximum number of SLC iterations is set to 50. We perform the estimation using 20 independent datasets. To mitigate the risk of convergence to local optima, we employ five different starting values. The details of the model settings are also described in the Supplemental Appendix.

\subsection{Dynamic discrete game with time-varying unobserved heterogeneity}\label{subsec:Dynamic-discrete-game-numerical}

\subsubsection*{Settings}

The dynamic game model we consider introduces time-varying unobserved heterogeneity into the model studied in \citet{aguirregabiria2007sequential}, \citet{egesdal2015estimating}, and \citet{dearing2024efficient}. In addition, the model is similar to the one considered in \citet{arcidiacono2011conditional}.\footnote{\citet{arcidiacono2011conditional} considered a setting in which firms exit the market forever once they choose to exit. In contrast, the current study considers a setting where firms can re-enter the market even after having exited, as in \citet{aguirregabiria2007sequential}, \citet{egesdal2015estimating}, and \citet{dearing2024efficient}.} The numerical experiments are conducted by modifying the replication code of \citet{dearing2024efficient}.

Time is discrete and indexed by $t=1,2,\cdots$. In each market, there are $|\mathcal{J}|$ firms, indexed by $j\in\{1,2,\cdots,|\mathcal{J}|\}$. Given states $\omega_{t}$ and private information $\epsilon_{t}^{j}$ observed only by each firm, each firm simultaneously chooses its action $a_{t}^{j}\in\mathcal{A}=\{0,1,\cdots,|\mathcal{A}|-1\}$.

Agents maximize expected discounted utility $E\left\{ \sum_{s=t}^{\infty}\beta^{s-t}\left[\overline{u}^{j}\left(\omega_{s},a_{s}^{j},a_{s}^{-j};\theta_{u}\right)+\epsilon_{s}^{j}(a_{s}^{j})\right]\left|\omega_{t},\epsilon_{t}^{j}\right.\right\} $, where $\beta\in(0,1)$ denotes the discount factor. $\theta_{u}$ and $\theta_{f}$ represent parameters concerning utility and state transitions, respectively, and let $\theta\equiv(\theta_{u},\theta_{f})$. Under the assumptions of conditional independence, independent private values, a finite state space $\left(\omega\in\Omega=\{1,2,\cdots,|\Omega|\}\right)$, and Markov perfect equilibrium, the choice-specific value function $v^{j}\in\mathbb{R}^{|\Omega|\times|\mathcal{A}|}$ satisfies the following equation:\footnote{See Lemma 1 of \citet{dearing2024efficient} for the proof.}

\begin{eqnarray*}
v^{j}(\omega,a^{j}) & \equiv & \Phi_{v}^{j}(\omega,a^{j};v^{j},v^{-j},\theta)\\
 & = & u^{j}\left(\omega,a^{j};\Lambda^{-j}(v^{-j}),\theta_{u}\right)+\beta\sum_{\omega^{\prime}}f^{j}\left(\omega^{\prime}|\omega,a^{j};\Lambda^{-j}(v^{-j}),\theta_{f}\right)S\left(v^{j}(\omega^{\prime})\right).
\end{eqnarray*}

Here, we define a function $\Phi_{v}:\Theta\times\mathbb{R}^{|\mathcal{J}|\times|\Omega|\times|\mathcal{A}|}\rightarrow\mathbb{R}^{|\mathcal{J}|\times|\Omega|\times|\mathcal{A}|}$. $\Lambda^{j}(\omega,a^{j};v^{j})$ is the choice probability that agent $j$ chooses action $a^{j}$ in state $\omega$, conditional on having choice-specific value function $v^{j}$. The function $S(\cdot)$ is McFadden's surplus function. We also define $\Lambda^{-j}(v^{-j})\equiv\left(\Lambda^{1}(v^{1}),\cdots,\Lambda^{j-1}(v^{j-1}),\Lambda^{j+1}(v^{j+1}),\cdots,\Lambda^{|\mathcal{J}|}(v^{|\mathcal{J}|})\right)$, and:

\begin{eqnarray*}
u^{j}\left(\omega,a^{j},\Lambda^{-j}(v^{-j}),\theta_{u}\right) & \equiv & \sum_{a^{-j}\in\mathcal{A}^{|\mathcal{J}|-1}}\Lambda^{-j}\left(\omega,a^{-j};v^{-j}\right)\overline{u}^{j}\left(\omega,a^{j},a^{-j};\theta_{u}\right),\\
f^{j}\left(\omega^{\prime}|\omega,a^{j};\Lambda^{-j}(v^{-j});\theta_{f}\right) & \equiv & \sum_{a^{-j}\in\mathcal{A}^{|\mathcal{J}|-1}}\Lambda^{-j}\left(\omega,a^{-j};v^{-j}\right)f\left(\omega^{\prime}|\omega,a^{j},a^{-j};\theta_{f}\right).
\end{eqnarray*}

The choice-specific value function $v$ is a solution of the nonlinear equation $G(\theta,v)\equiv v-\Phi_{v}(\theta,v)=0$ given $\theta$.

We consider a setting where firms make entry/exit decisions $(a^{j}=1,0)$. Concerning $\overline{u}^{j}$, we assume:

\begin{eqnarray*}
\overline{u}^{j}(\omega_{t},a_{t}^{j},a_{t}^{-j};\theta_{u}) & = & \begin{cases}
\theta_{FC,j}+\theta_{RS_{1}}z_{t}^{(1)}+\theta_{RS_{2}}z_{t}^{(2)}-\theta_{RN}\ln\left(1+\sum_{l\neq j}a_{t}^{l}\right)-\theta_{EC}\left(1-a_{t-1}^{j}\right) & \text{if}\ a_{t}^{j}=1,\\
0 & \text{if}\ a_{t}^{j}=0.
\end{cases}
\end{eqnarray*}
where $\omega_{t}\equiv\left(z_{t},\left\{ a_{t-1}^{j}\right\} _{j\in\mathcal{J}}\right)$ and $z_{t}^{(1)}\in\{1,\cdots,|\mathcal{Z}^{(1)}|\}$ denotes market size, while $z_{t}^{(2)}\in\{1,\cdots,|\mathcal{Z}^{(2)}|\}$ denotes another state affecting the market.

As in \citet{dearing2024efficient}, we choose parameter values $\left|\mathcal{Z}^{(1)}\right|=5,\theta_{FC,j}=-2+0.1\times j,\theta_{RS_{1}}=\theta_{RN}=\theta_{EC}=1.0$, and $\beta=0.95$. We consider a setting in which data from $N=640$ or $160$ markets with $|\mathcal{J}|=3$ firms are available for $T=10$ consecutive periods. Regarding $\theta_{RS_{2}}$, we assume $\theta_{RS_{2}}=1.0$. Concerning $\theta_{RN}$, we assume $\theta_{RN}=4.0$.\footnote{Under $\theta_{RN}=4.0$ and without the additional state $z_{t}^{(2)}$, the NPL algorithm is unstable unless additional strategies (e.g., a spectral algorithm with a line search) are introduced, as shown in \citet{aguirregabiria2021imposing} and \citet{dearing2024efficient}. Under the current parameter setting, the present study found no evidence of multiple equilibria.} Regarding the states $z_{t}^{(1)}$ and $z_{t}^{(2)}$, we assume the state transition probability matrix $\left(\begin{array}{ccccc}
\pi^{(i)} & 1-\pi^{(i)} & 0 & \cdots & 0\\
1-\frac{\pi^{(i)}}{2} & \pi^{(i)} & 1-\frac{\pi^{(i)}}{2} & \cdots & 0\\
0 & 1-\frac{\pi^{(i)}}{2} & \pi^{(i)} & \cdots & 0\\
\vdots & \vdots & \vdots & \vdots & \vdots\\
0 & 0 & 0 & 1-\pi^{(i)} & \pi^{(i)}
\end{array}\right)$ $(i=1,2)$, where $\pi^{(1)}=\pi^{(2)}=0.8$. In the estimation, we assume that the values of $\pi^{(1)}$ and $\beta$ are known.

We assume that the state $z_{t}^{(2)}$ is observed by the firms but unobserved by econometricians. Concerning $\left|\mathcal{Z}^{(2)}\right|$, we let $\left|\mathcal{Z}^{(2)}\right|=3$. We estimate the utility parameters $\theta$ and the state transition parameter $\pi^{(2)}$. In both the SLC and NFXP algorithms, we treat $\left(v,\overline{p_{1}}\right)$ as the nuisance parameters $Y$, where $\overline{p_{1}}$ denotes the vector of the initial distribution of $\omega$. We assume that the initial period is in a stationary environment and satisfies $\overline{p_{1}}=\Phi_{\overline{p_{1}}}\left(\overline{p_{1}}\right)=\Pi\overline{p_{1}}$, where $\Pi$ denotes the transition probability matrix of $\omega$ in the initial period. Concerning NFXP, we solve for $\left(v,\overline{p_{1}}\right)$ by applying the fixed-point mappings $\Phi_{v}$ and $\Phi_{\overline{p_{1}}}$ with Anderson acceleration.\footnote{The Supplemental Appendix additionally shows numerical results for SLC and NFXP when the idea of Endogenous Value Function Iteration (EVFI), which can accelerate convergence, is applied to these algorithms.} To uniquely identify the parameters, we impose an inequality constraint $\theta_{RS_{2}}\geq0$ in the estimation. We also impose $\pi\in[0,1]$. Note that we can easily introduce inequality constraints on the structural parameters, as discussed in the Supplemental Appendix. The parameters $\left(\theta_{u},\pi^{(2)}\right)$ are estimated by MLE given the constraints on $\left(v,\overline{p_{1}}\right)$. The current study applies the interior-point algorithm in the \texttt{fmincon} function of MATLAB for the outer loop of NFXP.\footnote{The present study also experimented with the Berndt--Hall--Hall--Hausman (BHHH) algorithm (\citealp{hall1974estimation}) as an outer\nobreakdash-loop optimization method for NFXP, whose effectiveness is noted in \citet{iskhakov2016comment}. However, in the current setting, the algorithm occasionally moved to parameter values for which the inner\nobreakdash-loop iterations failed to converge. The current study also experimented with the commercial solver KNITRO, but it did not yield substantial speed improvements.} Supplemental Appendix \ref{subsec:Dynamic-game-details} describes the procedure for computing the likelihood; see that section for further details.

\subsubsection*{Results}

Table \ref{tab:Results-dynamic-game-unobs-hetero-comp} shows the computational performance, and Table \ref{tab:Results-dynamic-game-unobs-hetero-est} shows the estimated parameters. Table \ref{tab:Results-dynamic-game-unobs-hetero-comp} suggests that SLC is roughly 3--4 times faster than NFXP, regardless of the use of analytical derivatives and the number of independent markets $N$. This speedup mainly comes from the smaller number of main (outer-loop) iterations. Note that the last column of Table \ref{tab:Results-dynamic-game-unobs-hetero-est} indicates that the differences between the SLC and NFXP estimator values are mostly negligible.

Furthermore, SLC converges in most cases even in the absence of stabilization strategies, suggesting that SLC is practical as a solution method for estimating dynamic games even with time-varying unobserved heterogeneity.\footnote{\citet{dearing2024efficient} showed numerical results where the convergence of EPL, which is equivalent to SLC under the linearity of $G$ with respect to $\theta$, is stable even when starting from random initial values in a dynamic game model without unobserved heterogeneity.}

\begin{table}[H]
\caption{Results of numerical experiments (Dynamic game with time-varying unobserved heterogeneity; Computational performance)}\label{tab:Results-dynamic-game-unobs-hetero-comp}

\begin{centering}
{\scriptsize{}}{\scriptsize{}%
\begin{tabular}{cccccccccc}
\hline 
\multirow{1}{*}{{\footnotesize$N$}} & \multirow{1}{*}{{\footnotesize Analy. diff.}} & \multirow{1}{*}{{\footnotesize Algorithm}} & \multicolumn{2}{c}{{\footnotesize Comp. Time (sec)}} & {\footnotesize Feval $Q$}{\scriptsize{} } & {\footnotesize Feval $G$}{\scriptsize{} } & {\footnotesize\# of}{\scriptsize{} } & {\footnotesize Data sets}{\scriptsize{} } & {\footnotesize Runs}\tabularnewline
\hline 
 &  &  & {\footnotesize Mean}{\scriptsize{} } & {\footnotesize Std.}{\scriptsize{} } & {\footnotesize (mean)}{\scriptsize{} } & {\footnotesize (mean)}{\scriptsize{} } & {\footnotesize Main iter.}{\scriptsize{} } & {\footnotesize Conv. (\%)}{\scriptsize{} } & {\footnotesize Conv. (\%)}\tabularnewline
\hline 
\multirow{4}{*}{{\footnotesize 640}} & \multirow{2}{*}{{\footnotesize Yes}} & {\footnotesize SLC}{\scriptsize{} } & {\footnotesize 8.3}{\scriptsize{} } & {\footnotesize 1}{\scriptsize{} } & {\footnotesize 311.9}{\scriptsize{} } & {\footnotesize 9.3}{\scriptsize{} } & {\footnotesize 9.3}{\scriptsize{} } & {\footnotesize 100}{\scriptsize{} } & {\footnotesize 100}\tabularnewline
 &  & {\footnotesize NFXP}{\scriptsize{} } & {\footnotesize 26.1}{\scriptsize{} } & {\footnotesize 1.9}{\scriptsize{} } & {\footnotesize 42.8}{\scriptsize{} } & {\footnotesize 3999.3}{\scriptsize{} } & {\footnotesize 33.4}{\scriptsize{} } & {\footnotesize 100}{\scriptsize{} } & {\footnotesize 100}\tabularnewline
 & \multirow{2}{*}{{\footnotesize No}} & {\footnotesize SLC}{\scriptsize{} } & {\footnotesize 41.1}{\scriptsize{} } & {\footnotesize 5}{\scriptsize{} } & {\footnotesize 3817.7}{\scriptsize{} } & {\footnotesize 5492.3}{\scriptsize{} } & {\footnotesize 9.2}{\scriptsize{} } & {\footnotesize 100}{\scriptsize{} } & {\footnotesize 100}\tabularnewline
 &  & {\footnotesize NFXP}{\scriptsize{} } & {\footnotesize 164.3}{\scriptsize{} } & {\footnotesize 8.1}{\scriptsize{} } & {\footnotesize 10074.2}{\scriptsize{} } & {\footnotesize 40779.9}{\scriptsize{} } & {\footnotesize 33.4}{\scriptsize{} } & {\footnotesize 100}{\scriptsize{} } & {\footnotesize 100}\tabularnewline
\hline 
\multirow{2}{*}{{\footnotesize 160}} & \multirow{2}{*}{{\footnotesize Yes}} & {\footnotesize SLC}{\scriptsize{} } & {\footnotesize 6}{\scriptsize{} } & {\footnotesize 1.4}{\scriptsize{} } & {\footnotesize 423.6}{\scriptsize{} } & {\footnotesize 11.5}{\scriptsize{} } & {\footnotesize 11.5}{\scriptsize{} } & {\footnotesize 100}{\scriptsize{} } & {\footnotesize 100}\tabularnewline
 &  & {\footnotesize NFXP}{\scriptsize{} } & {\footnotesize 26.8}{\scriptsize{} } & {\footnotesize 2.9}{\scriptsize{} } & {\footnotesize 44.2}{\scriptsize{} } & {\footnotesize 4210.2}{\scriptsize{} } & {\footnotesize 34.2}{\scriptsize{} } & {\footnotesize 100}{\scriptsize{} } & {\footnotesize 100}\tabularnewline
\hline 
\end{tabular}}{\scriptsize\par}
\par\end{centering}
\begin{centering}
 
\par\end{centering}
{\footnotesize Notes. Based on 20 randomly generated datasets with five random initial starting points. ``Feval'' denotes the mean number of function evaluations.}{\footnotesize\par}

\end{table}

\begin{center}
{\scriptsize{} }{\scriptsize{}
\begin{table}[H]
{\scriptsize\caption{Results of numerical experiments (Dynamic game with time-varying unobserved heterogeneity; Parameter estimates)}\label{tab:Results-dynamic-game-unobs-hetero-est}
}{\scriptsize\par}

{\scriptsize{} }{\scriptsize\par}
\begin{centering}
{\scriptsize{}}{\scriptsize{}%
\begin{tabular}{ccccccccccccc}
\hline 
\multirow{1}{*}{{\scriptsize$N$}} & \multirow{1}{*}{{\scriptsize Analy. diff.}} & \multirow{1}{*}{{\scriptsize Algorithm}} &  & {\scriptsize$\theta_{FC,j=1}$ } & {\scriptsize$\theta_{FC,j=2}$ } & {\scriptsize$\theta_{FC,j=3}$ } & {\scriptsize$\theta_{RS_{1}}$ } & {\scriptsize$\theta_{RS_{2}}$ } & {\scriptsize$\theta_{RN}$ } & {\scriptsize$\theta_{EC}$ } & {\scriptsize$\pi^{(2)}$ } & {\scriptsize Diff. with}\tabularnewline
\hline 
 &  &  & {\scriptsize True } & {\scriptsize -1.9 } & {\scriptsize -1.8 } & {\scriptsize -1.7 } & {\scriptsize 1.0 } & {\scriptsize 1.0 } & {\scriptsize 4.0 } & {\scriptsize 1.0 } & {\scriptsize 0.8 } & {\scriptsize SLC (Analy.)}\tabularnewline
\hline 
\multirow{8}{*}{{\scriptsize 640}} & \multirow{4}{*}{{\scriptsize Yes}} & \multirow{2}{*}{{\scriptsize SLC}} & {\scriptsize Mean } & {\scriptsize -1.995 } & {\scriptsize -1.895 } & {\scriptsize -1.794 } & {\scriptsize 1.026 } & {\scriptsize 1.028 } & {\scriptsize 4.018 } & {\scriptsize 0.992 } & {\scriptsize 0.751 } & \multirow{2}{*}{{\scriptsize -}}\tabularnewline
 &  &  & {\scriptsize Std. } & {\scriptsize 0.266 } & {\scriptsize 0.269 } & {\scriptsize 0.288 } & {\scriptsize 0.046 } & {\scriptsize 0.121 } & {\scriptsize 0.173 } & {\scriptsize 0.036 } & {\scriptsize 0.052 } & \tabularnewline
 &  & \multirow{2}{*}{{\scriptsize NFXP}} & {\scriptsize Mean } & {\scriptsize -1.995 } & {\scriptsize -1.895 } & {\scriptsize -1.794 } & {\scriptsize 1.026 } & {\scriptsize 1.028 } & {\scriptsize 4.018 } & {\scriptsize 0.992 } & {\scriptsize 0.751 } & \multirow{2}{*}{{\scriptsize 2.2E-06}}\tabularnewline
 &  &  & {\scriptsize Std. } & {\scriptsize 0.266 } & {\scriptsize 0.269 } & {\scriptsize 0.288 } & {\scriptsize 0.046 } & {\scriptsize 0.121 } & {\scriptsize 0.173 } & {\scriptsize 0.036 } & {\scriptsize 0.052 } & \tabularnewline
\cline{2-13}
 & \multirow{4}{*}{{\scriptsize No}} & \multirow{2}{*}{{\scriptsize SLC}} & {\scriptsize Mean } & {\scriptsize -1.995 } & {\scriptsize -1.895 } & {\scriptsize -1.794 } & {\scriptsize 1.026 } & {\scriptsize 1.028 } & {\scriptsize 4.018 } & {\scriptsize 0.992 } & {\scriptsize 0.751 } & \multirow{2}{*}{{\scriptsize 1.7E-06}}\tabularnewline
 &  &  & {\scriptsize Std. } & {\scriptsize 0.266 } & {\scriptsize 0.269 } & {\scriptsize 0.288 } & {\scriptsize 0.046 } & {\scriptsize 0.121 } & {\scriptsize 0.173 } & {\scriptsize 0.036 } & {\scriptsize 0.052 } & \tabularnewline
 &  & \multirow{2}{*}{{\scriptsize NFXP}} & {\scriptsize Mean } & {\scriptsize -1.995 } & {\scriptsize -1.895 } & {\scriptsize -1.794 } & {\scriptsize 1.026 } & {\scriptsize 1.028 } & {\scriptsize 4.018 } & {\scriptsize 0.992 } & {\scriptsize 0.751 } & \multirow{2}{*}{{\scriptsize 2.2E-06}}\tabularnewline
 &  &  & {\scriptsize Std. } & {\scriptsize 0.266 } & {\scriptsize 0.269 } & {\scriptsize 0.288 } & {\scriptsize 0.046 } & {\scriptsize 0.121 } & {\scriptsize 0.173 } & {\scriptsize 0.036 } & {\scriptsize 0.052 } & \tabularnewline
\hline 
\multirow{4}{*}{{\scriptsize 160}} & \multirow{4}{*}{{\scriptsize Yes}} & \multirow{2}{*}{{\scriptsize SLC}} & {\scriptsize Mean } & {\scriptsize -1.935 } & {\scriptsize -1.834 } & {\scriptsize -1.729 } & {\scriptsize 1.012 } & {\scriptsize 0.985 } & {\scriptsize 3.932 } & {\scriptsize 1.036 } & {\scriptsize 0.752 } & \multirow{2}{*}{{\scriptsize -}}\tabularnewline
 &  &  & {\scriptsize Std. } & {\scriptsize 0.528 } & {\scriptsize 0.525 } & {\scriptsize 0.554 } & {\scriptsize 0.086 } & {\scriptsize 0.258 } & {\scriptsize 0.316 } & {\scriptsize 0.067 } & {\scriptsize 0.137 } & \tabularnewline
 &  & \multirow{2}{*}{{\scriptsize NFXP}} & {\scriptsize Mean } & {\scriptsize -1.935 } & {\scriptsize -1.834 } & {\scriptsize -1.729 } & {\scriptsize 1.012 } & {\scriptsize 0.985 } & {\scriptsize 3.932 } & {\scriptsize 1.036 } & {\scriptsize 0.752 } & \multirow{2}{*}{{\scriptsize 1.1E-05}}\tabularnewline
 &  &  & {\scriptsize Std. } & {\scriptsize 0.528 } & {\scriptsize 0.525 } & {\scriptsize 0.554 } & {\scriptsize 0.086 } & {\scriptsize 0.258 } & {\scriptsize 0.316 } & {\scriptsize 0.067 } & {\scriptsize 0.137 } & \tabularnewline
\hline 
\end{tabular}}{\scriptsize\par}
\par\end{centering}
\begin{centering}
{\scriptsize{} }{\scriptsize\par}
\par\end{centering}
{\footnotesize Notes. Based on 20 randomly generated datasets with five random initial starting points. The column ``Diff. with SLC (Analy.)'' reports the sup\nobreakdash-norm distance between each estimator and the SLC estimator computed using analytical derivatives. The maximum number of iterations is set to 50.}{\footnotesize\par}

{\scriptsize{} }{\scriptsize\par}
\end{table}
}{\scriptsize\par}
\par\end{center}

\subsection{Dynamic BLP model}\label{subsec:Dynamic-BLP-numerical}

\subsubsection*{Settings}

We consider a dynamic BLP model where forward-looking decision-making of consumers is explicitly incorporated into the standard static BLP model (\citealp{berry1995automobile}). Here, we focus on the model of perfectly durable goods considered in \citet{sun2019computationally}.

Consumer $i$'s utility when purchasing product $j$ at time $t$ is:

\[
U_{ijt}=\chi_{jt}\theta_{i}+\xi_{jt}+\epsilon_{ijt}
\]
where $\chi_{jt}\equiv\left(1,\chi_{1jt},\chi_{2jt},\chi_{3jt},p_{jt}\right)$ and $\xi_{jt}$ denote the observed and unobserved product characteristics of product $j$ at time $t$.\footnote{Here, $p_{jt}$ denotes the product price.} The coefficient $\theta_{i}\sim N(\overline{\theta},\Sigma)$ represents the random coefficients of consumer $i$.

Consumer $i$'s utility when purchasing nothing is:

\[
U_{i0t}=\beta E_{t}[V_{it+1}(\Omega_{t+1})|\Omega_{t}]+\epsilon_{i0t},
\]
where $\beta$ denotes the consumers' discount factor, and $\epsilon$ denotes idiosyncratic utility shocks following a Gumbel distribution. The variable $V_{it}$ denotes consumer $i$'s (integrated) value function at time $t$, and $\Omega$ denotes state variables.

We estimate the demand parameters $\theta\equiv\left(\overline{\theta},\Sigma\right)$ by GMM by imposing the moment conditions $E[\xi|Z]=0$, where $Z$ denotes instrumental variables. We introduce consumer heterogeneity in preferences for the characteristics $\chi_{1jt},\chi_{2jt},p_{jt}$, and let $\sigma_{\chi_{1}},\sigma_{\chi_{2}},\sigma_{p}$ denote the standard deviation parameters of the normal distribution governing the corresponding random coefficients. We further assume that $S_{jt}^{(data)}=s_{jt}\left(\xi,V,\theta\right)\ \forall j,t$, namely, observed market shares are equal to predicted market shares. Concerning consumer expectations, we assume inclusive value sufficiency, as in previous studies (e.g., \citealp{gowrisankaran2012dynamics}). We use 50 grid points to interpolate the inclusive values. We consider a setting where $J=25$ products exist in a market, and market data are available for $T=25$ consecutive periods. For details of the model setting, see Section 5.2 of \citet{sun2019computationally} or \citet{fukasawa2024fast}. The Supplemental Appendix also describes these details.

For the inner loop of NFXP, we use the computationally efficient fixed-point mapping with the Anderson acceleration method developed in \citet{fukasawa2024fast}. \citet{fukasawa2024fast} shows that there exists a closed-form fixed-point mapping of the value functions $V$, $\Phi_{V}$, which does not depend on the mean product utilities $\delta$: $V=\Phi_{V}\left(V;\theta\right)$. The paper shows that the new inner-loop algorithm with Anderson acceleration is roughly ten times faster than the traditional fixed-point iteration algorithm that combines the BLP contraction mapping and value function iterations. It also shows that there exists a closed-form mapping from the space of $V$ to $\delta$, and we use this mapping to construct the objective function $Q$ as a function of $V$ and the nonlinear parameters in SLC.\footnote{As in the standard BLP estimation procedure, the linear parameters $\overline{\theta}$ are concentrated out in SLC.}{} As the constraint, we use $V-\Phi_{V}\left(V;\theta\right)=0$ in SLC.

In the present application, we rely on the numerical derivative-based Jacobian-free approach. In dynamic BLP models with random coefficients, deriving the Jacobian of function $G$ is not straightforward due to the complexity of the model structure. In addition, the Jacobian matrix is typically high-dimensional, implying that a substantial amount of computer memory is required merely to compute and store it.\footnote{In the current setting where the number of discretized consumer types is 50, the Jacobian is a $3750\times3750$ matrix, which consumes roughly 107 MB of memory. If we allow for more flexible consumer heterogeneity and let the number of discretized consumer types be 1,000, which is typical in static BLP applications, the Jacobian is a $75000\times75000$ matrix and requires roughly 43 GB of memory. Although the present model includes only two consumer-level state variables (owning a durable good or not), the dimensionality would increase further as additional states are introduced. It is worth noting that, for the MPEC approach (interior-point methods that explicitly compute Hessians), \citet{sun2019computationally} reported in Appendix C of their paper that more than 32 GB of memory is required even when the number of discretized consumer types is only 50.}

\subsubsection*{Results}

Table \ref{tab:Results-dynamic-BLP-comp} shows the computational performance, and Table \ref{tab:Results-dynamic-BLP-est} shows the results on estimated parameter values. The results suggest that the original SLC usually performs well and is roughly 3 times faster than NFXP. This large speedup mainly arises from the much smaller number of evaluations of the function $G$: the number of $G$ evaluations required by SLC is roughly five times smaller than that required by NFXP. Although SLC requires a larger number of evaluations of the objective function $Q$, code profiling indicates that the computational cost of evaluating $Q$ is roughly 10 times smaller than that of evaluating $G$. Consequently, the substantially fewer evaluations of $G$ dominate the overall computational savings. Note that the SLC and NFXP estimators take mostly the same values, as shown in the latter table.

It is worth noting that the original SLC is occasionally slow and sometimes does not converge within 50 iterations.\footnote{Occasional non\nobreakdash-convergence is also not uncommon under the ``MPEC'' approach (\citealp{su2012constrained}; \citealp{egesdal2015estimating}).} The table additionally reports the results for SLC when the tolerance level is set to 1E-4. The results show a higher convergence rate than under the baseline setting with a tolerance level of 1E-6, suggesting that much of the slow convergence observed in SLC can be attributed to its slow convergence behavior near the solution. As discussed in Section \ref{subsec:Acceleration}, sequential algorithms with ZJP may require many iterations when the objective function becomes nearly flat around the solution. In addition, local convergence in large samples does not necessarily imply convergence from distant initial points. To mitigate these issues, the current study also experimented with SLC combined with the spectral algorithm, which is a fixed-point acceleration method utilized in the NPL algorithm for dynamic games (\citealp{aguirregabiria2021imposing}). The results suggest that the spectral algorithm significantly mitigates the problem and enhances the robustness of SLC---SLC-Spectral is roughly 4 times faster than NFXP and converges reliably.\footnote{In the current setting, the SLC-Spectral and NFXP estimators take different values in 4 out of the 20 randomly generated datasets, although the maximum absolute difference is less than 0.08. In the current numerical experiment, the standard deviation parameters of the random coefficients are not restricted to be positive, which may increase the number of potential local minima. The Supplemental Appendix further discusses this issue.}

\begin{table}[H]
\caption{Results of numerical experiments (Dynamic BLP; Computational performance)}\label{tab:Results-dynamic-BLP-comp}

\begin{centering}
{\scriptsize{}}{\scriptsize{}%
\begin{tabular}{cccccccc}
\hline 
\multirow{1}{*}{{\footnotesize Algorithm}} & \multicolumn{2}{c}{{\footnotesize Comp. Time (sec)}} & {\footnotesize Feval $Q$}{\scriptsize{} } & {\footnotesize Feval $G$}{\scriptsize{} } & {\footnotesize\# of}{\scriptsize{} } & {\footnotesize Data sets}{\scriptsize{} } & {\footnotesize Runs}\tabularnewline
\hline 
 & {\footnotesize Mean}{\scriptsize{} } & {\footnotesize Std.}{\scriptsize{} } & {\footnotesize (mean)}{\scriptsize{} } & {\footnotesize (mean)}{\scriptsize{} } & {\footnotesize Main iter.}{\scriptsize{} } & {\footnotesize Conv. (\%)}{\scriptsize{} } & {\footnotesize Conv. (\%)}\tabularnewline
\hline 
{\footnotesize SLC(TOL: 1E-6)}{\scriptsize{} } & {\footnotesize 229.2}{\scriptsize{} } & {\footnotesize 154.1}{\scriptsize{} } & {\footnotesize 1658.8}{\scriptsize{} } & {\footnotesize 5382.7}{\scriptsize{} } & {\footnotesize 17.6}{\scriptsize{} } & {\footnotesize 90}{\scriptsize{} } & {\footnotesize 84}\tabularnewline
{\footnotesize SLC(TOL: 1E-4)}{\scriptsize{} } & {\footnotesize 211.5}{\scriptsize{} } & {\footnotesize 154.7}{\scriptsize{} } & {\footnotesize 1401.6}{\scriptsize{} } & {\footnotesize 4934.4}{\scriptsize{} } & {\footnotesize 14.6}{\scriptsize{} } & {\footnotesize 100}{\scriptsize{} } & {\footnotesize 89}\tabularnewline
{\footnotesize SLC-Spectral}{\scriptsize{} } & {\footnotesize 160.3}{\scriptsize{} } & {\footnotesize 66.8}{\scriptsize{} } & {\footnotesize 1106.3}{\scriptsize{} } & {\footnotesize 3698.5}{\scriptsize{} } & {\footnotesize 11.5}{\scriptsize{} } & {\footnotesize 100}{\scriptsize{} } & {\footnotesize 99}\tabularnewline
{\footnotesize NFXP}{\scriptsize{} } & {\footnotesize 588}{\scriptsize{} } & {\footnotesize 159.1}{\scriptsize{} } & {\footnotesize 816.9}{\scriptsize{} } & {\footnotesize 15785.6}{\scriptsize{} } & {\footnotesize 14.6}{\scriptsize{} } & {\footnotesize 100}{\scriptsize{} } & {\footnotesize 100}\tabularnewline
\hline 
\end{tabular}}{\scriptsize\par}
\par\end{centering}
\begin{centering}
 
\par\end{centering}
{\footnotesize Notes. Based on 20 randomly generated datasets with five random initial starting points. ``SLC--Spectral'' denotes the SLC algorithm combined with the spectral method as a fixed\nobreakdash-point iteration acceleration method. The NFXP algorithm in the table utilizes the inner-loop algorithm proposed in \citet{fukasawa2024fast}, which is roughly 10 times faster than the traditional NFXP algorithm applied in previous studies. The maximum number of iterations is set to 50. SLC-Spectral converged in the previously non-convergent case when the maximum number of iterations was increased to 60.}{\footnotesize\par}

\end{table}

\begin{center}
{\scriptsize{} }{\scriptsize{}
\begin{table}[H]
{\scriptsize\caption{Results of numerical experiments (Dynamic BLP; Parameter estimates)}\label{tab:Results-dynamic-BLP-est}
}{\scriptsize\par}

{\scriptsize{} }{\scriptsize\par}
\begin{centering}
{\scriptsize{}}{\scriptsize{}%
\begin{tabular}{ccccc}
\hline 
\multirow{1}{*}{{\footnotesize Algorithm}} &  & {\footnotesize$\left|\sigma_{\chi_{1}}\right|$}{\scriptsize{} } & {\footnotesize$\left|\sigma_{\chi_{2}}\right|$}{\scriptsize{} } & {\footnotesize$\left|\sigma_{p}\right|$}\tabularnewline
\hline 
 & {\footnotesize True}{\scriptsize{} } & {\footnotesize 0.5}{\scriptsize{} } & {\footnotesize 0.5}{\scriptsize{} } & {\footnotesize 0.25}\tabularnewline
\hline 
\multirow{2}{*}{{\footnotesize SLC(TOL: 1E-6)}} & {\footnotesize Mean}{\scriptsize{} } & {\footnotesize 0.552}{\scriptsize{} } & {\footnotesize 0.505}{\scriptsize{} } & {\footnotesize 0.175}\tabularnewline
 & {\footnotesize Std.}{\scriptsize{} } & {\footnotesize 0.207}{\scriptsize{} } & {\footnotesize 0.212}{\scriptsize{} } & {\footnotesize 0.119}\tabularnewline
\multirow{2}{*}{{\footnotesize SLC-Spectral}} & {\footnotesize Mean}{\scriptsize{} } & {\footnotesize 0.553}{\scriptsize{} } & {\footnotesize 0.506}{\scriptsize{} } & {\footnotesize 0.175}\tabularnewline
 & {\footnotesize Std.}{\scriptsize{} } & {\footnotesize 0.209}{\scriptsize{} } & {\footnotesize 0.212}{\scriptsize{} } & {\footnotesize 0.117}\tabularnewline
\multirow{2}{*}{{\footnotesize NFXP}} & {\footnotesize Mean}{\scriptsize{} } & {\footnotesize 0.551}{\scriptsize{} } & {\footnotesize 0.507}{\scriptsize{} } & {\footnotesize 0.173}\tabularnewline
 & {\footnotesize Std.}{\scriptsize{} } & {\footnotesize 0.206}{\scriptsize{} } & {\footnotesize 0.212}{\scriptsize{} } & {\footnotesize 0.117}\tabularnewline
\hline 
\end{tabular}}{\scriptsize\par}
\par\end{centering}
\begin{centering}
{\scriptsize{} }{\scriptsize\par}
\par\end{centering}
{\footnotesize Notes. Based on 20 randomly generated datasets with five random initial starting points. ``SLC-Spectral'' denotes the SLC algorithm combined with the spectral algorithm as a fixed\nobreakdash-point iteration acceleration method.}{\footnotesize\par}

{\scriptsize{} }{\scriptsize\par}
\end{table}
}{\scriptsize\par}
\par\end{center}

\section{Conclusions}\label{sec:Conclusions}

The current study has investigated sequential algorithms possessing the Zero Jacobian Property (ZJP) for estimating structural models with equilibrium constraints, which can also be regarded as numerical algorithms for solving constrained optimization problems. By developing a more general framework than those in previous studies, the study shows that they achieve near-quadratic local convergence in large samples, even without consistent initial estimates, in the context of MLE or GMM.

It then proposes a novel algorithm termed the Sequential Linearly Constrained (SLC) algorithm. It is easily applicable to a broader class of structural models. A key advantage of SLC is that it can be implemented without explicitly computing the Jacobian of the equilibrium constraints, and it can be several times faster than the NFXP approach. This study demonstrates its strong performance through two numerical experiments: a dynamic discrete game with time-varying unobserved heterogeneity and a dynamic BLP model.

While the SLC algorithm performs well even without consistent initial estimates of $\theta$ and $Y$, incorporating such estimates is desirable as the algorithm can yield statistically efficient parameter estimates even after one iteration. Computationally efficient methods for obtaining consistent estimates exist for specific economic models---such as the CCP-based approach for dynamic discrete choice models with finite dependence (\citet{arcidiacono2011conditional})---and such methods can be combined with SLC to obtain statistically efficient parameter estimates. However, whether such computationally light methods exist for a broader class of models remains an open question. Developing these general methods with low computational costs for consistent parameter estimation remains a subject for future research.

\appendix

\section{Further discussions}\label{sec:Further-discussions}

\subsection{Comparison with Lagrangian-based sequential algorithms}\label{subsec:Comparison-with-Lagrangian-based}

Many of the economic studies proposing the ``MPEC'' approach apply the KNITRO solver based on the Interior Point Method. In the absence of inequality constraints in the constrained optimization problems, the solver essentially relies on the SQP algorithm,\footnote{I thank Richard Waltz for clarifying this point.} which explicitly introduces the Lagrangian of the constrained optimization problem and achieves quadratic local convergence speed. Sequential algorithms share similarities with the SQP algorithm, in that variables are updated sequentially and the constraint $G\left(Y;\theta\right)=0$ is not solved exactly in each iteration.

Table \ref{tab:Comparisons-of-Sequential} compares the subproblem solved at each iteration of the respective sequential algorithms. In addition to SLC and SQP, we compare EPL, NPL, and ABLP (non-Lagrangian-based) with the Linearly Constrained Lagrangian (LCL) algorithm (Lagrangian-based).\footnote{The LCL algorithm is implemented in the MINOS solver. The local convergence speed of LCL is also quadratic, as shown in \citet{robinson1972quadratically}. Another widely known Lagrangian-based algorithm is the Augmented Lagrangian Method (ALM), which is implemented in the LANCELOT solver. See also Chapter 17 of \citet{nocedal2006numerical}.}

As is clear from the table, all the algorithms presented can be reformulated so that each solves a constrained optimization problem in each iteration. Both SLC and SQP solve constrained optimization problems whose constraints are the linearization of the original constraint $G\left(\gamma\right)=0$. The only difference between them is the objective function of the subproblem solved in each iteration: SLC minimizes the original objective function $Q\left(\theta,Y\right)$ subject to the linearized constraint, whereas SQP minimizes the quadratic objective function $\left(\gamma-\gamma_{k}\right)^{\prime}W_{k}\left(\gamma-\gamma_{k}\right)+\left(\nabla_{\gamma}Q\left(\gamma_{k}\right)\right)\left(\gamma-\gamma_{k}\right)$, which is derived from the original objective function $Q\left(\theta,Y\right)$ and the Lagrangian function. Note that LCL also uses the linearized constraint as in SLC, but minimizes the objective function $Q\left(\gamma\right)-\lambda_{k}^{T}\overline{G}\left(\gamma;\gamma_{k}\right)$, which explicitly uses the Lagrange multiplier $\lambda_{k}$.

\begin{table}[H]
\caption{Comparisons of sequential algorithms}\label{tab:Comparisons-of-Sequential}

\begin{centering}
\begin{tabular}{ccl}
\hline 
{\footnotesize Lagrangian-based}  & {\footnotesize Algorithm}  & {\footnotesize Problem to be solved in each iteration}\tabularnewline
\hline 
\hline 
\multirow{9}{*}{{\footnotesize No}} & \multirow{3}{*}{{\footnotesize SLC}} & {\footnotesize$\min_{\theta,Y}Q\left(\theta,Y\right)$}\tabularnewline
 &  & {\footnotesize$\text{s.t.}\ \left(\nabla_{\theta}G(Y_{k};\theta_{k})\right)\left(\theta-\theta_{k}\right)+\left(\nabla_{Y}G(Y_{k};\theta_{k})\right)\left(Y-Y_{k}\right)+G(Y_{k};\theta_{k})=0$}\tabularnewline
 &  & {\footnotesize The constraint can be rewritten as $\left(\nabla_{\gamma}G(\gamma_{k})\right)\left(\gamma-\gamma_{k}\right)+G(\gamma_{k})=0$.}\tabularnewline
\cline{2-3}
 & \multirow{2}{*}{{\footnotesize EPL}} & {\footnotesize$\min_{\theta,Y}Q\left(\theta,Y\right)$}\tabularnewline
 &  & {\footnotesize$\text{s.t.}\ \left(\nabla_{\theta}G(Y_{k};\theta_{k})\right)\left(\theta-\theta_{k}\right)+G(Y_{k};\theta)=0$}\tabularnewline
\cline{2-3}
 & \multirow{2}{*}{{\footnotesize NPL}} & {\footnotesize$\min_{\theta,Y}Q\left(\theta,Y\right)$}\tabularnewline
 &  & {\footnotesize$\text{s.t.}\ Y=\Phi\left(Y_{k};\theta\right)$}\tabularnewline
\cline{2-3}
 & \multirow{2}{*}{{\footnotesize ABLP}} & {\footnotesize$\min_{\theta,Y}Q\left(\theta,Y\right)$}\tabularnewline
 &  & {\footnotesize$\text{s.t.}\left(\nabla_{\theta}G(Y_{k};\theta)\right)\left(\theta-\theta_{k}\right)+G(Y_{k};\theta)=0$}\tabularnewline
\hline 
\multirow{7}{*}{{\footnotesize Yes}} & \multirow{4}{*}{{\footnotesize SQP}} & {\footnotesize$\min_{\gamma}\left(\gamma-\gamma_{k}\right)^{\prime}W_{k}\left(\gamma-\gamma_{k}\right)+\left(\nabla_{\gamma}Q\left(\gamma_{k}\right)\right)\left(\gamma-\gamma_{k}\right)$}\tabularnewline
 &  & {\footnotesize$\text{s.t.}\ \left(\nabla_{\gamma}G(\gamma_{k})\right)\left(\gamma-\gamma_{k}\right)+G(\gamma_{k})=0$}\tabularnewline
 &  & {\footnotesize where $W_{k}\equiv\nabla_{\gamma\gamma^{\prime}}\mathcal{L}\left(\gamma_{k};\lambda_{k}\right)$, $\mathcal{L}$ is the Lagrangian function, and}\tabularnewline
 &  & {\footnotesize$\lambda_{k}$ is the counterpart of the Lagrange multiplier.}\tabularnewline
\cline{2-3}
 & \multirow{3}{*}{{\footnotesize LCL}} & {\footnotesize$\min_{\gamma}Q\left(\gamma\right)-\lambda_{k}^{T}\overline{G}\left(\gamma;\gamma_{k}\right)$}\tabularnewline
 &  & {\footnotesize$\text{s.t.}\ \left(\nabla_{\gamma}G(\gamma_{k})\right)\left(\gamma-\gamma_{k}\right)+G(\gamma_{k})=0$}\tabularnewline
 &  & {\footnotesize where $\overline{G}\left(\gamma;\gamma_{k}\right)=G\left(\gamma\right)-G\left(\gamma_{k}\right)-\nabla_{\gamma}G\left(\gamma_{k}\right)^{T}\left(\gamma-\gamma_{k}\right)$}\tabularnewline
\hline 
\end{tabular}
\par\end{centering}
{\scriptsize Notes:}{\scriptsize\par}

{\scriptsize SLC: Sequential Linearly Constrained method, proposed in the current paper}{\scriptsize\par}

{\scriptsize EPL: Efficient Pseudo-Likelihood algorithm, proposed by \citet{dearing2024efficient}}{\scriptsize\par}

{\scriptsize NPL: Nested Pseudo-Likelihood algorithm, proposed by \citet{aguirregabiria2002swapping} for estimating single-agent DDC models}{\scriptsize\par}

{\scriptsize ABLP: Approximate BLP algorithm, proposed by \citet{lee2015computationally} for estimating static BLP models}{\scriptsize\par}

{\scriptsize SQP: Sequential Quadratic Programming}{\scriptsize\par}

{\scriptsize LCL: Linearly Constrained Lagrangian method (cf. \citealp{robinson1972quadratically,murtagh1982projected})}{\scriptsize\par}

\end{table}

It is worth mentioning that the algorithm where $\lambda_{k}$ is arbitrarily set to 0 in the LCL algorithm, which coincides with the SLC algorithm, was mentioned in \citet{murtagh1982projected}, which proposes the MINOS solver based on the LCL algorithm. They noted in Section 3.3 of their paper that the algorithm converges in certain cases, particularly when the solution lies at a vertex (e.g., reverse convex problems). However, in general problems, the local convergence speed of SLC is not guaranteed to be fast. Yet, as mentioned in Section \ref{subsec:Local-convergence-speed}, the local convergence speed of SLC is guaranteed to be close to quadratic in large samples in MLE and GMM, making SLC a strong candidate for solving constrained optimization problems in statistical applications.

Note that SQP finds a solution of the Karush--Kuhn--Tucker (KKT) conditions of the constrained optimization problem (\ref{eq:obj}). Regarding the KKT conditions, the following statement holds: 
\begin{prop}
\label{prop:FOC-KKT}Define {\footnotesize
\begin{eqnarray*}
\Gamma_{KKT} & \equiv & \left\{ \gamma:\text{Solution of }G\left(\gamma\right)=0\ \text{and}\ \nabla_{\gamma}Q\left(\gamma\right)-\lambda^{T}\left(\nabla_{\gamma}G\left(\gamma\right)\right)=0;\ \text{rows of }\nabla_{\gamma}G\left(\gamma\right)\ \text{are linearly independent}\right\} 
\end{eqnarray*}
and} $\Gamma_{KKT,nonsingular}\equiv\left\{ \gamma\in\Gamma_{KKT}:\nabla_{Y}G\left(\gamma\right)\ \text{is nonsingular}\right\} \subset\Gamma_{KKT}$. Then, under Assumptions \ref{as:existence_constrained_opt_sol} and \ref{as:conti_diffble},

(a). For $\gamma\in\Gamma_{KKT,nonsingular}$ and the corresponding Lagrange multiplier $\lambda$, the following holds:

\begin{eqnarray*}
\begin{cases}
\lambda^{T}=\left(\nabla_{Y}Q\left(\gamma\right)\right)\left(\nabla_{Y}G\left(\gamma\right)\right)^{-1}\\
\nabla_{\theta}Q\left(\theta,Y\right)-\left(\nabla_{Y}Q\left(\theta,Y\right)\right)\left(\nabla_{Y}G\left(Y;\theta\right)\right)^{-1}\left(\nabla_{\theta}G\left(Y;\theta\right)\right)=0\\
G\left(\gamma\right)=0
\end{cases}
\end{eqnarray*}

(b). $\widehat{\gamma}\in\Gamma_{FOC}=\Gamma_{KKT,nonsingular}$. 
\end{prop}

\subsection{Comparison of computational costs}\label{subsec:Comparison-of-computational-costs}

Computational costs are one of the most important factors that practitioners care about. Table \ref{tab:Computational-costs-of-algorithms} compares the computational costs of the main algorithms discussed so far: NFXP, SLC, and SQP. Because practitioners sometimes rely on the numerical derivative-based (Jacobian-free) NFXP algorithm, I also classify the algorithms based on whether the Jacobian-free approach is applied or not.

\begin{table}[H]
\caption{Computational costs of algorithms }\label{tab:Computational-costs-of-algorithms}

\begin{centering}
{\tiny{}}{\tiny{}%
\begin{tabular}{ccccc}
\hline 
{\tiny Compute } & \multirow{1}{*}{{\tiny Algorithm}} & \multirow{1}{*}{{\tiny Computational cost per iteration}} & {\tiny Memory } & {\tiny Outer-loop}\tabularnewline
{\tiny Jacobians } &  &  & {\tiny size } & {\tiny local conv. speed}\tabularnewline
\hline 
\multirow{7}{*}{{\tiny Yes}} & \multirow{2}{*}{{\tiny NFXP}} & {\tiny$c\left(\text{Solve}\ G(Y;\theta)=0\ \text{for}\ Y\right)+c\left(\nabla_{Y}Q\left(\theta,Y\right)\right)+c\left(\nabla_{\theta}Q\left(\theta,Y\right)\right)+$ } & \multirow{2}{*}{{\tiny$O\left(\max\left(n_{G},n_{Y}^{2}\right)\right)$}} & {\tiny Superlinear{*}}\tabularnewline
 &  & {\tiny$c\left(\nabla_{Y}G\left(Y;\theta\right)\right)+\sum_{i=1}^{n_{\theta}}c\left(\left(\nabla_{Y}G\left(Y;\theta\right)\right)^{-1}\nabla_{\theta_{i}}G\left(Y;\theta\right)\right)$ } &  & {\tiny (BFGS case)}\tabularnewline
\cline{2-5}
 & \multirow{3}{*}{{\tiny SLC}} & {\tiny$c\left(\text{Solve}\ \theta_{k+1}=\arg\min_{\theta}Q\left(\theta,\Upsilon\left(\theta;\gamma\right)\right)\right)+$ } & \multirow{3}{*}{{\tiny$O\left(\max\left(n_{G},n_{Y}^{2}\right)\right)$}} & {\tiny Nearly Quadratic}\tabularnewline
 &  & {\tiny$c\left(\nabla_{Y}G\left(Y;\theta\right)\right)+c\left(\left(\nabla_{Y}G\left(Y;\theta\right)\right)^{-1}\left(G\left(Y;\theta\right)\right)\right)+$ } &  & {\tiny (MLE, GMM; Large sample)}\tabularnewline
 &  & {\tiny$\sum_{i=1}^{n_{\theta}}c\left(\left(\nabla_{Y}G\left(Y;\theta\right)\right)^{-1}\nabla_{\theta_{i}}G\left(Y;\theta\right)\right)$ } &  & \tabularnewline
\cline{2-5}
 & \multirow{2}{*}{{\tiny SQP}} & {\tiny$c\left(\nabla_{\gamma}G\left(\gamma\right)\right)+c\left(\nabla_{\gamma}Q\left(\gamma\right)\right)+c\left(\nabla_{\gamma\gamma^{\prime}}\mathcal{L}\left(\gamma,\lambda\right)\right)+$ } & \multirow{2}{*}{{\tiny$O\left(\max\left(n_{G},\left(2n_{Y}+n_{\theta}\right)^{2}\right)\right)$}} & \multirow{2}{*}{{\tiny Quadratic}}\tabularnewline
 &  & {\tiny$c\left(\left(\begin{array}{cc}
\nabla_{\gamma\gamma^{\prime}}\mathcal{L}\left(\gamma,\lambda\right) & -\left(\nabla_{\gamma}G\left(\gamma\right)\right)^{\prime}\\
\nabla_{\gamma}G\left(\gamma\right) & 0
\end{array}\right)^{-1}\left(\begin{array}{c}
-\nabla_{\gamma}Q\left(\gamma\right)\\
-G\left(\gamma\right)
\end{array}\right)\right)$ } &  & \tabularnewline
\hline 
\multirow{5}{*}{{\tiny No}} & \multirow{2}{*}{{\tiny NFXP}} & {\tiny$c\left(\text{Solve}\ G(Y;\theta)=0\ \text{for}\ Y\right)+$ } & \multirow{2}{*}{{\tiny$O\left(n_{G}\right)$}} & {\tiny Superlinear{*}}\tabularnewline
 &  & {\tiny$\sum_{i=1}^{n_{\theta}}c\left(\text{Solve}\ G(Y;\theta_{i}\pm\epsilon_{i},\theta_{-i})=0\ \text{for}\ Y\right)$ } &  & {\tiny (BFGS case)}\tabularnewline
\cline{2-5}
 & \multirow{3}{*}{{\tiny SLC}} & {\tiny$c\left(\text{Solve}\ \theta_{k+1}=\arg\min_{\theta}Q\left(\theta,\Upsilon\left(\theta;\gamma\right)\right)\right)+$ } & \multirow{3}{*}{{\tiny$O\left(n_{G}\right)$}} & {\tiny Nearly Quadratic}\tabularnewline
 &  & {\tiny$c\left(\text{Solve}\ \left(\nabla_{Y}G\left(Y;\theta\right)\right)v=G\left(Y;\theta\right)\ \text{for}\ v\in\mathbb{R}^{n_{Y}}\right)+$ } &  & {\tiny (MLE, GMM; Large sample)}\tabularnewline
 &  & {\tiny$\sum_{i=1}^{n_{\theta}}c\left(\text{Solve}\ \left(\nabla_{Y}G\left(Y;\theta\right)\right)v=\nabla_{\theta_{i}}G\left(Y;\theta\right)\ \text{for}\ v\in\mathbb{R}^{n_{Y}}\right)$ } &  & \tabularnewline
\hline 
\end{tabular}}{\tiny\par}
\par\end{centering}
{\footnotesize Notes. $c(\cdot)$ denotes the computational cost of computing a variable or implementing an operation. $n_{G}$ denotes the memory size required to compute $G\left(Y;\theta\right)$.}{\footnotesize\par}

{\footnotesize{*}: If we apply the BHHH algorithm in maximum likelihood estimation, the local convergence speed is close to quadratic in large samples.}{\footnotesize\par}

{\footnotesize If we cannot rule out the possibility of multiple solutions to $G\left(Y;\theta\right)=0$, the objective function used in the NFXP algorithm may not be differentiable (cf. \citealp{su2014estimating}), and we may need to rely on optimization algorithms that do not require differentiability of the objective function for good performance.}{\footnotesize\par}

\end{table}

Concerning the computational cost per iteration, the relatively small cost of evaluating the objective function $Q$, compared with the cost of evaluating the constraint function $G$, is important for the good performance of the SLC algorithm. In SLC, repeated evaluations of $Q$ are required in each iteration, in contrast to NFXP. As long as the computational burden of evaluating $Q$ is much smaller than that of evaluating $G$, the overall cost per iteration remains low, contributing to the favorable performance of SLC.

Table \ref{tab:Computational-costs-of-algorithms} also shows the scale of the worst-case memory size of variables stored during the iterations. Here, let $n_{G}$ be the memory size required to compute $G\left(Y;\theta\right)$. In the case of NFXP and SLC with Jacobians, we need to compute the values of $G\left(Y;\theta\right)$, which require memory size $n_{G}$, and the Jacobians, whose memory size is $n_{Y}^{2}$. Consequently, the memory size required is of order $O\left(\max\left(n_{G},n_{Y}^{2}\right)\right)$. Similarly, in the SQP algorithm, we need to compute $\left(\begin{array}{c}
\nabla_{\gamma}Q\left(\gamma\right)\\
\nabla_{\gamma}G\left(\gamma\right)
\end{array}\right)$, and the memory size is of order $O\left(\max\left(n_{G},\left(2n_{Y}+n_{\theta}\right)^{2}\right)\right)$. In contrast, in the NFXP and SLC algorithms using numerical derivatives, they do not require computing Jacobians, and the memory requirement is of order $O\left(n_{G}\right)$. In that sense, the memory requirement is generally smaller for Jacobian-free algorithms than for Jacobian-based algorithms.

\subsection{Model with Stochastic $Y$}\label{subsec:Stochastic-Y}

In the main part of this paper, we have focused on models where $Y$ is deterministic and fixed-dimensional, and the constraint $G$ does not depend on data, to investigate the properties of sequential algorithms with ZJP. However, BLP models, whose numerical experiments are shown in Section \ref{subsec:Dynamic-BLP-numerical}, do not satisfy these conditions. In this section, we investigate the properties of sequential algorithms with ZJP allowing for models with stochastic, variable-dimensional $Y$ and data-dependent constraint $G$.\footnote{For $\theta$, we maintain the assumption of finite-dimensional $\theta$. In addition, we also assume that the number of moments is fixed.}{} The discussion below clarifies what additional assumptions are required to maintain the statements in Section \ref{sec:structural-estimation}. Note that the discussion in Section \ref{sec:Sequential-algorithms} does not depend on any statistical properties, and we can inherit its results.

Under the stochastic-$Y$ setting, the constrained optimization problem can be formulated as follows:

\begin{eqnarray*}
 & \min_{\theta,Y} & Q\left(\theta,Y,\omega^{Q}\right)\\
 & \text{s.t.} & G\left(Y;\theta,\omega^{G}\right)=0,
\end{eqnarray*}
where $\omega\equiv\left(\omega^{Q},\omega^{G}\right)$ denotes the observation and can be stochastic. Note that $\omega^{Q}$ and $\omega^{G}$ can be dependent. We additionally define $Y^{*}\left(\theta;\omega\right)\equiv\arg\min_{Y\ \text{s.t.}\ G\left(Y;\theta,\omega^{G}\right)=0}Q\left(\theta,Y,\omega^{Q}\right)$ and $\gamma^{*}\left(\omega\right)\equiv\left(\theta^{*},Y^{*}\left(\theta^{*};\omega\right)\right)$, which should be regarded as a random variable due to its dependence on $\omega$.

\subsubsection{Discussions without relying on consistent estimates of $\gamma_{0}$}

Section \ref{subsec:dQ_dY} shows that $\nabla_{Y}Q\left(\widehat{\theta},\widehat{Y}\right)=O_{p}\left(N^{-\frac{1}{2}}\right)$ holds in deterministic-$Y$ models, and derives Proposition \ref{prop: asy_property}. In stochastic-$Y$ models, we introduce the following assumption, which is the counterpart of Assumption \ref{as:stat_ass} in Section \ref{sec:structural-estimation}:

\begin{assumption}

Assumption \ref{as:stat_ass} (a), (b), (d), (f), (g-MLE), (h-MLE), (g-GMM), (h-GMM), and the following alternative assumptions hold:

(e{*}). $\left\Vert \widehat{\gamma}-\gamma^{*}\right\Vert _{\infty}=O_{p}\left(N^{-\frac{1}{2}}\right)$ and $\sqrt{N}\left(\widehat{\theta}-\theta^{*}\right)\rightarrow_{d}N\left(0,\Sigma_{\theta}\right)$.

For MLE:

(c{*}-MLE): $\frac{1}{N}\sum_{i=1}^{N}\ln f\left(\omega_{i}^{Q}|\theta,Y\left(\theta,\omega_{i}\right)\right)\rightarrow_{a.s.}E_{\omega}\left[\ln f\left(\omega^{Q}|\theta,Y\left(\theta,\omega\right)\right)\right]$ holds uniformly in $\theta\in\Theta$ for any $Y\left(\theta,\omega\right)$ such that $Y\left(\theta,\omega\right)-Y^{*}\left(\theta^{*},\omega\right)\rightarrow_{a.s.}0$.

For GMM:

(c{*}-GMM): $\frac{1}{N}\sum_{i=1}^{N}\left[m_{i}\left(\theta,Y\left(\theta,\omega_{i}\right),\omega_{i}^{Q}\right)\right]\rightarrow_{a.s.}E_{\omega}\left[m\left(\theta,Y\left(\theta,\omega\right),\omega^{Q}\right)\right]$ holds uniformly in $\theta\in\Theta$ for any $Y\left(\theta,\omega\right)$ such that $Y\left(\theta,\omega\right)-Y^{*}\left(\theta^{*},\omega\right)\rightarrow_{a.s.}0$.

(j{*}-GMM): $\left(E\left[m\left(\omega;\theta,Y\right)\right]=0\ \text{and }G\left(Y;\theta\right)=0\right)$ $\Rightarrow\theta=\theta^{*},Y=Y^{*}\left(\theta^{*},\omega\right)$.

\label{as:stat_ass_stochastic_Y} \end{assumption}

In part (e{*}), we impose the assumption $\left\Vert \widehat{\gamma}-\gamma^{*}\right\Vert _{\infty}=O_{p}\left(N^{-\frac{1}{2}}\right)$. Although Assumption \ref{as:stat_ass} assumes $\widehat{\gamma}-\gamma^{*}=O_{p}\left(N^{-\frac{1}{2}}\right)$ in the finite-dimensional $Y$ setting, the current section considers a variable-dimensional $Y$ setting, so introducing the assumption $\left\Vert \widehat{\gamma}-\gamma^{*}\right\Vert _{\infty}=O_{p}\left(N^{-\frac{1}{2}}\right)$ is appropriate for further discussion. Parts (c{*}-MLE) and (c{*}-GMM) are introduced instead of assumption (c) that observations are i.i.d.

\medskip{}

We additionally impose the following assumption.

\begin{assumption}

(a). $\max_{m=1,\cdots,n_{Y}}\sum_{r=1}^{n_{\gamma}}\left\Vert \nabla_{\theta\gamma_{r}}\Upsilon_{m}\left(\widehat{\theta},\widehat{Y}\right)\right\Vert _{\infty}=O_{p}\left(1\right)$

For MLE,

(b-MLE). $\sum_{r=1}^{n_{Y}}\left|\frac{1}{N}\sum_{i=1}^{N}\nabla_{Y_{r}}\ln f\left(\omega_{i}|\gamma^{*}\right)\right|=O_{p}\left(N^{-\frac{1}{2}}\right)$

(c-MLE). $\sum_{p,q=1}^{n_{Y}}\left|\frac{1}{N}\sum_{i=1}^{N}\nabla_{Y_{p}Y_{q}}\ln f\left(\omega_{i}|\gamma^{*}\right)\right|=O_{p}\left(1\right)$

For GMM,

(b-GMM). $\sum_{r=1}^{n_{Y}}\left|\nabla_{Y_{r}}m\left(\widehat{\theta},\widehat{Y}\right)\right|=O_{p}\left(1\right)$.\label{as:stochastic_Y_ass} \end{assumption}

Related to part (b-MLE), the discussion in Section \ref{subsec:dQ_dY} implies that $\frac{1}{N}\sum_{i=1}^{N}\nabla_{Y}\ln f\left(w_{i}|\gamma^{*}\right)=O_{p}\left(N^{-\frac{1}{2}}\right)$ holds by $\int\left(\nabla_{Y}\ln f\left(w|\gamma^{*}\right)\right)f^{*}\left(w\right)dw=0$, the central limit theorem, and regularity conditions in the deterministic-$Y$ setting. Assumption (b-MLE) requires that the $L_{1}$ norm of the vector $\frac{1}{N}\sum_{i=1}^{N}\nabla_{Y}\ln f\left(w_{i}|\gamma^{*}\right)$ is still $O_{p}\left(N^{-\frac{1}{2}}\right)$ in the stochastic-$Y$ setting, where the dimension of $Y$ can grow as the sample size increases.

\medskip{}

To clarify the meaning of Assumption \ref{as:stochastic_Y_ass}, consider a static entry/exit game as an example of MLE and a static BLP model as an example of GMM. 
\begin{example}
Static entry/exit game

Consider a static incomplete-information entry/exit game whose estimation problem is given as follows:

\begin{eqnarray*}
\min_{\theta,p\equiv\left\{ p_{jt}\right\} _{j,t}} & \sum_{j=1}^{J}\sum_{t=1}^{T}\ln f\left(d_{jt}|\theta,p_{t},X_{t}\right)\\
\text{s.t.} & p_{jt}=\Phi\left(\theta,p_{-jt},X_{t}\right) & \left(j=1,\cdots,J,\ t=1,\cdots,T\right)
\end{eqnarray*}
 
\end{example}
Here, $d_{jt}$, $p_{jt}$, $X_{t}$, and $\theta$ denote a dummy variable for firm $j$'s entry decision, firm $j$'s entry probability, the market-level observed variable in market $t$, and parameters governing firms' profit functions, respectively. The function $f\left(d_{jt}|\theta,p_{t},X_{t}\right)$ denotes firm $j$'s choice probability. Let $J$ denote the number of firms and $T$ denote the number of markets used in the estimation. The function $\Phi$ characterizes the response function of the firms. When we allow for the case where $X_{t}$ takes continuous values, $p\equiv\left\{ p_{jt}\right\} _{j=1,\cdots,J,t=1,\cdots,T}$, which can be thought of as the counterpart of $Y$, should be treated as a variable-dimensional variable whose dimension increases as the number of markets increases.

In this setting, suppose that each market is independent. Then, $p_{jt}$ in one market does not affect another market's constraint $p_{jt^{\prime}}=\Phi\left(\theta,p_{-jt^{\prime}},X_{t^{\prime}}\right)\ \left(t^{\prime}\neq t\right)$, and Assumption \ref{as:stochastic_Y_ass}(a) holds at least when focusing on SLC, EPL, or ABLP mappings, even when the number of markets grows. Similarly, $p_{jt}$ in one market does not affect choice probabilities and $f\left(\theta,p_{t},X_{t}\right)$ in other markets, which suggests that Assumption \ref{as:stochastic_Y_ass}(b-MLE) and (c-MLE) hold even when the number of markets grows. 
\begin{example}
Static BLP model 
\end{example}
Consider a static BLP (random-coefficient logit) model whose estimation problem is given as follows:

\begin{eqnarray*}
 & \min_{\theta,V\equiv\left\{ V_{it}\right\} _{i,t}} & \left[\frac{1}{T}\sum_{t=1}^{T}\left(\iota_{V\rightarrow\xi}\left(V_{t};\theta,X_{t},S_{t}^{(data)}\right)\right)^{\prime}Z_{t}\right]^{\prime}W\left[\frac{1}{T}\sum_{t=1}^{T}\left(\iota_{V\rightarrow\xi}\left(V_{t};\theta,X_{t},S_{t}^{(data)}\right)\right)^{\prime}Z_{t}\right]\\
 & \text{s.t.} & V_{it}=\Phi_{V,it}\left(V_{t};\theta,X_{t},S_{t}^{(data)}\right)\ \left(i=1,\cdots,I,\ t=1,\cdots,T\right)
\end{eqnarray*}

To clarify the connection with the dynamic BLP model in Section \ref{subsec:Dynamic-BLP-numerical}, we treat the inclusive value $V$ as a nuisance variable here.\footnote{Analogous discussions also hold even when treating unobserved product characteristics as the counterpart of $Y$, as in \citet{dube2012improving}.}{} Here, $i$ denotes an (appropriately discretized) consumer type, $t$ denotes a market, and $I$ and $T$ denote the number of consumer types and the number of markets, respectively. The parameter $\theta$ represents utility parameters. $V_{it}$, $X_{t}$, $\xi_{t}$, $S_{t}^{(data)}$, and $Z_{t}$ denote consumer $i$'s inclusive value (value function), the vector of observed product characteristics, unobserved product characteristics, market shares, and the set of instrumental variables in market $t$, respectively. The functions $\iota_{V\rightarrow\xi}$ and $\Phi_{V}$ denote a mapping from $V_{t}$ to $\xi_{t}$ and a fixed-point mapping on $V$, which are introduced in \citet{fukasawa2024fast}.

In this setting, suppose that each market is independent. Then, $V_{it}$ in one market does not affect another market's constraint $V_{it^{\prime}}=\Phi_{V,it^{\prime}}\left(V_{t^{\prime}};\theta,X_{t^{\prime}},S_{t^{\prime}}^{(data)}\right)\ \left(t^{\prime}\neq t\right)$, and Assumption \ref{as:stochastic_Y_ass}(a) holds even when the number of markets grows. Similarly, $V_{it}$ in market $t^{\prime}$ does not affect $\left(\iota_{V\rightarrow\xi}\left(V_{t^{\prime}};\theta,X_{t^{\prime}},S_{t^{\prime}}^{(data)}\right)\right)^{\prime}Z_{t^{\prime}}$ for $t^{\prime}\neq t$, and Assumption \ref{as:stochastic_Y_ass}(b-GMM), namely, $\sum_{r=1}^{n_{Y}}\left|\nabla_{Y_{r}}m\left(\widehat{\theta},\widehat{Y}\right)\right|=\sum_{r=1}^{n_{Y}}\left|\frac{1}{T}\sum_{t=1}^{T}\nabla_{Y_{r}}m_{t}\left(\widehat{\theta},\widehat{Y}\right)\right|=O_{p}\left(1\right)$, holds even when the number of markets grows. Note that the dynamic BLP model in Section \ref{subsec:Dynamic-BLP-numerical} considers a setting where each market (time period) is interdependent. That said, we can expect the assumption to hold as long as the interdependence is moderate, and we maintain the assumption.

\medskip{}

Under these assumptions, we can derive the following claim, which is the counterpart of the discussion in Section \ref{subsec:dQ_dY}. 
\begin{prop}
\label{prop: stoch_property_stoch_Y}Suppose that Assumptions \ref{as:existence_constrained_opt_sol}, \ref{as:conti_diffble}, \ref{as:mapping_cdns}, \ref{as:nonsingular_Jacobian}, \ref{as:stat_ass_stochastic_Y}, and \ref{as:stochastic_Y_ass}{} hold. Then, the following statements hold:

(1). $\left\Vert \nabla_{Y}Q\left(\widehat{\theta},\widehat{Y}\right)\right\Vert _{1}=O_{p}\left(N^{-\frac{1}{2}}\right)$

(2). $\left(\nabla_{\theta\gamma^{\prime}}\widetilde{Q}\left(\widehat{\theta};\widehat{\gamma}\right)\right)\left(\gamma_{k-1}-\widehat{\gamma}\right)=O_{p}\left(N^{-\frac{1}{2}}\left\Vert \gamma_{k-1}-\widehat{\gamma}\right\Vert _{\infty}\right)$

(3). $\gamma_{k}-\widehat{\gamma}=O_{p}\left(N^{-\frac{1}{2}}\left\Vert \gamma_{k-1}-\widehat{\gamma}\right\Vert _{\infty}\right)+O\left(\left\Vert \gamma_{k-1}-\widehat{\gamma}\right\Vert _{\infty}^{2}\right)$

(4). $\nabla_{\theta\theta^{\prime}}\widetilde{Q}\left(\widehat{\theta};\widehat{\gamma}\right)-\nabla_{\theta\theta^{\prime}}Q_{NFXP}\left(\widehat{\theta}\right)=O_{p}\left(N^{-\frac{1}{2}}\right)$. 
\end{prop}

\subsubsection{Asymptotic property under consistent initial values}

To show the counterpart of Proposition \ref{prop: asy_property_consistent_initial_values} in Section \ref{subsec:Asymptotic-property-consistent-initial-values}, we further introduce the following assumption:

\begin{assumption}

(a). $\sum_{p,q=1}^{n_{\gamma}}\left|\nabla_{\theta_{j}\gamma_{p}\gamma_{q}}\widetilde{Q}\left(\theta^{*};\gamma\right)\right|=O_{p}\left(1\right)\ \left(j=1,\cdots,n_{\theta}\right)$ for $\gamma$ such that $\gamma-\gamma^{*}=o_{p}\left(1\right)$.

(b). $\max_{j=1,\cdots,n_{Y}}\left|\nabla_{\theta\theta^{\prime}}\Upsilon_{j}(\widehat{\theta};\widehat{\gamma})-\nabla_{\theta\theta^{\prime}}\mathcal{Y}_{j}(\widehat{\theta})\right|=O_{p}\left(1\right)$.

\label{as:stochastic_Y_ass_for_asy_dist} \end{assumption}

In Assumption \ref{as:stochastic_Y_ass_for_asy_dist}(a), we impose an assumption on $\widetilde{Q}\left(\theta;\gamma\right)$. Although imposing assumptions on primitives (e.g., $\Upsilon$, moments, and likelihood) as higher-order-derivative counterparts of Assumption \ref{as:stochastic_Y_ass} is an alternative strategy, it requires lengthy derivations due to the presence of higher-order derivatives. Hence, the current study directly imposes an assumption on $\widetilde{Q}\left(\theta;\gamma\right)$ to simplify the discussion. The intuition behind this assumption is similar to that for Assumption \ref{as:stochastic_Y_ass}.

Under these assumptions, we obtain the following statement. 
\begin{prop}
\label{prop: asy_property_consistent_initial_values_stoch_Y}

Suppose that the initial value $\gamma_{0}\equiv\left(\theta_{0},Y_{0}\right)$ satisfies $\gamma_{0}-\gamma^{*}\rightarrow_{a.s.}0$. In addition, suppose that Assumptions \ref{as:existence_constrained_opt_sol}, \ref{as:conti_diffble}, \ref{as:mapping_cdns}, \ref{as:nonsingular_Jacobian}, \ref{as:stat_ass_stochastic_Y}, \ref{as:stochastic_Y_ass}, and{} \ref{as:stochastic_Y_ass_for_asy_dist} hold. Then, the following statements hold:

(1). (Consistency) $\gamma_{k}-\gamma^{*}\rightarrow_{a.s.}0$.

(2). (Efficiency) If $\left\Vert \gamma_{0}-\gamma^{*}\right\Vert _{\infty}=o_{p}\left(N^{-\frac{1}{4}}\right)$, $\sqrt{N}\left(\theta_{k}-\theta^{*}\right)\rightarrow_{d}N\left(0,\Sigma_{\theta}\right)$ holds.

(3). If $\left\Vert \gamma_{0}-\gamma^{*}\right\Vert _{\infty}=O_{p}\left(N^{-b}\right)$ for $b\in(1/4,1/2]$, $\left\Vert \gamma_{k}-\widehat{\gamma}\right\Vert _{\infty}=O_{p}\left(N^{-(k-1)/2-2b}\right)$ holds. If $b=\frac{1}{2}$, $\left\Vert \gamma_{k}-\widehat{\gamma}\right\Vert _{\infty}=O_{p}\left(N^{-(k+1)/2}\right)$ holds. 
\end{prop}

\subsubsection{Neyman orthogonality}

In the discussion of Neyman orthogonality under deterministic-$Y$ models, the property $E\left[\nabla_{\theta\gamma^{\prime}}\widetilde{Q}\left(\theta^{*},\gamma^{*}\right)\right]=0$ is the key for Neyman orthogonality in MLE, and $\nabla_{\theta\gamma^{\prime}}\widetilde{Q}\left(\theta^{*},\gamma^{*}\right)\rightarrow_{p}0$ is the key for asymptotic Neyman orthogonality in GMM. When we allow for stochastic $Y$, exact Neyman orthogonality may not hold for MLE, as $\nabla_{\theta\gamma^{\prime}}\widetilde{Q}\left(\theta^{*},\gamma^{*}\right)=\sum_{j=1}^{n_{Y}}\left(\nabla_{Y_{j}}Q\left(\theta^{*},Y^{*}\right)\right)\left(\nabla_{\theta\gamma^{\prime}}\Upsilon_{j}\left(\theta^{*},\gamma^{*}\right)\right)$ holds and the term $\nabla_{\theta\gamma^{\prime}}\Upsilon_{j}\left(\theta^{*},\gamma^{*}\right)$ may involve stochastic terms unlike in the deterministic-$Y$ setting. Nevertheless, ``asymptotic Neyman orthogonality'' still holds under Assumptions \ref{as:stochastic_Y_ass} and \ref{as:stochastic_Y_ass_for_asy_dist} because Proposition \ref{prop: stoch_property_stoch_Y}(2) implies $\nabla_{\theta\gamma^{\prime}}\widetilde{Q}\left(\theta^{*},\gamma^{*}\right)=O_{p}\left(N^{-\frac{1}{2}}\right)$.

\section{Proof}\label{sec:Proof}

\subsection{Descent direction of the SLC update (Proof of Proposition \ref{prop:descent_direction})}

To prove the proposition, we use the following lemma. 
\begin{lem}
\label{lem: ineq_merit_func}Let $\Delta_{k}\equiv\left(\nabla_{\gamma}Q\left(\gamma_{k}\right)\right)d_{k}-\mu\left\Vert G\left(\gamma_{k}\right)\right\Vert _{1}$, where $d_{k}$ denotes the direction generated by the SLC iteration. Then, for any $\alpha\in(0,1)$, there exists $C>0$ such that $\alpha\Delta_{k}-C\alpha^{2}\left\Vert d_{k}\right\Vert ^{2}\leq\phi_{1}\left(\gamma_{k}+\alpha d_{k};\mu\right)-\phi_{1}\left(\gamma_{k};\mu\right)\leq\alpha\Delta_{k}+C\alpha^{2}\left\Vert d_{k}\right\Vert ^{2}$. 
\end{lem}
\begin{proof}
First,

\begin{eqnarray}
 &  & \left\Vert G\left(\gamma_{k}\right)+\alpha\left(\nabla_{\gamma}G\left(\gamma_{k}\right)d_{k}\right)\right\Vert _{1}\nonumber \\
 & = & \left\Vert \left(1-\alpha\right)G\left(\gamma_{k}\right)+\alpha\left[G\left(\gamma_{k}\right)+\left(\nabla_{\gamma}G\left(\gamma_{k}\right)\right)d_{k}\right]\right\Vert _{1}\nonumber \\
 & = & \left(1-\alpha\right)\left\Vert G\left(\gamma_{k}\right)\right\Vert _{1}\ \left(\because G\left(\gamma_{k}\right)+\left(\nabla_{\gamma}G\left(\gamma_{k}\right)\right)d_{k}=0,\ 1-\alpha>0\right)\label{eq:G_linearize_norm}
\end{eqnarray}

By Taylor's theorem, there exist constants $c_{Q},c_{G}>0$ such that $Q\left(\gamma_{k}\right)+\alpha\nabla_{\gamma}Q\left(\gamma_{k}\right)d_{k}-c_{Q}\alpha^{2}\left\Vert d_{k}\right\Vert ^{2}\leq Q\left(\gamma_{k}+\alpha d_{k}\right)\leq Q\left(\gamma_{k}\right)+\alpha\nabla_{\gamma}Q\left(\gamma_{k}\right)d_{k}+c_{Q}\alpha^{2}\left\Vert d_{k}\right\Vert ^{2}$ and $\left\Vert G\left(\gamma_{k}\right)+\alpha\nabla_{\gamma}G\left(\gamma_{k}\right)d_{k}\right\Vert _{1}-c_{G}\alpha^{2}\left\Vert d_{k}\right\Vert ^{2}\leq\left\Vert G\left(\gamma_{k}+\alpha d_{k}\right)\right\Vert _{1}\leq\left\Vert G\left(\gamma_{k}\right)+\alpha\nabla_{\gamma}G\left(\gamma_{k}\right)d_{k}\right\Vert _{1}+c_{G}\alpha^{2}\left\Vert d_{k}\right\Vert ^{2}$. Let $C\equiv c_{Q}+\mu c_{G}$. Then, by $\phi_{1}\left(\gamma;\mu\right)\equiv Q\left(\gamma\right)+\mu\left\Vert G\left(\gamma\right)\right\Vert _{1}$,

\begin{eqnarray*}
 &  & \phi_{1}\left(\gamma_{k}+\alpha d_{k};\mu\right)-\phi_{1}\left(\gamma_{k};\mu\right)\\
 & = & \left[Q\left(\gamma_{k}+\alpha d_{k}\right)-Q\left(\gamma_{k}\right)\right]+\mu\left[\left\Vert G\left(\gamma_{k}+\alpha d_{k}\right)\right\Vert _{1}-\left\Vert G\left(\gamma_{k}\right)\right\Vert _{1}\right]\\
 & \leq & \alpha\nabla_{\gamma}Q\left(\gamma_{k}\right)d_{k}+\mu\left[\left\Vert G\left(\gamma_{k}\right)+\alpha\left(\nabla_{\gamma}G\left(\gamma_{k}\right)d_{k}\right)\right\Vert _{1}-\left\Vert G\left(\gamma_{k}\right)\right\Vert _{1}\right]+(c_{Q}+\mu c_{G})\alpha^{2}\left\Vert d_{k}\right\Vert ^{2}\\
 & = & \alpha\nabla_{\gamma}Q\left(\gamma_{k}\right)d_{k}-\alpha\mu\left\Vert G\left(\gamma_{k}\right)\right\Vert _{1}+C\alpha^{2}\left\Vert d_{k}\right\Vert ^{2}\ \left(\because(\ref{eq:G_linearize_norm})\right)\\
 & = & \alpha\Delta_{k}+C\alpha^{2}\left\Vert d_{k}\right\Vert ^{2}
\end{eqnarray*}

Similarly, we can derive $\phi_{1}\left(\gamma_{k}+\alpha d_{k};\mu\right)-\phi_{1}\left(\gamma_{k};\mu\right)\geq\alpha\Delta_{k}-C\alpha^{2}\left\Vert d_{k}\right\Vert ^{2}$.
\end{proof}

\subsubsection*{Proof of Proposition \ref{prop:descent_direction}}
\begin{proof}
Lemma \ref{lem: ineq_merit_func} implies that there exists $C>0$ such that $\Delta_{k}-C\alpha\left\Vert d_{k}\right\Vert ^{2}\leq\frac{\phi_{1}\left(\gamma_{k}+\alpha d_{k};\mu\right)-\phi_{1}\left(\gamma_{k};\mu\right)}{\alpha}\leq\Delta_{k}+C\alpha\left\Vert d_{k}\right\Vert ^{2}$. By letting $\alpha\rightarrow0$, we have $D\left(\phi_{1}\left(\gamma_{k};\mu\right);d_{k}\right)\equiv\lim_{\alpha\rightarrow0}\frac{\phi_{1}\left(\gamma_{k}+\alpha d_{k};\mu\right)-\phi_{1}\left(\gamma_{k};\mu\right)}{\alpha}=\Delta_{k}$. 
\end{proof}

\subsection{Proof of propositions related to $\Gamma_{FOC}$, $\Gamma_{KKT,nonsingular}$, and $\Gamma_{seq}$}

\subsubsection*{Proof of Proposition \ref{prop:fixed-points-KKT-seq}}
\begin{proof}
(a). Let $\widetilde{\gamma}\equiv\left(\widetilde{\theta},\widetilde{Y}\right)\in\Gamma_{seq}$ be a fixed point of Algorithm \ref{alg:general_seq}. Then, because $\widetilde{\theta}=\arg\min_{\theta}Q\left(\theta,\Upsilon\left(\theta;\widetilde{\gamma}\right)\right)$, $\nabla_{\theta}\widetilde{Q}\left(\widetilde{\theta};\widetilde{\gamma}\right)=0$ holds.

Then, $\nabla_{\theta}\Upsilon\left(\theta;\gamma=\left(\theta,Y\right)\right)=-\left(\nabla_{Y}G\left(Y;\theta\right)\right)^{-1}\left(\nabla_{\theta}G\left(Y;\theta\right)\right)$ by Assumption \ref{as:mapping_cdns}(c) and $\nabla_{\theta}\widetilde{Q}\left(\widetilde{\theta};\widetilde{\gamma}\right)=\nabla_{\theta}Q\left(\widetilde{\theta},\widetilde{Y}\right)+\left(\nabla_{Y}Q\left(\widetilde{\theta},\widetilde{Y}\right)\right)\left(\nabla_{\theta}\Upsilon\left(\widetilde{\theta};\widetilde{\gamma}\right)\right)$ imply:

\[
\nabla_{\theta}Q\left(\widetilde{\theta},\widetilde{Y}\right)-\left(\nabla_{Y}Q\left(\widetilde{\theta},\widetilde{Y}\right)\right)\left(\nabla_{Y}G\left(\widetilde{Y};\widetilde{\theta}\right)\right)^{-1}\left(\nabla_{\theta}G\left(\widetilde{Y};\widetilde{\theta}\right)\right)=0
\]

In addition, $\Upsilon\left(\widetilde{\theta};\widetilde{\gamma}\right)=\widetilde{Y}$ and Assumption \ref{as:mapping_cdns}(a) imply $G\left(\widetilde{Y};\widetilde{\theta}\right)=0$.

Therefore, the discussion above implies that there exists $\lambda^{T}=\left(\nabla_{Y}Q\left(\widetilde{\gamma}\right)\right)\left(\nabla_{Y}G\left(\widetilde{\gamma}\right)\right)^{-1}$ such that $\nabla_{\theta}Q\left(\widetilde{\gamma}\right)-\lambda^{T}\nabla_{\theta}G\left(\widetilde{\gamma}\right)=0$ and $\nabla_{Y}Q\left(\widetilde{\gamma}\right)-\lambda^{T}\nabla_{Y}G\left(\widetilde{\gamma}\right)=0$. Hence, $\widetilde{\gamma}\equiv\left(\widetilde{\theta},\widetilde{Y}\right)\in\Gamma_{KKT}$ holds.

(b). Let $\widetilde{\gamma}\equiv\left(\widetilde{\theta},\widetilde{Y}\right)\in\Gamma_{KKT}$. Then, there exists $\lambda$ such that $\nabla_{\theta}Q\left(\widetilde{\gamma}\right)-\lambda^{T}\nabla_{\theta}G\left(\widetilde{\gamma}\right)=0$, $\nabla_{Y}Q\left(\widetilde{\gamma}\right)-\lambda^{T}\nabla_{Y}G\left(\widetilde{\gamma}\right)=0$, and $G\left(\widetilde{\gamma}\right)=0$.

Then,

\begin{eqnarray*}
0 & = & \nabla_{\theta}Q\left(\widetilde{\gamma}\right)-\left(\nabla_{Y}Q\left(\widetilde{\theta},\widetilde{Y}\right)\right)\left(\nabla_{Y}G\left(\widetilde{Y};\widetilde{\theta}\right)\right)^{-1}\left(\nabla_{\theta}G\left(\widetilde{Y};\widetilde{\theta}\right)\right)\\
 & = & \nabla_{\theta}\widetilde{Q}\left(\widetilde{\theta};\widetilde{\gamma}\right)\ \left(\because\text{Assumption \ref{as:mapping_cdns}(c)}\right)
\end{eqnarray*}

Because $\nabla_{\theta}\widetilde{Q}\left(\widetilde{\theta};\widetilde{\gamma}\right)=0$ and $\nabla_{\theta\theta^{\prime}}\widetilde{Q}\left(\widetilde{\theta};\widetilde{\gamma}\right)$ is positive definite, $\widetilde{\theta}=\arg\min_{\theta}\widetilde{Q}\left(\theta;\widetilde{\gamma}\right)$ holds. Because $G\left(\widetilde{\gamma}\right)=0$ implies $\widetilde{Y}=\Upsilon\left(\widetilde{\theta};\widetilde{\gamma}\right)$ by Assumption \ref{as:mapping_cdns}(a), $\widetilde{\gamma}\equiv\left(\widetilde{\theta},\widetilde{Y}\right)\in\Gamma_{KKT}$ is a fixed point of Algorithm \ref{alg:general_seq}. Hence, $\widetilde{\gamma}\in\Gamma_{seq}$ holds.

(c). By setting $\widetilde{\gamma}=\widehat{\gamma}\in\Gamma_{KKT}$ in the proof of (b), an analogous argument applies, yielding $\widehat{\gamma}=\widetilde{\gamma}\in\Gamma_{seq}$. 
\end{proof}

\subsubsection*{Proof of Proposition \ref{prop:FOC-KKT}}
\begin{proof}
(a). Because $\gamma\in\Gamma_{KKT,nonsingular}$ and the corresponding Lagrange multiplier $\lambda_{*}$ satisfy $\nabla_{\gamma}Q\left(\gamma\right)-\lambda_{*}^{T}\left(\nabla_{\gamma}G\left(\gamma\right)\right)=0$, $\nabla_{Y}Q\left(\gamma\right)-\lambda_{*}^{T}\left(\nabla_{Y}G\left(\gamma\right)\right)=0$ holds. Because $\nabla_{Y}G\left(\gamma\right)$ is nonsingular, $\lambda_{*}^{T}=\left(\nabla_{Y}Q\left(\gamma\right)\right)\left(\nabla_{Y}G\left(\gamma\right)\right)^{-1}$ holds. Then, using $\nabla_{\theta}Q\left(\gamma\right)-\lambda_{*}^{T}\left(\nabla_{\theta}G\left(\gamma\right)\right)=0$, we obtain $\nabla_{\theta}Q\left(\theta,Y\right)-\left(\nabla_{Y}Q\left(\theta,Y\right)\right)\left(\nabla_{Y}G\left(Y;\theta\right)\right)^{-1}\left(\nabla_{\theta}G\left(Y;\theta\right)\right)=0$.

(b). It suffices to show that $\Gamma_{FOC}=\Gamma_{KKT,nonsingular}$.

First, we show $\Gamma_{FOC}\subset\Gamma_{KKT,nonsingular}$. Let $\gamma\in\Gamma_{FOC}$. Then, $\nabla_{Y}G\left(\gamma\right)$ is nonsingular, which implies that the rows of $\nabla_{Y}G\left(\gamma\right)$ are linearly independent.

Here, let $\sum_{i=1}^{n_{Y}}c_{i}\left(\nabla_{\gamma}G_{i}\left(\gamma\right)\right)=0$. Then, $\sum_{i=1}^{n_{Y}}c_{i}\left(\nabla_{\theta}G_{i}\left(\gamma\right)\right)=0$ and $\sum_{i=1}^{n_{Y}}c_{i}\left(\nabla_{Y}G_{i}\left(\gamma\right)\right)=0$ hold. The linear independence of the rows of $\nabla_{Y}G\left(\gamma\right)$ implies $c_{i}=0\ \forall i$. Note that this is consistent with $\sum_{i=1}^{n_{Y}}c_{i}\left(\nabla_{\theta}G_{i}\left(\gamma\right)\right)=0$. Hence, the rows of $\nabla_{\gamma}G\left(\gamma\right)$ are linearly independent. Because $\gamma\in\Gamma_{FOC}$ satisfies $\nabla_{\theta}Q\left(\theta,Y\right)-\left(\nabla_{Y}Q\left(\theta,Y\right)\right)\left(\nabla_{Y}G\left(Y;\theta\right)\right)^{-1}\left(\nabla_{\theta}G\left(Y;\theta\right)\right)=0$ and $G\left(\gamma\right)=0$, $\gamma\in\Gamma_{FOC}\Rightarrow\gamma\in\Gamma_{KKT,nonsingular}$ holds.

Next, we show $\Gamma_{KKT,nonsingular}\subset\Gamma_{FOC}$. Let $\gamma\in\Gamma_{KKT,nonsingular}$. Then, $\nabla_{\theta}Q\left(\theta,Y\right)-\left(\nabla_{Y}Q\left(\theta,Y\right)\right)\left(\nabla_{Y}G\left(Y;\theta\right)\right)^{-1}\left(\nabla_{\theta}G\left(Y;\theta\right)\right)=0$ and $G\left(\gamma\right)=0$ hold as shown in (a), and consequently $\gamma\in\Gamma_{FOC}$ holds. 
\end{proof}

\subsubsection*{Proof of Proposition \ref{prop:FOC}}

This follows directly from Proposition \ref{prop:FOC-KKT}.

\subsection{Zero Jacobian property}

\subsubsection{Proof that the SLC mapping satisfies Assumption \ref{as:mapping_cdns}}\label{subsec:Proof-SLC-ZJP}
\begin{proof}
(a).

\begin{eqnarray*}
\Upsilon\left(\theta,\gamma=\left(\theta,Y\right)\right)=Y & \Leftrightarrow & Y-\left(\nabla_{Y}G\left(Y;\theta\right)\right)^{-1}\left(G\left(Y;\theta\right)+\left(\nabla_{\theta}G\left(Y;\theta\right)\right)\left(\theta-\theta\right)\right)=Y\\
 & \Leftrightarrow & G\left(Y;\theta\right)=0\ \left(\because\nabla_{Y}G\left(Y;\theta\right)\text{ is nonsingular}\right)
\end{eqnarray*}

(b). First,

\begin{eqnarray*}
 &  & \frac{\partial\Upsilon}{\partial\theta_{0}}\left(\theta,\gamma=\left(\theta_{0},Y_{0}\right)\right)\\
 & = & \frac{\partial}{\partial\theta_{0}}\left[Y_{0}-\left(\nabla_{Y}G\left(Y_{0};\theta_{0}\right)\right)^{-1}\left(G\left(Y_{0};\theta_{0}\right)+\left(\nabla_{\theta}G\left(Y_{0};\theta_{0}\right)\right)\left(\theta-\theta_{0}\right)\right)\right]\\
 & = & -\frac{\partial\left(\left(\nabla_{Y}G\left(Y_{0};\theta_{0}\right)\right)^{-1}\right)}{\partial\theta_{0}}\left[G\left(Y_{0};\theta_{0}\right)+\left(\nabla_{\theta}G\left(Y_{0};\theta_{0}\right)\right)\left(\theta-\theta_{0}\right)\right]\\
 &  & -\left(\nabla_{Y}G\left(Y_{0};\theta_{0}\right)\right)^{-1}\left[\nabla_{\theta}G\left(Y_{0};\theta_{0}\right)+\left(\nabla_{\theta\theta^{\prime}}G_{i}\left(Y_{0};\theta_{0}\right)\left(\theta-\theta_{0}\right)\right)_{i=1,\cdots,n_{Y}}+\left(\nabla_{\theta}G\left(Y_{0};\theta_{0}\right)\right)\cdot(-1)\right].
\end{eqnarray*}
Hence, $G\left(Y;\theta\right)=0$ implies $\left.\frac{\partial\Upsilon}{\partial\theta_{0}}\left(\theta,\gamma=\left(\theta_{0},Y\right)\right)\right|_{\theta_{0}=\theta}=0$.

Next,

\begin{eqnarray*}
 &  & \frac{\partial\Upsilon}{\partial Y_{0}}\left(\theta,\gamma=\left(\theta_{0},Y_{0}\right)\right)\\
 & = & \frac{\partial}{\partial Y_{0}}\left[Y_{0}-\left(\nabla_{Y}G\left(Y_{0};\theta_{0}\right)\right)^{-1}\left(G\left(Y_{0};\theta_{0}\right)+\left(\nabla_{\theta}G\left(Y_{0};\theta_{0}\right)\right)\left(\theta-\theta_{0}\right)\right)\right]\\
 & = & I-\frac{\partial\left(\left(\nabla_{Y}G\left(Y_{0};\theta_{0}\right)\right)^{-1}\right)}{\partial Y_{0}}\left[G\left(Y_{0};\theta_{0}\right)+\left(\nabla_{\theta}G\left(Y_{0};\theta_{0}\right)\right)\left(\theta-\theta_{0}\right)\right]\\
 &  & -\left(\nabla_{Y}G\left(Y_{0};\theta_{0}\right)\right)^{-1}\left[\nabla_{Y}G\left(Y_{0};\theta_{0}\right)+\left(\nabla_{Y\theta^{\prime}}G_{i}\left(Y_{0};\theta_{0}\right)\right)_{i=1,\cdots,n_{Y}}\left(\theta-\theta_{0}\right)\right].
\end{eqnarray*}
Hence, $G\left(Y;\theta\right)=0$ implies $\left.\frac{\partial\Upsilon}{\partial Y_{0}}\left(\theta,\gamma=\left(\theta,Y_{0}\right)\right)\right|_{Y_{0}=Y}=0$.

(c). Because $\Upsilon\left(\theta,\gamma=\left(\theta_{0},Y_{0}\right)\right)=Y_{0}-\left(\nabla_{Y}G\left(Y_{0};\theta_{0}\right)\right)^{-1}\left(G\left(Y_{0};\theta_{0}\right)+\left(\nabla_{\theta}G\left(Y_{0};\theta_{0}\right)\right)\left(\theta-\theta_{0}\right)\right)$,

$\nabla_{\theta}\Upsilon\left(\theta;\gamma=\left(\theta,Y\right)\right)=-\left(\nabla_{Y}G\left(Y;\theta\right)\right)^{-1}\left(\nabla_{\theta}G\left(Y;\theta\right)\right)$. 
\end{proof}

\subsubsection{Proof that the NPL mapping with \textquotedblleft Zero Jacobian property\textquotedblright{} satisfies Assumption \ref{as:mapping_cdns}}\label{subsec:Proof-NPL-ZJP}
\begin{proof}
(a).

\begin{eqnarray*}
\Upsilon\left(\theta,\gamma=\left(\theta,Y\right)\right)=Y & \Leftrightarrow & \Phi\left(Y;\theta\right)=Y\\
 & \Leftrightarrow & G\left(Y;\theta\right)=0
\end{eqnarray*}

(b). First, $\frac{\partial\Upsilon}{\partial\theta_{0}}\left(\theta,\gamma=\left(\theta_{0},Y_{0}\right)\right)=\frac{\partial\Phi\left(Y_{0};\theta_{0}\right)}{\partial\theta_{0}}=0$ holds.

In addition, $G\left(Y;\theta\right)=0$ implies $\left.\frac{\partial\Upsilon}{\partial Y_{0}}\left(\theta,\gamma=\left(\theta_{0},Y_{0}\right)\right)\right|_{\theta_{0}=\theta,Y_{0}=Y}=\nabla_{Y}\Phi(Y;\theta_{0})=0$, because $\Phi$ satisfies $\Phi\left(Y;\theta\right)=Y\Rightarrow\nabla_{Y}\Phi\left(Y;\theta\right)=0$.

(c).

\begin{eqnarray*}
\nabla_{\theta}\Upsilon\left(\theta;\gamma=\left(\theta,Y\right)\right) & = & \nabla_{\theta}\Phi\left(Y;\theta\right)\\
 & = & -\left(\nabla_{\theta}G\left(Y;\theta\right)\right)\ \left(\because G\left(Y;\theta\right)=Y-\Phi\left(Y;\theta\right)\right)
\end{eqnarray*}

Because $G\left(Y;\theta\right)=0\Leftrightarrow\Phi\left(Y;\theta\right)=Y\text{ implies }\nabla_{Y}\Phi\left(Y;\theta\right)=0\Leftrightarrow\nabla_{Y}G\left(Y;\theta\right)=I$, $G\left(Y;\theta\right)=0$ implies $\nabla_{\theta}\Upsilon\left(\theta;\gamma=\left(\theta,Y\right)\right)=-\left(\nabla_{Y}G\left(Y;\theta\right)\right)^{-1}\left(\nabla_{\theta}G\left(Y;\theta\right)\right)$. 
\end{proof}

\subsection{Convexity}

\subsubsection*{Proof of Lemma \ref{lem:difference-Hessian-NFXP}}
\begin{proof}
First, by $\nabla_{\theta}\widetilde{Q}\left(\theta;\gamma\right)=\nabla_{\theta}Q\left(\theta,\Upsilon\left(\theta,\gamma\right)\right)+\left(\nabla_{Y}Q\left(\theta,\Upsilon\left(\theta,\gamma\right)\right)\right)\left(\nabla_{\theta}\Upsilon\left(\theta,\gamma\right)\right)$, 
\begin{eqnarray*}
\nabla_{\theta\theta^{\prime}}\widetilde{Q}\left(\theta;\gamma\right) & = & \left(\nabla_{\theta\theta^{\prime}}Q\left(\theta,\Upsilon\left(\theta,\gamma\right)\right)\right)+\left(\nabla_{\theta Y^{\prime}}Q\left(\theta,\Upsilon\left(\theta,\gamma\right)\right)\right)\left(\nabla_{\theta}\Upsilon\left(\theta,\gamma\right)\right)+\\
 &  & \left(\nabla_{\theta Y^{\prime}}Q\left(\theta,\Upsilon\left(\theta,\gamma\right)\right)\right)\left(\nabla_{\theta}\Upsilon\left(\theta;\gamma\right)\right)+\sum_{i=1}^{n_{Y}}\frac{\partial Q\left(\theta,\Upsilon\left(\theta,\gamma\right)\right)}{\partial Y_{i}}\left(\nabla_{\theta\theta^{\prime}}\Upsilon_{i}\left(\theta;\gamma\right)\right)+\\
 &  & \left(\nabla_{\theta^{\prime}}\Upsilon\left(\theta;\gamma\right)\right)\left(\nabla_{YY^{\prime}}Q\left(\theta,\Upsilon\left(\theta,\gamma\right)\right)\right)\left(\nabla_{\theta}\Upsilon\left(\theta;\gamma\right)\right).
\end{eqnarray*}

Next,

\begin{eqnarray*}
\nabla_{\theta\theta^{\prime}}Q_{NFXP}\left(\theta\right) & = & \frac{d}{d\theta}\left[\nabla_{\theta^{\prime}}Q\left(\theta,\mathcal{Y}\left(\theta\right)\right)+\left(\nabla_{\theta^{\prime}}\mathcal{Y}\left(\theta\right)\right)\left(\nabla_{Y^{\prime}}Q\left(\theta,\mathcal{Y}\left(\theta\right)\right)\right)\right]\\
 & = & \left(\nabla_{\theta\theta^{\prime}}Q\left(\theta,\mathcal{Y}\left(\theta\right)\right)\right)+\left(\nabla_{\theta Y^{\prime}}Q\left(\theta,\mathcal{Y}\left(\theta\right)\right)\right)\left(\nabla_{\theta}\mathcal{Y}\left(\theta\right)\right)+\\
 &  & \left(\nabla_{\theta Y^{\prime}}Q\left(\theta,\mathcal{Y}\left(\theta\right)\right)\right)\left(\nabla_{\theta}\mathcal{Y}\left(\theta\right)\right)+\sum_{i=1}^{n_{Y}}\frac{\partial Q\left(\theta,\mathcal{Y}\left(\theta\right)\right)}{\partial Y_{i}}\left(\nabla_{\theta\theta^{\prime}}\mathcal{Y}_{i}\left(\theta\right)\right)+\\
 &  & \left(\nabla_{\theta^{\prime}}\mathcal{Y}\left(\theta\right)\right)\left(\nabla_{YY^{\prime}}Q\left(\theta,\mathcal{Y}\left(\theta\right)\right)\right)\left(\nabla_{\theta}\mathcal{Y}\left(\theta\right)\right).
\end{eqnarray*}

Here,

\begin{eqnarray*}
 &  & \left(\nabla_{\theta^{\prime}}\Upsilon\left(\widehat{\theta};\widehat{\gamma}\right)\right)\left(\nabla_{YY^{\prime}}Q\left(\widehat{\theta},\widehat{Y}\right)\right)\left(\nabla_{\theta}\Upsilon\left(\widehat{\theta};\widehat{\gamma}\right)\right)-\left(\nabla_{\theta^{\prime}}\mathcal{Y}\left(\widehat{\theta}\right)\right)\left(\nabla_{YY^{\prime}}Q\left(\widehat{\theta},\widehat{Y}\right)\right)\left(\nabla_{\theta}\mathcal{Y}\left(\widehat{\theta}\right)\right)\\
 & = & \left(\nabla_{\theta^{\prime}}\Upsilon\left(\widehat{\theta};\widehat{\gamma}\right)+\nabla_{\theta^{\prime}}\mathcal{Y}\left(\widehat{\theta}\right)\right)\left(\nabla_{YY^{\prime}}Q\left(\widehat{\theta},\widehat{Y}\right)\right)\left(\nabla_{\theta}\Upsilon\left(\widehat{\theta};\widehat{\gamma}\right)-\nabla_{\theta}\mathcal{Y}\left(\widehat{\theta}\right)\right)
\end{eqnarray*}
holds because the Hessian $\nabla_{YY^{\prime}}Q\left(\widehat{\theta},\widehat{Y}\right)$ is a symmetric matrix. Then, under $\mathcal{Y}\left(\widehat{\theta}\right)=\Upsilon\left(\widehat{\theta};\widehat{\gamma}\right)=\widehat{Y}$,

\begin{eqnarray*}
 &  & \nabla_{\theta\theta^{\prime}}\widetilde{Q}\left(\widehat{\theta};\widehat{\gamma}\right)-\nabla_{\theta\theta^{\prime}}Q_{NFXP}\left(\widehat{\theta}\right)\\
 & = & \sum_{i=1}^{n_{Y}}\frac{\partial Q\left(\widehat{\theta},\widehat{Y}\right)}{\partial Y_{i}}\left(\nabla_{\theta\theta^{\prime}}\Upsilon_{i}(\widehat{\theta};\widehat{\gamma})-\nabla_{\theta\theta^{\prime}}\mathcal{Y}_{i}(\widehat{\theta})\right)+\\
 &  & 2\left(\nabla_{\theta Y^{\prime}}Q\left(\widehat{\theta},\widehat{Y}\right)\right)\left(\nabla_{\theta}\Upsilon(\widehat{\theta};\widehat{\gamma})-\nabla_{\theta}\mathcal{Y}(\widehat{\theta})\right)+\\
 &  & \left(\nabla_{\theta^{\prime}}\Upsilon\left(\widehat{\theta};\widehat{\gamma}\right)+\nabla_{\theta^{\prime}}\mathcal{Y}\left(\widehat{\theta}\right)\right)\left(\nabla_{YY^{\prime}}Q\left(\widehat{\theta},\widehat{Y}\right)\right)\left(\nabla_{\theta}\Upsilon\left(\widehat{\theta};\widehat{\gamma}\right)-\nabla_{\theta}\mathcal{Y}\left(\widehat{\theta}\right)\right).
\end{eqnarray*}

Therefore,

\begin{eqnarray*}
\nabla_{\theta\theta^{\prime}}\widetilde{Q}\left(\widehat{\theta};\widehat{\gamma}\right)-\nabla_{\theta\theta^{\prime}}Q_{NFXP}\left(\widehat{\theta}\right) & = & \sum_{i=1}^{n_{Y}}\frac{\partial Q\left(\widehat{\theta},\widehat{Y}\right)}{\partial Y_{i}}\left(\nabla_{\theta\theta^{\prime}}\Upsilon_{i}(\widehat{\theta};\widehat{\gamma})-\nabla_{\theta\theta^{\prime}}\mathcal{Y}_{i}(\widehat{\theta})\right)
\end{eqnarray*}
holds by $\nabla_{\theta}\mathcal{Y}\left(\widehat{\theta}\right)=-\left(\nabla_{Y}G\left(\widehat{Y};\widehat{\theta}\right)\right)^{-1}\left(\nabla_{\theta}G\left(\widehat{Y};\widehat{\theta}\right)\right)=\nabla_{\theta}\Upsilon\left(\widehat{\theta};\widehat{\gamma}\right)$ under Assumption \ref{as:mapping_cdns}. 
\end{proof}

\subsection{Asymptotic property (Proof of Proposition \ref{prop: asy_property_consistent_initial_values})}
\begin{proof}
(1). We show the claim by induction.\footnote{The strategy of the proof of (1) is the same as that for Theorem 1(1) of \citet{dearing2024efficient}.} First, we have uniform continuity of $\widetilde{Q^{*}}\left(\theta;\gamma\right)\equiv Q^{*}\left(\theta,\Upsilon\left(\theta,\gamma\right)\right)$, and $\widetilde{Q}\left(\theta;\gamma\right)$ converges almost surely and uniformly in $\left(\theta,\gamma\right)\in\Theta\times\left(\Theta\times\mathcal{Y}\right)$ to $\widetilde{Q^{*}}\left(\theta;\gamma\right)$. Here, suppose $\gamma_{k-1}$ converges almost surely to $\gamma^{*}$. By Lemma 24.1 of \citet{gourieroux1995statistics}, $\widetilde{Q}\left(\theta;\gamma_{k-1}\right)$ converges almost surely and uniformly in $\theta\in\Theta$ to $\widetilde{Q^{*}}\left(\theta;\gamma^{*}\right)$. Because $\theta^{*}$ uniquely minimizes $\widetilde{Q^{*}}\left(\theta;\gamma^{*}\right)$ on $\Theta$, Property 24.2 of \citet{gourieroux1995statistics} implies that $\theta_{k}$ converges almost surely to $\theta^{*}$. Continuity of $\Upsilon\left(\theta,\gamma\right)$ and the Mann--Wald theorem then give almost sure convergence of $Y_{k}$ to $Y^{*}$.

(2).\footnote{The proof of (2) is motivated by that for proving Proposition 2 of \citet{lee2015computationally}, which shows the asymptotic equivalence of the NFXP estimator and the sequential algorithm (ABLP) in the context of the static BLP model.}

By $\theta_{k}=\arg\min_{\theta\in\Theta}\widetilde{Q}\left(\theta;\gamma_{k-1}\right)$, $0=\nabla_{\theta}\widetilde{Q}\left(\theta_{k};\gamma_{k-1}\right)$ holds. By the mean value theorem, there exist $c_{1},c_{2}\in[0,1]$ such that $\overline{\theta}=c_{1}\theta^{*}+(1-c_{1})\theta_{k}$ and $\overline{\gamma}=c_{2}\gamma^{*}+(1-c_{2})\gamma_{k-1}$, and 

\begin{eqnarray*}
 &  & 0\\
 & = & \nabla_{\theta}\widetilde{Q}\left(\theta_{k};\gamma_{k-1}\right)\\
 & = & \nabla_{\theta}\widetilde{Q}\left(\theta^{*};\gamma_{k-1}\right)+\left(\nabla_{\theta\theta^{\prime}}\widetilde{Q}\left(\overline{\theta};\gamma_{k-1}\right)\right)\left(\theta_{k}-\theta^{*}\right)\\
 & = & \nabla_{\theta}\widetilde{Q}\left(\theta^{*};\gamma^{*}\right)+\left(\nabla_{\theta\gamma^{\prime}}\widetilde{Q}\left(\theta^{*};\gamma^{*}\right)\right)\left(\gamma_{k-1}-\gamma^{*}\right)+\\
 &  & \frac{1}{2}\left(\left(\gamma_{k-1}-\gamma^{*}\right)^{\prime}\left(\nabla_{\theta_{j}\gamma\gamma^{\prime}}\widetilde{Q}\left(\theta^{*};\overline{\gamma}\right)\right)\left(\gamma_{k-1}-\gamma^{*}\right)\right)_{j=1,\cdots,n_{\theta}}+\left(\nabla_{\theta\theta^{\prime}}\widetilde{Q}\left(\overline{\theta};\gamma_{k-1}\right)\right)\left(\theta_{k}-\theta^{*}\right).
\end{eqnarray*}

Then, we have:

\begin{eqnarray*}
\sqrt{N}\left(\theta_{k}-\theta^{*}\right) & = & -\left(\nabla_{\theta\theta^{\prime}}\widetilde{Q}\left(\overline{\theta};\gamma_{k-1}\right)\right)^{-1}\left[\sqrt{N}\left(\nabla_{\theta}\widetilde{Q}\left(\theta^{*};\gamma^{*}\right)\right)\right]-\\
 &  & \left(\nabla_{\theta\theta^{\prime}}\widetilde{Q}\left(\overline{\theta};\gamma_{k-1}\right)\right)^{-1}\left[\left(\nabla_{\theta\gamma^{\prime}}\widetilde{Q}\left(\theta^{*};\gamma^{*}\right)\right)\sqrt{N}\left(\gamma_{k-1}-\gamma^{*}\right)\right]-\\
 &  & \sqrt{N}\left(\nabla_{\theta\theta^{\prime}}\widetilde{Q}\left(\overline{\theta};\gamma_{k-1}\right)\right)^{-1}\frac{1}{2}\left(\left(\gamma_{k-1}-\gamma^{*}\right)^{\prime}\left(\nabla_{\theta_{j}\gamma\gamma^{\prime}}\widetilde{Q}\left(\theta^{*};\overline{\gamma}\right)\right)\left(\gamma_{k-1}-\gamma^{*}\right)\right)_{j=1,\cdots,n_{\theta}}.
\end{eqnarray*}

Hence, it suffices to show (a) $-\left(\nabla_{\theta\theta^{\prime}}\widetilde{Q}\left(\overline{\theta};\gamma_{k-1}\right)\right)^{-1}\left[\sqrt{N}\left(\nabla_{\theta}\widetilde{Q}\left(\theta^{*};\gamma^{*}\right)\right)\right]\rightarrow_{d}N\left(0,\Sigma_{\theta}\right)$, (b) \textbf{$\left(\nabla_{\theta\theta^{\prime}}\widetilde{Q}\left(\overline{\theta};\gamma_{k-1}\right)\right)^{-1}\left[\left(\nabla_{\theta\gamma^{\prime}}\widetilde{Q}\left(\theta^{*};\gamma^{*}\right)\right)\sqrt{N}\left(\gamma_{k-1}-\gamma^{*}\right)\right]\rightarrow_{p}0$}, and (c) $\sqrt{N}\left(\nabla_{\theta\theta^{\prime}}\widetilde{Q}\left(\overline{\theta};\gamma_{k-1}\right)\right)^{-1}\left(\left(\gamma_{k-1}-\gamma^{*}\right)^{\prime}\left(\nabla_{\theta_{j}\gamma\gamma^{\prime}}\widetilde{Q}\left(\theta^{*};\overline{\gamma}\right)\right)\left(\gamma_{k-1}-\gamma^{*}\right)\right)_{j=1,\cdots,n_{\theta}}\rightarrow_{p}0$.

\textbf{(a). Proof of $-\left(\nabla_{\theta\theta^{\prime}}\widetilde{Q}\left(\overline{\theta};\gamma_{k-1}\right)\right)^{-1}\left[\sqrt{N}\left(\nabla_{\theta}\widetilde{Q}\left(\theta^{*};\gamma^{*}\right)\right)\right]\rightarrow_{d}N\left(0,\Sigma_{\theta}\right)$:}

First, $\widehat{\theta}$ satisfies $0=\nabla_{\theta}Q_{NFXP}\left(\widehat{\theta}\right)$. Then, when $\widehat{\theta}$ and $\theta^{*}$ are sufficiently close, the following holds by the mean value theorem:

\begin{eqnarray*}
0 & = & \nabla_{\theta}Q_{NFXP}\left(\theta^{*}\right)+\left(\nabla_{\theta\theta^{\prime}}Q_{NFXP}\left(\widetilde{\theta}\right)\right)\left(\widehat{\theta}-\theta^{*}\right)
\end{eqnarray*}
where $\widetilde{\theta}=c\widehat{\theta}+(1-c)\theta^{*}$ for some $c\in[0,1]$. Then,

\begin{eqnarray}
\sqrt{N}\left(\widehat{\theta}-\theta^{*}\right) & = & -\left(\nabla_{\theta\theta^{\prime}}Q_{NFXP}\left(\widetilde{\theta}\right)\right)^{-1}\left(\sqrt{N}\nabla_{\theta}Q_{NFXP}\left(\theta^{*}\right)\right).\label{eq:asy_normal_NFXP}
\end{eqnarray}

By Lemma \ref{lem:difference-Hessian-NFXP}, 
\begin{eqnarray*}
\nabla_{\theta\theta^{\prime}}\widetilde{Q}\left(\widehat{\theta};\widehat{\gamma}\right)-\nabla_{\theta\theta^{\prime}}Q_{NFXP}\left(\widehat{\theta}\right) & = & \sum_{i=1}^{n_{Y}}\frac{\partial Q\left(\widehat{\theta},\widehat{Y}\right)}{\partial Y_{i}}\left(\nabla_{\theta\theta^{\prime}}\Upsilon_{i}(\widehat{\theta};\widehat{\gamma})-\nabla_{\theta\theta^{\prime}}\mathcal{Y}_{i}(\widehat{\theta})\right).
\end{eqnarray*}

By the discussion in Section \ref{subsec:dQ_dY}, $\nabla_{Y}Q\left(\widehat{\theta},\widehat{Y}\right)\rightarrow_{p}0$ holds. Because $\overline{\theta}=c_{1}\theta^{*}+(1-c_{1})\theta_{k}$ with $c_{1}\in[0,1]$ and $\theta_{k}\rightarrow_{p}\theta^{*}$, $\overline{\theta}\rightarrow_{p}\theta^{*}$ holds. Hence, by the boundedness of $\nabla_{\theta\theta^{\prime}}\Upsilon_{i}(\widehat{\theta};\widehat{\gamma})-\nabla_{\theta\theta^{\prime}}\mathcal{Y}_{i}(\widehat{\theta})$, $\nabla_{Y}Q\left(\widehat{\theta},\widehat{Y}\right)\rightarrow_{p}0$, $\overline{\theta}\rightarrow_{p}\theta^{*}$, and $\gamma_{k-1}\rightarrow_{p}\gamma^{*}$, we obtain:

\begin{eqnarray}
\nabla_{\theta\theta^{\prime}}\widetilde{Q}\left(\overline{\theta};\gamma_{k-1}\right)-\nabla_{\theta\theta^{\prime}}Q_{NFXP}\left(\widetilde{\theta}\right) & \rightarrow_{p} & 0.\label{eq:Hessian_asy_equiv}
\end{eqnarray}

In addition, the following holds by $G\left(Y^{*};\theta^{*}\right)=0$: 
\begin{eqnarray}
 &  & \nabla_{\theta}\widetilde{Q}\left(\theta^{*};\gamma^{*}\right)\nonumber \\
 & = & \nabla_{\theta}Q\left(\theta^{*},\Upsilon\left(\theta^{*};\gamma^{*}\right)\right)+\left(\nabla_{Y}Q\left(\theta^{*},\Upsilon\left(\theta^{*};\gamma^{*}\right)\right)\right)\left(\nabla_{\theta}\Upsilon\left(\theta^{*};\gamma^{*}\right)\right)\nonumber \\
 & = & \nabla_{\theta}Q\left(\theta^{*},\Upsilon\left(\theta^{*};\gamma^{*}\right)\right)-\nonumber \\
 &  & \left(\nabla_{Y}Q\left(\theta^{*},\Upsilon\left(\theta^{*};\gamma^{*}\right)\right)\right)\left(\nabla_{Y}G\left(Y^{*};\theta^{*}\right)\right)^{-1}\left(\nabla_{\theta}G\left(Y^{*};\theta^{*}\right)\right)\ \left(\because\text{Assumption}\ \text{\ref{as:mapping_cdns}}\right)\nonumber \\
 & = & \nabla_{\theta}Q_{NFXP}\left(\theta^{*}\right).\label{eq:seq_NFXP_FOC_equiv}
\end{eqnarray}

Hence, by $\sqrt{N}\nabla_{\theta}Q_{NFXP}\left(\theta^{*}\right)=O_{p}\left(1\right)$,

{\footnotesize
\begin{eqnarray*}
 &  & \sqrt{N}\left(\widehat{\theta}-\theta^{*}\right)+\left(\nabla_{\theta\theta^{\prime}}\widetilde{Q}\left(\overline{\theta};\gamma_{k-1}\right)\right)^{-1}\left[\sqrt{N}\left(\nabla_{\theta}\widetilde{Q}\left(\theta^{*};\gamma^{*}\right)\right)\right]\\
 & = & -\left(\nabla_{\theta\theta^{\prime}}Q_{NFXP}\left(\widetilde{\theta}\right)\right)^{-1}\left(\sqrt{N}\nabla_{\theta}Q_{NFXP}\left(\theta^{*}\right)\right)+\left(\nabla_{\theta\theta^{\prime}}\widetilde{Q}\left(\overline{\theta};\gamma_{k-1}\right)\right)^{-1}\left[\sqrt{N}\left(\nabla_{\theta}\widetilde{Q}\left(\theta^{*};\gamma^{*}\right)\right)\right]\ \left(\because(\ref{eq:asy_normal_NFXP})\right)\\
 & = & \left[\left(\nabla_{\theta\theta^{\prime}}\widetilde{Q}\left(\overline{\theta};\gamma_{k-1}\right)\right)^{-1}-\left(\nabla_{\theta\theta^{\prime}}Q_{NFXP}\left(\widetilde{\theta}\right)\right)^{-1}\right]\sqrt{N}\nabla_{\theta}Q_{NFXP}\left(\theta^{*}\right)+\\
 &  & \left(\nabla_{\theta\theta^{\prime}}\widetilde{Q}\left(\overline{\theta};\gamma_{k-1}\right)\right)^{-1}\sqrt{N}\left[\left(\nabla_{\theta}\widetilde{Q}\left(\theta^{*};\gamma^{*}\right)\right)-\nabla_{\theta}Q_{NFXP}\left(\theta^{*}\right)\right]\\
 & \rightarrow_{p} & 0\ \left(\because(\ref{eq:Hessian_asy_equiv}),(\ref{eq:seq_NFXP_FOC_equiv})\right).
\end{eqnarray*}
}{\footnotesize\par}

\textbf{(b). Proof of $\left(\nabla_{\theta\theta^{\prime}}\widetilde{Q}\left(\overline{\theta};\gamma_{k-1}\right)\right)^{-1}\left[\left(\nabla_{\theta\gamma^{\prime}}\widetilde{Q}\left(\theta^{*};\gamma^{*}\right)\right)\sqrt{N}\left(\gamma_{k-1}-\gamma^{*}\right)\right]\rightarrow_{p}0$:}

By (\ref{eq:d2Q_d2theta_gamma_at_sol}), $\nabla_{\theta\gamma^{\prime}}\widetilde{Q}\left(\theta^{*};\gamma^{*}\right)=\sum_{j=1}^{n_{Y}}\left(\nabla_{Y_{j}}Q\left(\theta^{*},Y^{*}\right)\right)\left(\nabla_{\theta\gamma^{\prime}}\Upsilon_{j}\left(\theta^{*},\gamma^{*}\right)\right)$ holds. By the discussion in Section \ref{subsec:dQ_dY}, $\nabla_{Y}Q\left(\theta^{*},Y^{*}\right)=O_{p}\left(N^{-\frac{1}{2}}\right)$ holds, which implies $\nabla_{\theta\gamma^{\prime}}\widetilde{Q}\left(\theta^{*},\gamma^{*}\right)=O_{p}\left(N^{-\frac{1}{2}}\right)$ by the boundedness of $\nabla_{\theta\gamma^{\prime}}\Upsilon_{j}\left(\theta^{*},\gamma^{*}\right)$. Hence, $\sqrt{N}\nabla_{\theta\gamma^{\prime}}\widetilde{Q}\left(\theta^{*},\gamma^{*}\right)=O_{p}(1)$ holds.

Therefore, $\gamma_{k-1}-\gamma^{*}=o_{p}(1)$ and the boundedness of $\left(\nabla_{\theta\theta^{\prime}}\widetilde{Q}\left(\overline{\theta};\gamma_{k-1}\right)\right)^{-1}$ imply:

\[
\left(\nabla_{\theta\theta^{\prime}}\widetilde{Q}\left(\overline{\theta};\gamma_{k-1}\right)\right)^{-1}\left[\sqrt{N}\left(\nabla_{\theta\gamma^{\prime}}\widetilde{Q}\left(\theta^{*},\gamma^{*}\right)\right)\right]\left(\gamma_{k-1}-\gamma^{*}\right)\rightarrow_{p}0.
\]

\textbf{(c). Proof of $\sqrt{N}\left(\nabla_{\theta\theta^{\prime}}\widetilde{Q}\left(\overline{\theta};\gamma_{k-1}\right)\right)^{-1}\left(\left(\gamma_{k-1}-\gamma^{*}\right)^{\prime}\left(\nabla_{\theta_{j}\gamma\gamma^{\prime}}\widetilde{Q}\left(\theta^{*};\overline{\gamma}\right)\right)\left(\gamma_{k-1}-\gamma^{*}\right)\right)_{j=1,\cdots,n_{\theta}}\rightarrow_{p}0$:}

By the boundedness of $\nabla_{\theta_{j}\gamma\gamma^{\prime}}\widetilde{Q}\left(\theta^{*};\overline{\gamma}\right)$ $(j=1,\cdots,n_{\theta})$ and the assumption $\gamma_{k-1}-\gamma^{*}=o_{p}\left(N^{-\frac{1}{4}}\right)$,

\[
\left(\left(\gamma_{k-1}-\gamma^{*}\right)^{\prime}\left(\nabla_{\theta_{j}\gamma\gamma^{\prime}}\widetilde{Q}\left(\theta^{*};\overline{\gamma}\right)\right)\left(\gamma_{k-1}-\gamma^{*}\right)\right)_{j=1,\cdots,n_{\theta}}=o_{p}\left(N^{-\frac{1}{2}}\right)
\]
holds. Hence, by the boundedness of $\left(\nabla_{\theta\theta^{\prime}}\widetilde{Q}\left(\overline{\theta};\gamma_{k-1}\right)\right)^{-1}$, we obtain $\sqrt{N}\left(\nabla_{\theta\theta^{\prime}}\widetilde{Q}\left(\overline{\theta};\gamma_{k-1}\right)\right)^{-1}\left(\left(\gamma_{k-1}-\gamma^{*}\right)^{\prime}\left(\nabla_{\theta_{j}\gamma\gamma^{\prime}}\widetilde{Q}\left(\theta^{*};\overline{\gamma}\right)\right)\left(\gamma_{k-1}-\gamma^{*}\right)\right)_{j=1,\cdots,n_{\theta}}=o_{p}(1)$.

(3). We show the statement by induction. First, $O_{p}\left(N^{-\frac{1}{2}}\left\Vert \gamma_{0}-\widehat{\gamma}\right\Vert \right)=O_{p}\left(N^{-(b+\frac{1}{2})}\right)$ and $\left\Vert \gamma_{0}-\widehat{\gamma}\right\Vert ^{2}=O_{p}\left(N^{-2b}\right)$ hold by the assumption $\gamma_{0}-\gamma^{*}=O_{p}\left(N^{-b}\right)$ (and $\widehat{\gamma}-\gamma^{*}=O_{p}\left(N^{-\frac{1}{2}}\right)$). Then, by Proposition \ref{prop: asy_property}(1), 
\begin{eqnarray*}
\gamma_{1}-\widehat{\gamma} & = & O_{p}\left(N^{-\frac{1}{2}}\left\Vert \gamma_{0}-\widehat{\gamma}\right\Vert +\left\Vert \gamma_{0}-\widehat{\gamma}\right\Vert ^{2}\right)\\
 & = & O_{p}\left(N^{-2b}\right)\ \left(\because-\frac{1}{2}-b\leq-2b\ \text{under}\ b\leq\frac{1}{2}\right)
\end{eqnarray*}
holds.

Next, assume $\gamma_{k}-\widehat{\gamma}=O_{p}\left(N^{-(k-1)/2-2b}\right)$ $\left(k\geq1\right)$. Then, $O_{p}\left(N^{-\frac{1}{2}}\left\Vert \gamma_{k}-\widehat{\gamma}\right\Vert \right)=O_{p}\left(N^{-k/2-2b}\right)$ and $O_{p}\left(\left\Vert \gamma_{k}-\widehat{\gamma}\right\Vert ^{2}\right)=O_{p}\left(N^{-(k-1)-4b}\right)$ hold, which imply:

\begin{eqnarray*}
\gamma_{k+1}-\widehat{\gamma} & = & O_{p}\left(N^{-\frac{1}{2}}\left\Vert \gamma_{k}-\widehat{\gamma}\right\Vert +\left\Vert \gamma_{k}-\widehat{\gamma}\right\Vert ^{2}\right)\\
 & = & O_{p}\left(N^{-k/2-2b}\right)\ \left(\because-k/2-2b\geq-(k-1)-4b\ \text{under}\ b>\frac{1}{4}\right).
\end{eqnarray*}

Hence, $\gamma_{k}-\widehat{\gamma}=O_{p}\left(N^{-(k-1)/2-2b}\right)\Rightarrow\gamma_{k+1}-\widehat{\gamma}=O_{p}\left(N^{-((k+1)-1)/2-2b}\right)$ holds, and the statement is proved. 
\end{proof}

\subsection{Stochastic $Y$ setting: Asymptotic property (Proof of Proposition \ref{prop: stoch_property_stoch_Y})}

\subsubsection{Proof of Proposition \ref{prop: stoch_property_stoch_Y}(1)}

\subsubsection*{MLE}

First, we have:

\begin{eqnarray}
\left\Vert \nabla_{Y}Q\left(\widehat{\theta},\widehat{Y}\right)\right\Vert _{1} & = & \sum_{p=1}^{n_{Y}}\left|\frac{1}{N}\sum_{i=1}^{N}\nabla_{Y_{p}}\ln f\left(w_{i}|\widehat{\gamma}\right)\right|\nonumber \\
 & \leq & \sum_{p=1}^{n_{Y}}\left|\frac{1}{N}\sum_{i=1}^{N}\left(\nabla_{Y_{p}}\ln f\left(w_{i}|\widehat{\gamma}\right)-\nabla_{Y_{p}}\ln f\left(w_{i}|\gamma^{*}\right)\right)\right|\nonumber \\
 &  & +\sum_{p=1}^{n_{Y}}\left|\frac{1}{N}\sum_{i=1}^{N}\left(\nabla_{Y_{p}}\ln f\left(w_{i}|\gamma^{*}\right)\right)\right|.\label{eq:L1_dq_dY}
\end{eqnarray}

Regarding the first term, the mean value theorem implies that there exist $c\in[0,1]$ and $\overline{\gamma}\equiv c\widehat{\gamma}+(1-c)\gamma^{*}$ such that

\begin{eqnarray*}
\nabla_{Y_{p}}\ln f\left(w_{i}|\widehat{\gamma}\right)-\nabla_{Y_{p}}\ln f\left(w_{i}|\gamma^{*}\right) & = & \left(\nabla_{Y_{p}\gamma^{\prime}}\ln f\left(w_{i}|\overline{\gamma}\right)\right)\left(\widehat{\gamma}-\gamma^{*}\right).
\end{eqnarray*}

This equation implies:

\begin{eqnarray*}
 &  & \sum_{p=1}^{n_{Y}}\left|\frac{1}{N}\sum_{i=1}^{N}\left(\nabla_{Y_{p}}\ln f\left(w_{i}|\widehat{\gamma}\right)-\nabla_{Y_{p}}\ln f\left(w_{i}|\gamma^{*}\right)\right)\right|\\
 & = & \sum_{p=1}^{n_{Y}}\left|\frac{1}{N}\sum_{i=1}^{N}\sum_{q=1}^{n_{\gamma}}\left(\nabla_{Y_{p}\gamma_{q}}\ln f\left(w_{i}|\overline{\gamma}\right)\right)\left(\widehat{\gamma_{q}}-\gamma_{q}^{*}\right)\right|\\
 & \leq & \sum_{p=1}^{n_{Y}}\sum_{q=1}^{n_{\gamma}}\frac{1}{N}\sum_{i=1}^{N}\left|\left(\nabla_{Y_{p}\gamma_{q}}\ln f\left(w_{i}|\overline{\gamma}\right)\right)\left(\widehat{\gamma_{q}}-\gamma_{q}^{*}\right)\right|\\
 & \leq & \left(\sum_{p=1}^{n_{Y}}\sum_{q=1}^{n_{\gamma}}\frac{1}{N}\sum_{i=1}^{N}\left|\nabla_{Y_{p}\gamma_{q}}\ln f\left(w_{i}|\overline{\gamma}\right)\right|\right)\cdot\left(\max_{q=1,\cdots,n_{\gamma}}\left|\widehat{\gamma_{q}}-\gamma_{q}^{*}\right|\right)\\
 & = & O_{p}\left(1\right)\cdot O_{p}\left(N^{-\frac{1}{2}}\right)\ \left(\because\text{Assumptions}\ \text{\ref{as:stat_ass_stochastic_Y}(e*) and \ref{as:stochastic_Y_ass}(c-MLE)}\right)\\
 & = & O_{p}\left(N^{-\frac{1}{2}}\right).
\end{eqnarray*}

In addition, Assumption \ref{as:stochastic_Y_ass}(b-MLE) implies that $\sum_{p=1}^{n_{Y}}\left|\frac{1}{N}\sum_{i=1}^{N}\nabla_{Y_{p}}\ln f\left(w_{i}|\gamma^{*}\right)\right|=O_{p}\left(N^{-\frac{1}{2}}\right)$ holds. Hence, $\left\Vert \nabla_{Y}Q\left(\widehat{\theta},\widehat{Y}\right)\right\Vert _{1}=O_{p}\left(N^{-\frac{1}{2}}\right)$ holds by (\ref{eq:L1_dq_dY}).

\paragraph*{GMM}

First, $\nabla_{Y}Q\left(\widehat{\theta},\widehat{Y}\right)=2\sum_{p,q=1}^{n_{m}}\widehat{W}_{pq}\left(\nabla_{Y}m_{p}\left(\widehat{\theta},\widehat{Y}\right)\right)\left(m_{q}\left(\widehat{\theta},\widehat{Y}\right)\right)$ holds, where $\overline{m_{p}}\left(\theta,Y\right)$ denotes the $p$-th element of the vector $\overline{m}\left(\theta,Y\right)$. Regarding $\overline{m_{q}}\left(\widehat{\theta},\widehat{Y}\right)$, the following holds:

\begin{eqnarray}
\left|\overline{m_{q}}\left(\widehat{\theta},\widehat{Y}\right)\right| & = & \left|\overline{m_{q}}\left(\gamma^{*}\right)+\left(\nabla_{\gamma}\overline{m_{q}}\left(\overline{\gamma}\right)\right)\left(\widehat{\gamma}-\gamma^{*}\right)\right|\ \left(\text{for some }\overline{\gamma}\text{ by the mean value theorem}\right)\nonumber \\
 & \leq & \left|\overline{m_{q}}\left(\gamma^{*}\right)\right|+\left|\sum_{r=1}^{n_{\gamma}}\left(\nabla_{\gamma_{r}}\overline{m_{q}}\left(\overline{\gamma}\right)\right)\left(\widehat{\gamma_{r}}-\gamma_{r}^{*}\right)\right|\nonumber \\
 & \leq & \left|\overline{m_{q}}\left(\gamma^{*}\right)\right|+\left(\sum_{r=1}^{n_{\gamma}}\left|\nabla_{\gamma_{r}}\overline{m_{q}}\left(\overline{\gamma}\right)\right|\right)\cdot\left(\max_{r=1,\cdots,n_{\gamma}}\left|\widehat{\gamma_{r}}-\gamma_{r}^{*}\right|\right)\nonumber \\
 & = & O_{p}\left(N^{-\frac{1}{2}}\right)+O_{p}\left(1\right)\cdot O_{p}\left(N^{-\frac{1}{2}}\right)\ \left(\because\text{Assumptions}\ \text{\ref{as:stat_ass_stochastic_Y}(e*), (c*-GMM), and \ref{as:stochastic_Y_ass}(b-GMM)}\right)\nonumber \\
 & = & O_{p}\left(N^{-\frac{1}{2}}\right)\label{eq:moment_abs_order}
\end{eqnarray}

Then, we obtain:

\begin{eqnarray*}
\sum_{r=1}^{n_{Y}}\left|\nabla_{Y_{r}}Q\left(\widehat{\theta},\widehat{Y}\right)\right| & \leq & 2\sum_{p,q=1}^{n_{m}}\left|\widehat{W}_{pq}\right|\left(\sum_{s=1}^{n_{Y}}\left|\nabla_{Y_{s}}\overline{m_{p}}\left(\widehat{\theta},\widehat{Y}\right)\right|\right)\left|\overline{m_{q}}\left(\widehat{\theta},\widehat{Y}\right)\right|\\
 & = & O_{p}\left(1\right)\cdot O_{p}\left(N^{-\frac{1}{2}}\right)\ \left(\because\text{Assumption\ \ref{as:stochastic_Y_ass}(b-GMM)\ and\ (\ref{eq:moment_abs_order})}\right)\\
 & = & O_{p}\left(N^{-\frac{1}{2}}\right).
\end{eqnarray*}

\subsubsection{Proof of Proposition \ref{prop: stoch_property_stoch_Y}(2)}

By (\ref{eq:d2Q_d2theta_gamma_at_sol}), $\nabla_{\theta\gamma^{\prime}}\widetilde{Q}\left(\widehat{\theta};\widehat{\gamma}\right)=\sum_{j=1}^{n_{Y}}\left(\nabla_{Y_{j}}Q\left(\widehat{\theta},\widehat{Y}\right)\right)\left(\nabla_{\theta\gamma^{\prime}}\Upsilon_{j}\left(\widehat{\theta},\widehat{\gamma}\right)\right)$ holds. Hence, we obtain:

\begin{eqnarray*}
 &  & \left\Vert \left(\nabla_{\theta\gamma^{\prime}}\widetilde{Q}\left(\widehat{\theta};\widehat{\gamma}\right)\right)\left(\gamma_{k-1}-\widehat{\gamma}\right)\right\Vert _{\infty}\\
 & = & \left\Vert \sum_{j=1}^{n_{Y}}\left(\nabla_{Y_{j}}Q\left(\widehat{\theta},\widehat{Y}\right)\right)\left(\nabla_{\theta\gamma^{\prime}}\Upsilon_{j}\left(\widehat{\theta},\widehat{\gamma}\right)\right)\left(\gamma_{k-1}-\widehat{\gamma}\right)\right\Vert _{\infty}\\
 & = & \left\Vert \sum_{j=1}^{n_{Y}}\left(\nabla_{Y_{j}}Q\left(\widehat{\theta},\widehat{Y}\right)\right)\sum_{p=1}^{n_{\gamma}}\left(\nabla_{\theta\gamma_{p}}\Upsilon_{j}\left(\widehat{\theta},\widehat{\gamma}\right)\right)\left(\gamma_{k-1,p}-\widehat{\gamma_{p}}\right)\right\Vert _{\infty}\\
 & \leq & \left(\sum_{j=1}^{n_{Y}}\left|\nabla_{Y_{j}}Q\left(\widehat{\theta},\widehat{Y}\right)\right|\right)\left(\max_{j=1,\cdots,n_{Y}}\sum_{p=1}^{n_{\gamma}}\left\Vert \nabla_{\theta\gamma_{p}}\Upsilon_{j}\left(\widehat{\theta},\widehat{\gamma}\right)\right\Vert _{\infty}\left|\gamma_{k-1,p}-\widehat{\gamma_{p}}\right|\right)\\
 & \leq & \left(\sum_{j=1}^{n_{Y}}\left|\nabla_{Y_{j}}Q\left(\widehat{\theta},\widehat{Y}\right)\right|\right)\left\Vert \gamma_{k-1}-\widehat{\gamma}\right\Vert _{\infty}\cdot\left(\max_{j=1,\cdots,n_{Y}}\sum_{p=1}^{n_{\gamma}}\left\Vert \nabla_{\theta\gamma_{p}}\Upsilon_{j}\left(\widehat{\theta},\widehat{\gamma}\right)\right\Vert _{\infty}\right)\\
 & = & O_{p}\left(N^{-\frac{1}{2}}\left\Vert \gamma_{k-1}-\widehat{\gamma}\right\Vert _{\infty}\right)\cdot O_{p}\left(1\right)\ \left(\because\text{Proposition \ref{prop: stoch_property_stoch_Y}(1) and Assumption\ \ref{as:stochastic_Y_ass}(a)}\right)\\
 & = & O_{p}\left(N^{-\frac{1}{2}}\left\Vert \gamma_{k-1}-\widehat{\gamma}\right\Vert _{\infty}\right).
\end{eqnarray*}

\subsubsection{Proof of Proposition \ref{prop: stoch_property_stoch_Y}(3)}

The statement directly follows from (2).

\subsubsection{Proof of Proposition \ref{prop: stoch_property_stoch_Y}(4)}

By Lemma \ref{lem:difference-Hessian-NFXP}, 
\begin{eqnarray*}
\nabla_{\theta\theta^{\prime}}\widetilde{Q}\left(\widehat{\theta};\widehat{\gamma}\right)-\nabla_{\theta\theta^{\prime}}Q_{NFXP}\left(\widehat{\theta}\right) & = & \sum_{j=1}^{n_{Y}}\frac{\partial Q\left(\widehat{\theta},\widehat{Y}\right)}{\partial Y_{j}}\left(\nabla_{\theta\theta^{\prime}}\Upsilon_{j}(\widehat{\theta};\widehat{\gamma})-\nabla_{\theta\theta^{\prime}}\mathcal{Y}_{j}(\widehat{\theta})\right).
\end{eqnarray*}

Hence, we obtain:

\begin{eqnarray*}
 &  & \left\Vert \nabla_{\theta\theta^{\prime}}\widetilde{Q}\left(\widehat{\theta};\widehat{\gamma}\right)-\nabla_{\theta\theta^{\prime}}Q_{NFXP}\left(\widehat{\theta}\right)\right\Vert \\
 & \leq & \sum_{j=1}^{n_{Y}}\left|\frac{\partial Q\left(\widehat{\theta},\widehat{Y}\right)}{\partial Y_{j}}\right|\cdot\max_{j=1,\cdots,n_{Y}}\left\Vert \nabla_{\theta\theta^{\prime}}\Upsilon_{j}(\widehat{\theta};\widehat{\gamma})-\nabla_{\theta\theta^{\prime}}\mathcal{Y}_{j}(\widehat{\theta})\right\Vert \\
 & = & O_{p}\left(N^{-\frac{1}{2}}\right)\cdot O_{p}\left(1\right)\ \left(\because\text{Proposition \ref{prop: stoch_property_stoch_Y}(1) and Assumption\ \ref{as:stochastic_Y_ass_for_asy_dist}(b)}\right)\\
 & = & O_{p}\left(N^{-\frac{1}{2}}\right).
\end{eqnarray*}

\subsection{Stochastic $Y$ setting: Asymptotic property (Proof of Proposition \ref{prop: asy_property_consistent_initial_values_stoch_Y})}

\subsubsection{Proof of Proposition \ref{prop: asy_property_consistent_initial_values_stoch_Y}(1)}

To show Proposition \ref{prop: asy_property_consistent_initial_values_stoch_Y}(1), we rely on the following lemma. 
\begin{lem}
\label{lem:Lemma_for_consistency_stoch_Y}Suppose that $\widetilde{\gamma}\left(\omega\right)-\gamma^{*}\left(\theta^{*},\omega\right)\rightarrow_{a.s.}0$ and the assumptions of Proposition \ref{prop: asy_property_consistent_initial_values_stoch_Y} hold. Then,

(a). $\overline{Q}_{\widetilde{\gamma}}\left(\theta,\omega\right)\equiv\widetilde{Q}\left(\theta;\widetilde{\gamma}\left(\omega\right),\omega\right)=Q\left(\theta,\Upsilon\left(\theta,\widetilde{\gamma}\left(\omega\right),\omega^{G}\right),\omega^{Q}\right)$ converges almost surely and uniformly in $\theta\in\Theta$ to $\overline{Q}_{\infty}\left(\theta\right)$. Here, $\overline{Q}_{\infty}\left(\theta\right)$ is given by 

\[
\overline{Q}_{\infty}\left(\theta\right)\equiv-E_{\omega}\left[\ln f\left(\omega^{Q}|\theta,\Upsilon\left(\theta,\gamma^{*}\left(\theta^{*},\omega\right),\omega^{G}\right)\right)\right]
\]

for MLE, and

\[
\overline{Q}_{\infty}\left(\theta\right)\equiv\left(E_{\omega}\left[m\left(\theta,\Upsilon\left(\theta,\gamma^{*}\left(\theta^{*},\omega\right),\omega^{G}\right),\omega^{Q}\right)\right]\right)^{\prime}W\left(E_{\omega}\left[m\left(\theta,\Upsilon\left(\theta,\gamma^{*}\left(\theta^{*},\omega\right),\omega^{G}\right),\omega^{Q}\right)\right]\right)
\]

for GMM.

(b). $\widetilde{\theta}\equiv\arg\min_{\theta\in\Theta}\overline{Q}_{\widetilde{\gamma}}\left(\theta,\omega\right)$ converges almost surely to $\theta^{*}\equiv\arg\min_{\theta\in\Theta}\overline{Q}_{\infty}\left(\theta\right)$. 
\end{lem}
\begin{proof}
(a). For MLE, we define the functions $\overline{Q}_{\infty}\left(\theta\right)$ and $\overline{Q}_{\widetilde{\gamma},\infty}\left(\theta\right)$:

\begin{eqnarray*}
\overline{Q}_{\infty}\left(\theta\right) & \equiv & -E_{\omega}\left[\ln f\left(\omega^{Q}|\theta,\Upsilon\left(\theta,\gamma^{*}\left(\theta^{*},\omega\right),\omega^{G}\right)\right)\right],
\end{eqnarray*}

\begin{eqnarray*}
\overline{Q}_{\widetilde{\gamma},\infty}\left(\theta\right) & \equiv & -E_{\omega}\left[\ln f\left(\omega^{Q}|\theta,\Upsilon\left(\theta,\widetilde{\gamma}\left(\omega\right),\omega^{G}\right)\right)\right].
\end{eqnarray*}
Note that $\overline{Q}_{\widetilde{\gamma}}\left(\theta,\omega\right)$ is given by $\overline{Q}_{\widetilde{\gamma}}\left(\theta,\omega\right)=-\frac{1}{N}\sum_{i=1}^{N}\left[\ln f\left(\omega_{i}^{Q}|\theta,\Upsilon\left(\theta,\widetilde{\gamma}\left(\omega_{i}\right),\omega_{i}^{G}\right)\right)\right]$.

Here, the following holds:

\begin{eqnarray*}
 &  & \sup_{\theta\in\Theta}\left|\overline{Q}_{\widetilde{\gamma}}\left(\theta,\omega\right)-\overline{Q}_{\infty}\left(\theta\right)\right|\\
 & \leq & \sup_{\theta\in\Theta}\left|\overline{Q}_{\widetilde{\gamma}}\left(\theta,\omega\right)-\overline{Q}_{\widetilde{\gamma},\infty}\left(\theta\right)\right|+\sup_{\theta\in\Theta}\left|\overline{Q}_{\widetilde{\gamma},\infty}\left(\theta\right)-\overline{Q}_{\infty}\left(\theta\right)\right|.
\end{eqnarray*}

Regarding the first term, Assumption \ref{as:stat_ass_stochastic_Y}(c{*}-MLE) implies that the term converges to 0 almost surely. Concerning the second term, $\widetilde{\gamma}\left(\omega\right)-\gamma^{*}\left(\theta^{*},\omega\right)\rightarrow_{a.s.}0$ implies that the second term also converges to 0 almost surely. Note that analogous discussions also hold for GMM.

(b). By Property 24.2 of \citet{gourieroux1995statistics} and part (a), we obtain the statement. 
\end{proof}

\paragraph*{Proof of Proposition \ref{prop: asy_property_consistent_initial_values_stoch_Y}(1)}

We show the statement by induction. Suppose that $\gamma_{k-1}-\gamma^{*}\left(\theta^{*},\omega\right)\rightarrow_{a.s.}0$ holds. Then, $\theta_{k}$ converges almost surely to $\theta^{*}$ by Lemma \ref{lem:Lemma_for_consistency_stoch_Y}, and $Y_{k}=\Upsilon\left(\theta_{k};\gamma_{k-1}\right)$ converges almost surely to $Y^{*}$ by continuity.

\subsubsection{Proof of Proposition \ref{prop: asy_property_consistent_initial_values_stoch_Y}(2)}

We basically inherit the strategy of the proof in Proposition \ref{prop: asy_property_consistent_initial_values}(2). However, the following three parts must be shown under the current stochastic-$Y$ setting:

(A): $\nabla_{\theta\theta^{\prime}}\widetilde{Q}\left(\overline{\theta};\gamma_{k-1}\right)-\nabla_{\theta\theta^{\prime}}Q_{NFXP}\left(\widetilde{\theta}\right)=o_{p}\left(1\right)$ (cf. (\ref{eq:Hessian_asy_equiv}))

(B): $\sqrt{N}\left(\nabla_{\theta\gamma^{\prime}}\widetilde{Q}\left(\theta^{*},\gamma^{*}\right)\right)\left(\gamma_{k-1}-\gamma^{*}\right)=o_{p}\left(1\right)$ (cf. Part (b) in the proof of Proposition \ref{prop: asy_property_consistent_initial_values}(2))

(C): $\left(\gamma_{k-1}-\gamma^{*}\right)^{\prime}\left(\nabla_{\theta_{j}\gamma\gamma^{\prime}}\widetilde{Q}\left(\theta^{*};\overline{\gamma}\right)\right)\left(\gamma_{k-1}-\gamma^{*}\right)=o_{p}\left(N^{-\frac{1}{2}}\right)\ \left(j=1,\cdots,n_{\theta}\right)$ (cf. Part (c) in the proof of Proposition \ref{prop: asy_property_consistent_initial_values}(2)).

\textbf{Proof of (A). $\nabla_{\theta\theta^{\prime}}\widetilde{Q}\left(\overline{\theta};\gamma_{k-1}\right)-\nabla_{\theta\theta^{\prime}}Q_{NFXP}\left(\widetilde{\theta}\right)=o_{p}\left(1\right)$ (cf. (\ref{eq:Hessian_asy_equiv}))}

\begin{eqnarray}
 &  & \left\Vert \nabla_{\theta\theta^{\prime}}\widetilde{Q}\left(\overline{\theta};\gamma_{k-1}\right)-\nabla_{\theta\theta^{\prime}}Q_{NFXP}\left(\widetilde{\theta}\right)\right\Vert \nonumber \\
 & \leq & \left\Vert \nabla_{\theta\theta^{\prime}}\widetilde{Q}\left(\overline{\theta};\gamma_{k-1}\right)-\nabla_{\theta\theta^{\prime}}\widetilde{Q}\left(\widehat{\theta};\widehat{\gamma}\right)\right\Vert +\left\Vert \nabla_{\theta\theta^{\prime}}\widetilde{Q}\left(\widehat{\theta};\widehat{\gamma}\right)-\nabla_{\theta\theta^{\prime}}Q_{NFXP}\left(\widehat{\theta}\right)\right\Vert +\nonumber \\
 &  & \left\Vert \nabla_{\theta\theta^{\prime}}Q_{NFXP}\left(\widehat{\theta}\right)-\nabla_{\theta\theta^{\prime}}Q_{NFXP}\left(\widetilde{\theta}\right)\right\Vert \label{eq:d2Q_dtheta2_Q_bound}
\end{eqnarray}

Regarding the first term, the mean value theorem implies that there exists $c\in[0,1]$ such that $\check{\gamma}=c\widehat{\gamma}+(1-c)\gamma_{k-1}$, $\check{\theta}=c\widehat{\theta}+(1-c)\overline{\theta}$, and 

\begin{eqnarray*}
\nabla_{\theta\theta^{\prime}}\widetilde{Q}\left(\overline{\theta};\gamma_{k-1}\right)-\nabla_{\theta\theta^{\prime}}\widetilde{Q}\left(\widehat{\theta};\widehat{\gamma}\right) & = & \sum_{j=1}^{n_{\theta}}\left(\nabla_{\theta\theta^{\prime}\theta_{j}}\widetilde{Q}\left(\check{\theta};\check{\gamma}\right)\right)\left(\overline{\theta_{j}}-\widehat{\theta_{j}}\right)+\\
 &  & \sum_{j=1}^{n_{\gamma}}\left(\nabla_{\theta\theta^{\prime}\gamma_{j}}\widetilde{Q}\left(\check{\theta};\check{\gamma}\right)\right)\left(\gamma_{k-1,j}-\widehat{\gamma_{j}}\right).
\end{eqnarray*}

Hence, the first term satisfies:

\begin{eqnarray*}
 &  & \left\Vert \nabla_{\theta\theta^{\prime}}\widetilde{Q}\left(\overline{\theta};\gamma_{k-1}\right)-\nabla_{\theta\theta^{\prime}}\widetilde{Q}\left(\widehat{\theta};\widehat{\gamma}\right)\right\Vert \\
 & \leq & \left\Vert \sum_{j=1}^{n_{\theta}}\left(\nabla_{\theta\theta^{\prime}\theta_{j}}\widetilde{Q}\left(\check{\theta};\check{\gamma}\right)\right)\left(\overline{\theta_{j}}-\widehat{\theta_{j}}\right)\right\Vert +\left\Vert \sum_{j=1}^{n_{\gamma}}\left(\nabla_{\theta\theta^{\prime}\gamma_{j}}\widetilde{Q}\left(\check{\theta};\check{\gamma}\right)\right)\left(\gamma_{k-1,j}-\widehat{\gamma_{j}}\right)\right\Vert \\
 & \leq & \left(\sum_{j=1}^{n_{\theta}}\left\Vert \nabla_{\theta\theta^{\prime}\theta_{j}}\widetilde{Q}\left(\check{\theta};\check{\gamma}\right)\right\Vert \right)\left\Vert \overline{\theta}-\widehat{\theta}\right\Vert _{\infty}+\left(\sum_{j=1}^{n_{\gamma}}\left\Vert \nabla_{\theta\theta^{\prime}\gamma_{j}}\widetilde{Q}\left(\check{\theta};\check{\gamma}\right)\right\Vert \right)\left\Vert \gamma_{k-1}-\widehat{\gamma}\right\Vert _{\infty}\\
 & = & o_{p}\left(1\right)\ \left(\because\overline{\theta}-\widehat{\theta}=o_{p}\left(1\right),\ \gamma_{k-1}-\widehat{\gamma}=o_{p}\left(1\right),\ \text{Assumption\ \ref{as:stochastic_Y_ass_for_asy_dist}}\right).
\end{eqnarray*}

The second term in (\ref{eq:d2Q_dtheta2_Q_bound}) satisfies $\left\Vert \nabla_{\theta\theta^{\prime}}\widetilde{Q}\left(\widehat{\theta};\widehat{\gamma}\right)-\nabla_{\theta\theta^{\prime}}Q_{NFXP}\left(\widehat{\theta}\right)\right\Vert =O_{p}\left(N^{-\frac{1}{2}}\right)$ by Proposition \ref{prop: stoch_property_stoch_Y}(4). The third term satisfies $\nabla_{\theta\theta^{\prime}}Q_{NFXP}\left(\widehat{\theta}\right)-\nabla_{\theta\theta^{\prime}}Q_{NFXP}\left(\widetilde{\theta}\right)=o_{p}\left(1\right)$ by the continuity of $\nabla_{\theta\theta^{\prime}}Q_{NFXP}\left(\theta\right)$ and $\widehat{\theta}-\widetilde{\theta}=o_{p}\left(1\right)$. Therefore, we obtain $\nabla_{\theta\theta^{\prime}}\widetilde{Q}\left(\overline{\theta};\gamma_{k-1}\right)-\nabla_{\theta\theta^{\prime}}Q_{NFXP}\left(\widetilde{\theta}\right)=o_{p}\left(1\right)$.

\textbf{Proof of (B). $\sqrt{N}\left(\nabla_{\theta\gamma^{\prime}}\widetilde{Q}\left(\theta^{*},\gamma^{*}\right)\right)\left(\gamma_{k-1}-\gamma^{*}\right)=o_{p}\left(1\right)$:}

By Proposition \ref{prop: stoch_property_stoch_Y}(2), $\sqrt{N}\left(\nabla_{\theta\gamma^{\prime}}\widetilde{Q}\left(\widehat{\theta},\widehat{\gamma}\right)\right)\left(\gamma_{k-1}-\gamma^{*}\right)=o_{p}\left(1\right)$ holds under $\left\Vert \gamma_{k-1}-\gamma^{*}\right\Vert _{\infty}=o_{p}\left(1\right)$. By $\widehat{\gamma}-\gamma^{*}=o_{p}\left(1\right)$, the claim holds.

\textbf{Proof of (C). $\left(\gamma_{k-1}-\gamma^{*}\right)^{\prime}\left(\nabla_{\theta_{j}\gamma\gamma^{\prime}}\widetilde{Q}\left(\theta^{*};\overline{\gamma}\right)\right)\left(\gamma_{k-1}-\gamma^{*}\right)=o_{p}\left(N^{-\frac{1}{2}}\right)$:}

Because

\begin{eqnarray*}
 &  & \left|\left(\gamma_{k-1}-\gamma^{*}\right)^{\prime}\left(\nabla_{\theta_{j}\gamma\gamma^{\prime}}\widetilde{Q}\left(\theta^{*};\overline{\gamma}\right)\right)\left(\gamma_{k-1}-\gamma^{*}\right)\right|\\
 & \leq & \sum_{p,q=1}^{n_{\gamma}}\left|\nabla_{\theta_{j}\gamma_{p}\gamma_{q}}\widetilde{Q}\left(\theta^{*};\overline{\gamma}\right)\right|\cdot\left\Vert \gamma_{k-1}-\gamma^{*}\right\Vert _{\infty}\cdot\left\Vert \gamma_{k-1}-\gamma^{*}\right\Vert _{\infty}
\end{eqnarray*}
holds, Assumption \ref{as:stochastic_Y_ass_for_asy_dist} and $\left\Vert \gamma_{k-1}-\gamma^{*}\right\Vert _{\infty}=o_{p}\left(N^{-\frac{1}{4}}\right)$ imply $\left(\gamma_{k-1}-\gamma^{*}\right)^{\prime}\left(\nabla_{\theta_{j}\gamma\gamma^{\prime}}\widetilde{Q}\left(\theta^{*};\overline{\gamma}\right)\right)\left(\gamma_{k-1}-\gamma^{*}\right)=o_{p}\left(N^{-\frac{1}{2}}\right)$.

\subsubsection{Proof of Proposition \ref{prop: asy_property_consistent_initial_values_stoch_Y}(3)}

The strategy of the proof is essentially the same as that for Proposition \ref{prop: asy_property_consistent_initial_values}(3), except for the choice of norm. Whereas the proof of Proposition \ref{prop: asy_property_consistent_initial_values}(3) uses a Euclidean norm, here we use the $L_{\infty}$ norm. Applying the $L_{\infty}$ norm yields the stated result.

\bibliographystyle{apalike}
\bibliography{literature}

@article{gandhi2019measuring,
  title={Measuring substitution patterns in differentiated-products industries},
  author={Gandhi, Amit and Houde, Jean-Fran{\c{c}}ois},
  journal={NBER Working paper},
  number={w26375},
  year={2019}
}

@article{bray2019markov,
  title={Markov decision processes with exogenous variables},
  author={Bray, Robert L},
  journal={Management Science},
  volume={65},
  number={10},
  pages={4598--4606},
  year={2019},
  publisher={INFORMS}
}

@article{gallant2022constrained,
  title={Constrained estimation using penalization and MCMC},
  author={Gallant, A Ronald and Hong, Han and Leung, Michael P and Li, Jessie},
  journal={Journal of Econometrics},
  volume={228},
  number={1},
  pages={85--106},
  year={2022},
  publisher={Elsevier}
}

@article{li2000derivative,
  title={A derivative-free line search and global convergence of Broyden-like method for nonlinear equations},
  author={Li, Dong-Hui and Fukushima, Masao},
  journal={Optimization methods and software},
  volume={13},
  number={3},
  pages={181--201},
  year={2000},
  publisher={Taylor \& Francis}
}

@article{chen2025derivative,
  title={A Derivative-Free Regularized Primal-Dual Interior-Point Algorithm for Constrained Nonlinear Least Squares Problems},
  author={Chen, Xi and Fan, Jinyan},
  journal={Journal of Scientific Computing},
  volume={103},
  number={2},
  pages={48},
  year={2025},
  publisher={Springer}
}

@article{larson2019derivative,
  title={Derivative-free optimization methods},
  author={Larson, Jeffrey and Menickelly, Matt and Wild, Stefan M},
  journal={Acta Numerica},
  volume={28},
  pages={287--404},
  year={2019},
  publisher={Cambridge University Press}
}

@article{heinkenschloss2014matrix,
  title={A matrix-free trust-region SQP method for equality constrained optimization},
  author={Heinkenschloss, Matthias and Ridzal, Denis},
  journal={SIAM Journal on Optimization},
  volume={24},
  number={3},
  pages={1507--1541},
  year={2014},
  publisher={SIAM}
}

@article{arreckx2018regularized,
  title={A regularized factorization-free method for equality-constrained optimization},
  author={Arreckx, Sylvain and Orban, Dominique},
  journal={SIAM Journal on Optimization},
  volume={28},
  number={2},
  pages={1613--1639},
  year={2018},
  publisher={SIAM}
}

@article{monardo2025comparing,
  title={Comparing Methods for Estimating Random Coefficient Logit Demand Models},
  author={Monardo, Julien},
  year={2025}
}

@article{aitchison1958maximum,
  title={Maximum-likelihood estimation of parameters subject to restraints},
  author={Aitchison, John and Silvey, SD},
  journal={The annals of mathematical Statistics},
  pages={813--828},
  year={1958},
  publisher={JSTOR}
}

@article{bugni2021iterated,
  title={{On the iterated estimation of dynamic discrete choice games}},
  author={Bugni, Federico A and Bunting, Jackson},
  journal={The Review of Economic Studies},
  volume={88},
  number={3},
  pages={1031--1073},
  year={2021},
  publisher={Oxford University Press}
}

@article{hall1974estimation,
  title={{Estimation and inference in nonlinear structural models}},
  author={Hall, BH and Hall, RE and Hausman, JA},
  journal={Annals of economic and social measurement},
  volume={3},
  pages={653--666},
  year={1974}
}

@article{belloni2018high,
  title={{High-dimensional econometrics and regularized GMM}},
  author={Belloni, Alexandre and Chernozhukov, Victor and Chetverikov, Denis and Hansen, Christian and Kato, Kengo},
  journal={arXiv preprint arXiv:1806.01888},
  year={2018}
}

@book{gourieroux1995statistics,
  title={{Statistics and econometric models}},
  author={Gourieroux, Christian and Monfort, Alain},
  volume={1},
  year={1995},
  publisher={Cambridge University Press}
}

@article{robinson1972quadratically,
  title={{A quadratically-convergent algorithm for general nonlinear programming problems}},
  author={Robinson, Stephen M},
  journal={Mathematical programming},
  volume={3},
  number={1},
  pages={145--156},
  year={1972},
  publisher={Springer}
}

@article{chernozhukov2018double,
    author = {Chernozhukov, Victor and Chetverikov, Denis and Demirer, Mert and Duflo, Esther and Hansen, Christian and Newey, Whitney and Robins, James},
    title = {{Double/debiased machine learning for treatment and structural parameters}},
    journal = {The Econometrics Journal},
    volume = {21},
    number = {1},
    pages = {C1-C68},
    year = {2018},
    month = {01},
    issn = {1368-4221},
    doi = {10.1111/ectj.12097},
}

@article{kasahara2008pseudo,
  title={{Pseudo-likelihood estimation and bootstrap inference for structural discrete Markov decision models}},
  author={Kasahara, Hiroyuki and Shimotsu, Katsumi},
  journal={Journal of Econometrics},
  volume={146},
  number={1},
  pages={92--106},
  year={2008},
  publisher={Elsevier}
}

@article{chernozhukov2022locally,
  title={{Locally robust semiparametric estimation}},
  author={Chernozhukov, Victor and Escanciano, Juan Carlos and Ichimura, Hidehiko and Newey, Whitney K and Robins, James M},
  journal={Econometrica},
  volume={90},
  number={4},
  pages={1501--1535},
  year={2022},
  publisher={Wiley Online Library}
}

@article{sawadogo2025efficient,
  title={{Efficient and Debiased Estimation of Dynamic Discrete Choice Models}},
  author={Sawadogo, Sidi},
  journal={Available at SSRN 5325574},
  year={2025}
}

@article{pollock2021anderson,
  title={{Anderson acceleration for contractive and noncontractive operators}},
  author={Pollock, Sara and Rebholz, Leo G},
  journal={IMA Journal of Numerical Analysis},
  volume={41},
  number={4},
  pages={2841--2872},
  year={2021},
  publisher={Oxford University Press}
}

@article{murtagh1982projected,
  title={{A projected Lagrangian algorithm and its implementation for sparse nonlinear constraints}},
  author={Murtagh, Bruce A and Saunders, Michael A},
  journal={Mathematical Programming Study},
  volume={16},
  pages={84--117},
  year={1982},
  publisher={Springer}
}

@article{yamaguchi2019effects,
  title={{Effects of parental leave policies on female career and fertility choices}},
  author={Yamaguchi, Shintaro},
  journal={Quantitative Economics},
  volume={10},
  number={3},
  pages={1195--1232},
  year={2019},
  publisher={Wiley Online Library}
}

@article{imai2009bayesian,
  title={{Bayesian estimation of dynamic discrete choice models}},
  author={Imai, Susumu and Jain, Neelam and Ching, Andrew},
  journal={Econometrica},
  volume={77},
  number={6},
  pages={1865--1899},
  year={2009},
  publisher={Wiley Online Library}
}

@article{egesdal2015estimating,
  title={{Estimating dynamic discrete-choice games of incomplete information}},
  author={Egesdal, Michael and Lai, Zhenyu and Su, Che-Lin},
  journal={Quantitative Economics},
  volume={6},
  number={3},
  pages={567--597},
  year={2015},
  publisher={Wiley Online Library}
}

@article{hotz1993conditional,
  title={{Conditional choice probabilities and the estimation of dynamic models}},
  author={Hotz, V Joseph and Miller, Robert A},
  journal={The Review of Economic Studies},
  volume={60},
  number={3},
  pages={497--529},
  year={1993},
  publisher={Wiley-Blackwell}
}

@book{saad2003iterative,
  title={{Iterative methods for sparse linear systems}},
  author={Saad, Yousef},
  year={2003},
  publisher={SIAM}
}

@article{lu2023semi,
  title={{Semi-nonparametric estimation of random coefficients logit model for aggregate demand}},
  author={Lu, Zhentong and Shi, Xiaoxia and Tao, Jing},
  journal={Journal of Econometrics},
  volume={235},
  number={2},
  pages={2245--2265},
  year={2023},
  publisher={Elsevier}
}

@article{foster2023orthogonal,
  title={{Orthogonal statistical learning}},
  author={Foster, Dylan J and Syrgkanis, Vasilis},
  journal={The Annals of Statistics},
  volume={51},
  number={3},
  pages={879--908},
  year={2023},
  publisher={Institute of Mathematical Statistics}
}

@article{dearing2024efficient,
  title={{Efficient and convergent sequential pseudo-likelihood estimation of dynamic discrete games}},
  author={Dearing, Adam and Blevins, Jason R},
  journal={Review of Economic Studies},
  volume={92},
  number={2},
  pages={981--1021},
  year={2025},
  publisher={Oxford University Press UK}
}

@article{su2014estimating,
  title={{Estimating discrete-choice games of incomplete information: Simple static examples}},
  author={Su, Che-Lin},
  journal={Quantitative Marketing and Economics},
  volume={12},
  pages={167--207},
  year={2014},
  publisher={Springer}
}

@article{su2012constrained,
  title={{Constrained optimization approaches to estimation of structural models}},
  author={Su, Che-Lin and Judd, Kenneth L},
  journal={Econometrica},
  volume={80},
  number={5},
  pages={2213--2230},
  year={2012},
  publisher={Wiley Online Library}
}

@article{fukasawa2024fast,
  title={{Fast and simple inner-loop algorithms of static/dynamic BLP estimations}},
  author={Fukasawa, Takeshi},
  journal={arXiv preprint arXiv:2404.04494},
  year={2024}
}

@book{nocedal2006numerical,
  title={{Numerical optimization (Second Edition)}},
  author={Nocedal, Jorge and Wright, Stephen J},
  year={2006},
  publisher={Springer}
}

@article{iskhakov2016comment,
  title={{Comment on “constrained optimization approaches to estimation of structural models”}},
  author={Iskhakov, Fedor and Lee, Jinhyuk and Rust, John and Schjerning, Bertel and Seo, Kyoungwon},
  journal={Econometrica},
  volume={84},
  number={1},
  pages={365--370},
  year={2016},
  publisher={Wiley Online Library}
}

@article{bray2019strong,
  title={{Strong convergence and dynamic economic models}},
  author={Bray, Robert L},
  journal={Quantitative Economics},
  volume={10},
  number={1},
  pages={43--65},
  year={2019},
  publisher={Wiley Online Library}
}

@article{ahlfeldt2015economics,
  title={{The economics of density: Evidence from the Berlin Wall}},
  author={Ahlfeldt, Gabriel M and Redding, Stephen J and Sturm, Daniel M and Wolf, Nikolaus},
  journal={Econometrica},
  volume={83},
  number={6},
  pages={2127--2189},
  year={2015},
  publisher={Wiley Online Library}
}

@book{bonnans2006numerical,
  title={{Numerical optimization: theoretical and practical aspects}},
  author={Bonnans, Joseph-Fr{\'e}d{\'e}ric and Gilbert, Jean Charles and Lemar{\'e}chal, Claude and Sagastiz{\'a}bal, Claudia A},
  year={2006},
  publisher={Springer Science \& Business Media}
}

@article{pal2023comparing,
  title={{Comparing procedures for estimating random coefficient logit demand models with a special focus on obtaining global optima}},
  author={P{\'a}l, L{\'a}szl{\'o} and S{\'a}ndor, Zsolt},
  journal={International Journal of Industrial Organization},
  volume={88},
  pages={102950},
  year={2023},
  publisher={Elsevier}
}

@article{lee2015computationally,
  title={{A computationally fast estimator for random coefficients logit demand models using aggregate data}},
  author={Lee, Jinhyuk and Seo, Kyoungwon},
  journal={The RAND Journal of Economics},
  volume={46},
  number={1},
  pages={86--102},
  year={2015},
  publisher={Wiley Online Library}
}

@article{varadhan2008simple,
  title={{Simple and globally convergent methods for accelerating the convergence of any EM algorithm}},
  author={Varadhan, Ravi and Roland, Christophe},
  journal={Scandinavian Journal of Statistics},
  volume={35},
  number={2},
  pages={335--353},
  year={2008},
  publisher={Wiley Online Library}
}

@article{arcidiacono2011conditional,
  title={{Conditional choice probability estimation of dynamic discrete choice models with unobserved heterogeneity}},
  author={Arcidiacono, Peter and Miller, Robert A},
  journal={Econometrica},
  volume={79},
  number={6},
  pages={1823--1867},
  year={2011},
  publisher={Wiley Online Library}
}

@article{aguirregabiria2021imposing,
  title={{Imposing equilibrium restrictions in the estimation of dynamic discrete games}},
  author={Aguirregabiria, Victor and Marcoux, Mathieu},
  journal={Quantitative Economics},
  volume={12},
  number={4},
  pages={1223--1271},
  year={2021},
  publisher={Wiley Online Library}
}

@article{conlon2020best,
  title={{Best practices for differentiated products demand estimation with PyBLP}},
  author={Conlon, Christopher and Gortmaker, Jeff},
  journal={The RAND Journal of Economics},
  volume={51},
  number={4},
  pages={1108--1161},
  year={2020},
  publisher={Wiley Online Library}
}

@article{aguirregabiria2007sequential,
  title     = {{Sequential estimation of dynamic discrete games}},
  author    = {Aguirregabiria, Victor and Mira, Pedro},
  journal   = {Econometrica},
  volume    = {75},
  number    = {1},
  pages     = {1--53},
  year      = {2007},
  publisher = {Wiley Online Library}
}

@article{berry1995automobile,
  title={Automobile prices in market equilibrium},
  author={Berry, Steven and Levinsohn, James and Pakes, Ariel},
  journal={Econometrica},
  pages={841--890},
  volume={63},
  number={4},
  year={1995},
  publisher={JSTOR}
}

@article{dube2012improving,
  title     = {Improving the numerical performance of static and dynamic aggregate discrete choice random coefficients demand estimation},
  author    = {Dub{\'e}, Jean-Pierre and Fox, Jeremy T and Su, Che-Lin},
  journal   = {Econometrica},
  volume    = {80},
  number    = {5},
  pages     = {2231--2267},
  year      = {2012},
  publisher = {Wiley Online Library}
}

@article{sun2019computationally,
  title={A computationally efficient fixed point approach to dynamic structural demand estimation},
  author={Sun, Yutec and Ishihara, Masakazu},
  journal={Journal of Econometrics},
  volume={208},
  number={2},
  pages={563--584},
  year={2019},
  publisher={Elsevier}
}

@article{gowrisankaran2012dynamics,
  title     = {Dynamics of consumer demand for new durable goods},
  author    = {Gowrisankaran, Gautam and Rysman, Marc},
  journal   = {Journal of Political Economy},
  volume    = {120},
  number    = {6},
  pages     = {1173--1219},
  year      = {2012},
  publisher = {University of Chicago Press Chicago, IL}
}

@article{aguirregabiria2002swapping,
  title={Swapping the nested fixed point algorithm: A class of estimators for discrete Markov decision models},
  author={Aguirregabiria, Victor and Mira, Pedro},
  journal={Econometrica},
  volume={70},
  number={4},
  pages={1519--1543},
  year={2002},
  publisher={Wiley Online Library}
}

@article{kasahara2012sequential,
  title={Sequential estimation of structural models with a fixed point constraint},
  author={Kasahara, Hiroyuki and Shimotsu, Katsumi},
  journal={Econometrica},
  volume={80},
  number={5},
  pages={2303--2319},
  year={2012},
  publisher={Wiley Online Library}
}

@article{pesendorfer2008asymptotic,
  title={Asymptotic least squares estimators for dynamic games},
  author={Pesendorfer, Martin and Schmidt-Dengler, Philipp},
  journal={The Review of Economic Studies},
  volume={75},
  number={3},
  pages={901--928},
  year={2008},
  publisher={Wiley-Blackwell}
}

@article{fukasawa2025jacobian,
  title={{Jacobian-free Efficient Pseudo-Likelihood (EPL) Algorithm}},
  author={Fukasawa, Takeshi},
  journal={Economics Letters},
  volume={112195},
  year={2025},
  publisher={Elsevier}
}

@article{rust1987optimal,
  title     = {Optimal replacement of GMC bus engines: An empirical model of Harold Zurcher},
  author    = {Rust, John},
  journal   = {Econometrica},
  volume={55},
  number={5},
  pages     = {999--1033},
  year      = {1987},
  publisher = {JSTOR}
}

\pagebreak{}

\setcounter{page}{1}

\begin{center}
{\Large\textbf{{}Supplemental Appendix to }}\\
{\Large\textbf{ ``Sequential algorithm for structural estimations with equilibrium constraints''}}{\Large\par}
\par\end{center}

\begin{center}
\medskip{}
\par\end{center}

\begin{center}
{}Takeshi Fukasawa 
\par\end{center}

\global\long\def\thesection{S\arabic{section}}%
 
\global\long\def\thetable{S\arabic{table}}%

\setcounter{section}{0}
\setcounter{footnote}{0}
\setcounter{table}{0}
\setcounter{lem}{0}
\setcounter{prop}{0}
\setcounter{algorithm}{0}
\setcounter{figure}{0}
\setcounter{equation}{0}
\setcounter{assumption}{0}

\renewcommand{\theassumption}{S\arabic{assumption}}
\renewcommand{\theequation}{S\arabic{equation}}
\renewcommand{\thetable}{S\arabic{table}}
\renewcommand{\thelem}{S\arabic{lem}}
\renewcommand{\theprop}{S\arabic{prop}}
\renewcommand{\theass}{S\arabic{ass}}

\renewcommand{\thefigure}{S\arabic{figure}}
\renewcommand{\thealgorithm}{S\arabic{algorithm}}

This document provides additional discussions (Section \ref{sec:Additional-discussions-Supp-Appendix}), extensions of the algorithms (Section \ref{sec:Extensions-Supp-Appendix}), details of the numerical experiments (Section \ref{sec:Details-numerical-Supp-Appendix}), and additional numerical results (Section \ref{sec:Additional-results-Supp-Appendix}). Proofs of the additional statements are presented in Section \ref{sec:Proof-of-additional-Supp-Appendix}.

\section{Additional discussions}\label{sec:Additional-discussions-Supp-Appendix}

This section provides additional discussions, including those on the stable convergence of SLC (Section \ref{subsec:Relations-to-Newton}), convexity of $\widetilde{Q}\left(\theta;\gamma\right)$ (Section \ref{subsec:Convexity-Q_tilde}), Neyman orthogonality (Section \ref{subsec:Further-discussions-on-Neyman}), Lagrangian-based methods (Section \ref{subsec:Additional_discussions_Lagrangian_based}), model misspecification (Section \ref{subsec:Model-misspecification}), choice of $Q$ and $G$ for SLC (Section \ref{subsec:Choice-of-Q-G-SLC}), and a comparison of SLC and EPL (Section \ref{subsec:Comparison-of-SLC-EPL}).

\subsection{Relations to Newton's method and stable convergence of the SLC/EPL algorithm}\label{subsec:Relations-to-Newton}

The SLC and EPL algorithms essentially inherit the idea of Newton's method for solving nonlinear equations. To illustrate this connection, consider a setting in which $G\left(Y;\theta\right)$ is linear in $\theta$, in which case the EPL and SLC updates coincide. Suppose that the parameters $\theta$ are not updated in the SLC/EPL iteration. Then, the algorithm reduces to the following iteration: $Y_{k+1}=Y_{k}-\left(\nabla_{Y}G\left(Y_{k};\theta_{0}\right)\right)^{-1}\left(G\left(Y_{k};\theta_{0}\right)\right)$. This iteration corresponds to standard Newton's method for solving the equation $G\left(Y;\theta_{0}\right)=0$ for $Y$ given the parameter $\theta_{0}$.

As noted in \citet{dearing2024efficient}, the convergence of Newton's method is notoriously unstable without regularization (such as a line-search procedure), particularly when the initial value is far from the solution, despite its quadratic local convergence rate. Nevertheless, the same study reports numerical results for dynamic discrete games without unobserved heterogeneity in which the EPL iterations converge in most cases even without additional regularization.

At first glance, these observations may appear surprising given the close connection to Newton's method. However, the discussion in Section \ref{sec:Sequential-algorithms} suggests that the SLC/EPL algorithm can exhibit stability properties that differ from those of Newton's method. The remainder of this subsection explores potential reasons for this behavior, and the present study provides additional numerical results related to this issue. Here, we focus on the setting where $G\left(Y;\theta\right)$ is linear with respect to $\theta$ and EPL is equivalent to SLC.

\subsubsection*{Theoretical analysis}

In general, the convergence of a fixed-point iteration is influenced by the sup norm of the Jacobian of the fixed-point mapping, so we focus on this Jacobian. Concerning the Newton iteration, the Jacobian is given by 
\[
-\frac{\partial\left(\left(\nabla_{Y}G\left(Y;\theta_{0}\right)\right)^{-1}\right)}{\partial Y}G\left(Y;\theta_{0}\right).
\]
Especially when the value of $G(Y;\theta_{0})$ is far from zero, the Jacobian can also be far from zero, which may lead to divergence of the iteration.

Concerning the SLC/EPL iteration using the mapping $\Upsilon\left(\theta;\gamma_{k}\equiv\left(\theta_{k},Y_{k}\right)\right)\equiv Y_{k}-\left(\nabla_{Y}G\left(Y_{k};\theta_{k}\right)\right)^{-1}\left(G\left(Y_{k};\theta\right)\right)$, the Jacobian of the SLC mapping, $\nabla_{\gamma}H\left(\gamma\right)$, depends on

\[
\nabla_{\gamma}\Upsilon\left(\theta;\gamma_{k}\equiv\left(\theta_{k},Y_{k}\right)\right)=\left(\begin{array}{c}
-\frac{\partial\left(\left(\nabla_{Y}G\left(Y_{k};\theta_{k}\right)\right)^{-1}\right)}{\partial\theta_{k}}G\left(Y_{k};\theta\right)\\
I-\frac{\partial\left(\left(\nabla_{Y}G\left(Y_{k};\theta_{k}\right)\right)^{-1}\right)}{\partial Y_{k}}G\left(Y_{k};\theta\right)-\left(\nabla_{Y}G\left(Y_{k};\theta_{k}\right)\right)^{-1}\left(\nabla_{Y}G\left(Y_{k};\theta\right)\right)
\end{array}\right).
\]

This suggests that a large magnitude of $G(Y_{k};\theta)$ may also lead to divergence of the iteration, as in the Newton iteration. However, in the case of SLC, each update is equivalent to solving the linearly constrained optimization problem (\ref{eq:SLC_constrained_opt}) in which the constraint is the linearization of $G(Y;\theta)$. As a result, the value of $G\left(Y;\theta\right)$ does not move too far from zero unless the linearization works poorly. Although the convergence of SLC/EPL is also influenced by other factors (e.g., the shape of the objective function), this discussion indicates that the convergence properties of the two algorithms differ. It is therefore not unusual that SLC/EPL can remain stable even in situations where Newton's method becomes unstable.

\subsubsection*{Numerical experiments}

Figure \ref{fig:Comparison-EPL-Newton} compares the convergence processes of the SLC/EPL and Newton iterations in \citet{pesendorfer2008asymptotic}'s dynamic game model without unobserved heterogeneity. Table 16 of \citet{dearing2024efficient} shows that the EPL iterations converge in over 98.5\% of all numerical experiments based on \citet{pesendorfer2008asymptotic}'s model. The current study utilizes the replication code of that paper to compare the EPL and Newton iterations.\footnote{\citet{dearing2024efficient} only showed results for the EPL iterations.} Note that the model satisfies the condition that $G$ is linear with respect to $\theta$, so EPL coincides with SLC. In this numerical experiment, we examine how the SLC/EPL and Newton iterations converge or diverge when starting from the same initial values.\footnote{The current study employs the following settings based on \citet{dearing2024efficient}: $N=250$; the data are generated from Equilibrium (ii); and the NPL estimator (\citealp{aguirregabiria2007sequential}) using consistently estimated CCPs is used as the initial $\theta$. The experiment is conducted using 100 randomly generated datasets.} Concerning SLC/EPL, the left panel shows that the values of $\left\Vert G\left(Y_{k};\theta_{k}\right)\right\Vert _{\infty}$ gradually decrease, and convergence is mostly stable. In contrast, the Newton iterations sometimes lead to divergence after a few iterations, and roughly 15\% of the trials fail to converge. 
\begin{center}
\begin{figure}[H]
\begin{centering}
\caption{Comparison of the SLC/EPL and Newton iterations}\label{fig:Comparison-EPL-Newton}
\includegraphics[scale=0.7]{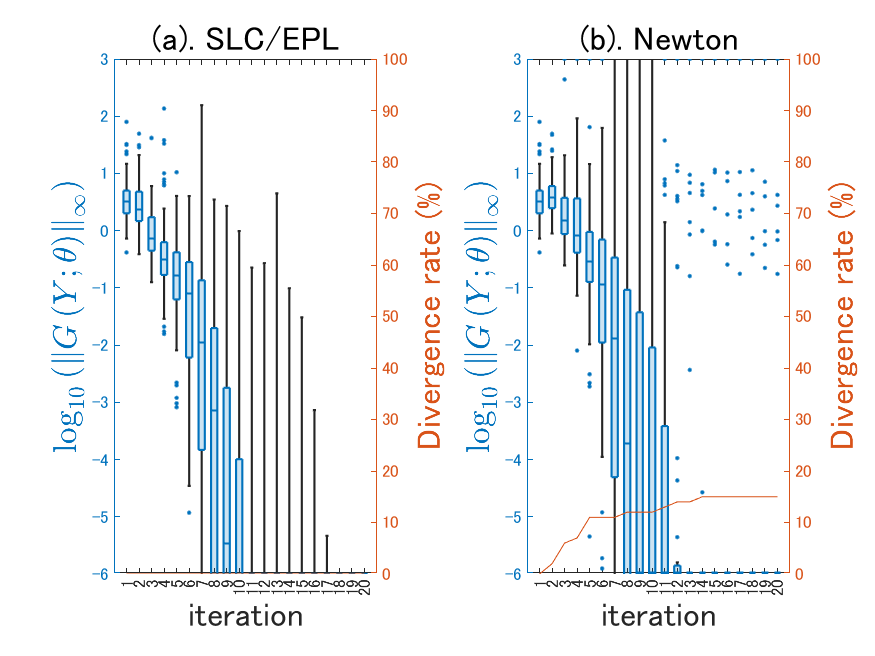}
\par\end{centering}
{\footnotesize Notes. The figure compares the convergence of the SLC/EPL and Newton iterations in \citet{pesendorfer2008asymptotic}'s dynamic game model without unobserved heterogeneity. The experiment is conducted using 100 randomly generated datasets. The box plot illustrates the distribution of the numerical error $\log_{10}\left(\left\Vert G\left(Y;\theta\right)\right\Vert _{\infty}\right)$ at each iteration, with values shown on the left axis. The polyline depicts the divergence rate across trials, corresponding to the right axis. Note that the SLC and EPL iterations coincide in the dynamic game model considered here due to the linearity of $G\left(Y;\theta\right)$ with respect to $\theta$.}{\footnotesize\par}
\end{figure}
\par\end{center}

\subsection{Convexity of $\widetilde{Q}\left(\theta;\gamma\right)$}\label{subsec:Convexity-Q_tilde}

In some specific cases, we can guarantee the convexity of $\widetilde{Q}$ without relying on any statistical properties. The following lemma provides one such instance:\footnote{The three conditions in the lemma are implicitly utilized in \citet{dearing2024efficient} (EPL for dynamic games without unobserved heterogeneity).} 
\begin{lem}
\label{lem:pd_Q_tilde_Hessian}Suppose Assumptions \ref{as:existence_constrained_opt_sol}, \ref{as:conti_diffble}, and \ref{as:mapping_cdns} and the following conditions hold:

(a). $Q\left(\theta,Y\right)$ does not depend on $\theta$.

(b). $\nabla_{YY^{\prime}}Q\left(\theta,Y\right)$ is positive definite.

(c). $\Upsilon\left(\theta;\gamma\right)$ is linear in $\theta$.

Then,

(1). $\nabla_{\theta\theta^{\prime}}\widetilde{Q}\left(\theta;\gamma\right)$ is positive semidefinite, and $\widetilde{Q}\left(\theta;\gamma\right)$ is convex.

(2). $\nabla_{\theta\theta^{\prime}}\widetilde{Q}\left(\theta;\gamma\right)$ is positive definite if the columns of $\nabla_{\theta}G\left(\gamma\right)$ are linearly independent and $\Upsilon\left(\theta;\gamma_{k}\right)\equiv Y_{k}-\left(\nabla_{Y}G\left(\gamma_{k}\right)\right)^{-1}\left(G\left(\gamma_{k}\right)\right)-\left(\nabla_{Y}G\left(\gamma_{k}\right)\right)^{-1}\left(\nabla_{\theta}G\left(\gamma_{k}\right)\right)\left(\theta-\theta_{k}\right)$ (the SLC mapping). 
\end{lem}

\subsubsection*{Linear independence of the columns of $\nabla_{\theta}G\left(\gamma\right)$}\label{subsubsec:Linear-independence}

Note that the assumption of linear independence of the columns of $\nabla_{\theta}G\left(\gamma\right)$ is natural for well-behaved models.

Let $\gamma$ be an element of $\Gamma_{FOC}$, which is a solution of the system of equations $\nabla_{\theta}Q\left(\theta,Y\right)-\left(\nabla_{Y}Q\left(\theta,Y\right)\right)\left(\nabla_{Y}G\left(Y;\theta\right)\right)^{-1}\left(\nabla_{\theta}G\left(Y;\theta\right)\right)=0$ and $G\left(Y;\theta\right)=0$. Suppose that linear independence of the columns of $\nabla_{\theta}G\left(\gamma\right)$ does not hold. Without loss of generality, suppose $\left(\begin{array}{cccc}
\nabla_{\theta_{1}}G\left(\gamma\right) & \cdots & \nabla_{\theta_{n_{\theta}-1}}G\left(\gamma\right) & \nabla_{\theta_{n_{\theta}}}G\left(\gamma\right)\end{array}\right)\left(\begin{array}{c}
d_{1}\left(\gamma\right)\\
\vdots\\
d_{n_{\theta}-1}\left(\gamma\right)\\
d_{n_{\theta}}\left(\gamma\right)
\end{array}\right)=0$ where $d_{n_{\theta}}\left(\gamma\right)\neq0$. Then, there exist scalars $c_{i}\left(\gamma\right)$ ($i=1,\cdots,n_{\theta}-1$) such that $\nabla_{\theta_{n_{\theta}}}G\left(\gamma\right)=\sum_{i=1}^{n_{\theta}-1}c_{i}\left(\gamma\right)\nabla_{\theta_{i}}G\left(\gamma\right)$.

If $Q$ does not depend on $\theta$,

\begin{eqnarray*}
 &  & \nabla_{\theta}Q\left(\theta,Y\right)-\left(\nabla_{Y}Q\left(\theta,Y\right)\right)\left(\nabla_{Y}G\left(Y;\theta\right)\right)^{-1}\left(\nabla_{\theta}G\left(Y;\theta\right)\right)=0\\
 & \Leftrightarrow & \left(\begin{array}{c}
\left(\nabla_{Y}Q\left(\theta,Y\right)\right)\left(\nabla_{Y}G\left(Y;\theta\right)\right)^{-1}\left(\nabla_{\theta_{1}}G\left(Y;\theta\right)\right)\\
\vdots\\
\left(\nabla_{Y}Q\left(\theta,Y\right)\right)\left(\nabla_{Y}G\left(Y;\theta\right)\right)^{-1}\left(\nabla_{\theta_{n_{\theta}-1}}G\left(Y;\theta\right)\right)\\
\sum_{i=1}^{n_{\theta}-1}c_{i}\left(\gamma\right)\left(\nabla_{Y}Q\left(\theta,Y\right)\right)\left(\nabla_{Y}G\left(Y;\theta\right)\right)^{-1}\left(\nabla_{\theta_{i}}G\left(Y;\theta\right)\right)
\end{array}\right)=\left(\begin{array}{c}
0\\
\vdots\\
0\\
0
\end{array}\right),
\end{eqnarray*}
and the equality for the last element can be represented as a weighted sum of the first $n_{\theta}-1$ equalities.

Because the elements of $\Gamma_{FOC}$ are $(n_{\theta}+n_{Y})$-dimensional, we need $n_{\theta}+n_{Y}$ independent equations to solve for $(\theta,Y)$. However, the discussion above implies that the number of independent equations can be at most $n_{\theta}+n_{Y}-1$, which may lead to ill-conditioning of the solution. Consequently, the linear independence condition should hold.

\subsection{Further discussions on Neyman orthogonality}\label{subsec:Further-discussions-on-Neyman}

Concerning Neyman-orthogonal score-based approaches, one can construct score functions other than those based on mappings with ZJP presented in the current paper. \citet{sawadogo2025efficient}, in the context of DDC models, proposes the Neyman-orthogonal score function $\phi\left(w;\theta,\eta\right)=\mu\frac{\partial Q\left(\theta,Y\right)}{\partial\theta}$ for an MLE problem characterized by the objective function $Q\left(\theta,Y\right)$, where the true value of $\mu$ is given by $\mu^{*}=I-\left(\nabla_{\theta Y^{\prime}}Q\right)^{\prime}\left(\left(\nabla_{\theta Y^{\prime}}Q\right)^{\prime}\left(\nabla_{\theta\theta^{\prime}}Q\right)^{-1}\left(\nabla_{\theta Y^{\prime}}Q\right)\right)^{-1}\left(\nabla_{\theta Y^{\prime}}Q\right)^{\prime}\left(\nabla_{\theta\theta^{\prime}}Q\right)^{-1}$.\footnote{All second derivatives are evaluated at $\left(\theta^{*},Y^{*}\right)$.} As shown by \citet{sawadogo2025efficient}, the score function $\phi$ is Neyman orthogonal with respect to the nuisance parameters $\left(Y,\mu\right)$, and the estimator $\widehat{\theta}$ using this score function satisfies $\sqrt{N}\left(\widehat{\theta}-\theta^{*}\right)\rightarrow_{d}N\left(0,\left(\nabla_{\theta\theta^{\prime}}Q^{*}\left(\theta^{*},Y^{*}\right)\right)^{-1}\right)$, as long as the first-stage estimate $\widehat{Y}$ and the estimate of $\mu$ are consistent.

Although both estimators satisfy Neyman orthogonality, their asymptotic efficiency may differ. For instance, the asymptotic variance of the Neyman-orthogonal score-based estimator is $\left(\nabla_{\theta\theta^{\prime}}Q^{*}\left(\theta^{*},Y^{*}\right)\right)^{-1}$ in \citet{sawadogo2025efficient}'s method, whereas this is in general different from the asymptotic variance of the NFXP estimator, $\left(\nabla_{\theta\theta^{\prime}}Q_{NFXP}^{*}\left(\theta^{*},Y^{*}\right)\right)^{-1}$,\footnote{As shown in the proof of Lemma \ref{lem:difference-Hessian-NFXP}, 
\begin{eqnarray*}
\nabla_{\theta\theta^{\prime}}Q_{NFXP}\left(\theta\right) & = & \left(\nabla_{\theta\theta^{\prime}}Q\left(\theta,\mathcal{Y}\left(\theta\right)\right)\right)+2\left(\nabla_{\theta Y^{\prime}}Q\left(\theta,\mathcal{Y}\left(\theta\right)\right)\right)\left(\nabla_{\theta}\mathcal{Y}\left(\theta\right)\right)+\\
 &  & \left(\nabla_{\theta^{\prime}}\mathcal{Y}\left(\theta\right)\right)\left(\nabla_{YY^{\prime}}Q\left(\theta,\mathcal{Y}\left(\theta\right)\right)\right)\left(\nabla_{\theta}\mathcal{Y}\left(\theta\right)\right)+\sum_{i=1}^{n_{Y}}\frac{\partial Q\left(\theta,\mathcal{Y}\left(\theta\right)\right)}{\partial Y_{i}}\left(\nabla_{\theta\theta^{\prime}}\mathcal{Y}_{i}\left(\theta\right)\right)
\end{eqnarray*}
holds.} which explicitly imposes the constraint $G\left(Y;\theta\right)=0$.

\subsection{Computational complexity and derivative requirements in Lagrangian-based methods}\label{subsec:Additional_discussions_Lagrangian_based}

Table \ref{tab:Comparison-of-algorithms} states that it is unclear whether Lagrangian-based methods (e.g., SQP, Interior Point (IP)) can be implemented without explicitly computing the Jacobian matrix $\nabla_{Y}G\left(Y;\theta\right)$. This section clarifies this point. Because original SQP/IP algorithms require explicit construction of the Jacobian matrix, which is computationally expensive in high-dimensional settings, previous studies in numerical analysis have investigated matrix-free approaches. One line of research focuses on SQP-based methods that only require Jacobian-vector products (JVPs) and vector-Jacobian products (VJPs). Examples include \citet{heinkenschloss2014matrix} and \citet{arreckx2018regularized}. Although promising, computing VJPs can sometimes be challenging in practice; unlike JVPs, VJPs cannot be efficiently approximated by standard numerical derivatives in high-dimensional spaces, typically requiring automatic differentiation or analytical derivation. In contrast, the SLC algorithm only requires JVPs, which can be easily and efficiently approximated by directional numerical derivatives.

Another line of research involves derivative-free optimization (DFO), which avoids gradient computations entirely. This includes model-based methods, which construct local approximations of the objective or Lagrangian functions (e.g., COBYLA), and direct search (or model-free) methods, which update parameters without explicit functional approximations.\footnote{See \citet{larson2019derivative} for a survey.} Additionally, \citet{chen2025derivative} proposed an IP-based algorithm for constrained nonlinear least squares that approximates the full Jacobian using finite differences. However, since these methods require many function evaluations per iteration, their scalability to high-dimensional problems---common in structural estimation in economics---remains a concern.

\subsection{Model misspecification}\label{subsec:Model-misspecification}

This section discusses how the theoretical results change when we allow for model misspecification.

First, the discussion in Section \ref{sec:Sequential-algorithms} still holds without any modification, as it does not rely on any statistical properties. Concerning the discussion in Section \ref{sec:structural-estimation}, the property $\nabla_{Y}Q\left(\widehat{\theta},\widehat{Y}\right)=O_{p}\left(N^{-\frac{1}{2}}\right)$ may not hold in general. Here, consider MLE\footnote{An analogous discussion also holds for GMM.}{} and suppose that $\gamma^{*}$ is the solution of the following constrained optimization problem:

\[
\min_{\gamma}-E_{w}\left[\ln f\left(w|\gamma\right)\right]\ \text{s.t.}\ G\left(\gamma\right)=0.
\]

Even in this setting, $\widehat{\gamma}\rightarrow_{p}\gamma^{*}$ can hold under regularity conditions. However, the true distribution, $f^{*}(w)$, can be distinct from $f(w|\gamma^{*})$. Then, the third term on the right-hand side of (\ref{eq:dQ_MLE}) cannot be equal to zero under model misspecification, although the first and second terms still satisfy $O_{p}\left(N^{-\frac{1}{2}}\right)$. This suggests that $\nabla_{Y}Q\left(\widehat{\theta},\widehat{Y}\right)=O_{p}\left(N^{-\frac{1}{2}}\right)$ may not hold. As this property is key to showing Propositions \ref{prop: asy_property}, \ref{prop: asy_property_consistent_initial_values}, and the (asymptotic) Neyman orthogonality results, these statements may not hold in general under model misspecification.

Note that the NFXP/MPEC estimator is still a fixed point of ZJP-based sequential algorithms even under such a setting. In addition, when model misspecification is of moderate size, we can expect the local convergence speed to remain close to quadratic (or linear with a small modulus) in large samples.

\subsection{Choice of $Q$ and $G$ for SLC}\label{subsec:Choice-of-Q-G-SLC}

Although SLC is designed for general structural estimation problems with equilibrium constraints, we should be careful about the choice of $Q$ and $G$ to ensure good performance of the algorithm. This section clarifies these considerations.

First, in finite samples, the convergence speed of SLC is not necessarily quadratic, and Proposition \ref{prop:local_conv} suggests that the values of $\nabla_{\theta\gamma^{\prime}}\Upsilon\left(\theta,\gamma\right)$ and $\nabla_{Y}Q\left(\theta,Y\right)$ largely affect the convergence speed. The term $\nabla_{\theta\gamma^{\prime}}\Upsilon\left(\theta,\gamma\right)$ governs the nonlinearity of the function $G$ with respect to $\gamma\equiv\left(\theta,Y\right)$.\footnote{For the SLC mapping $\Upsilon\left(\theta;\gamma_{k}\right)\equiv Y_{k}-\left(\left(\nabla_{Y}G\left(Y_{k};\theta_{k}\right)\right)^{-1}\left(G\left(Y_{k};\theta_{k}\right)\right)\right)-\left(\left(\nabla_{Y}G\left(Y_{k};\theta_{k}\right)\right)^{-1}\left(\nabla_{\theta}G\left(Y_{k};\theta_{k}\right)\right)\right)\left(\theta-\theta_{k}\right)$, $\nabla_{\theta\gamma^{\prime}}\Upsilon\left(\theta,\gamma\right)=0$ holds when $G$ is linear with respect to $Y$ and $\theta$. Intuitively, SLC repeatedly solves a linearly constrained approximated subproblem, so convergence can be faster when the constraint function is close to linear.}{} Regarding $\nabla_{Y}Q\left(\theta,Y\right)$, this quantity converges to 0 in probability in large samples, but it can be far from 0 in finite samples.

In addition, SLC involves an optimization step, and computationally inexpensive evaluation of $Q$ relative to $G$ is important for performance, as discussed in Section \ref{subsec:Dynamic-BLP-numerical} in the context of the dynamic BLP model and Appendix \ref{subsec:Comparison-of-computational-costs}.

In light of these points, the following considerations are important for the strong (i.e., computationally efficient and convergent) performance of SLC: 
\begin{itemize}
\item $G$ should be as close to linear with respect to $(\theta,Y)$ as possible. 
\item The dependence of $Q$ on $Y$ should be as small as possible. 
\item The computational cost of evaluating $Q$ should be low compared to that of $G$. 
\end{itemize}
Generally, the appropriate choice of these functions is model-specific. In the dynamic game and dynamic BLP models considered in the current paper, the choices of $Q$ and $G$ described in Section \ref{sec:Numerical-experiments} work well. The considerations outlined above should be helpful for practitioners evaluating the potential applicability of SLC to other models, and I leave further exploration of models where SLC is effective to future research.\footnote{Previously, I examined the applicability of SLC to demand-side-only static BLP models. When mean product utilities are treated as $Y$, $Q$ can be expressed as a quadratic function of $Y$, and the solution to the SLC optimization step can be obtained analytically at a largely negligible computational cost. Although the algorithm usually converged, it occasionally diverged, particularly when the sample size was small. In that specification, $Q$ depends only on $Y$, and this dependence can be substantial, which may partially explain why that specification does not always perform well. Note, however, that further refinements of SLC may resolve this issue, and I leave such investigations to future research.}

\subsection{Comparison of SLC and EPL}\label{subsec:Comparison-of-SLC-EPL}

This section examines the computational costs of SLC and EPL. The numerical results indicate that EPL is generally substantially more computationally demanding than SLC, except in special cases where the two methods are equivalent.

\subsubsection*{Analytical derivative-based}

In both methods, we compute $J_{k}\equiv\nabla_{Y}G\left(Y_{k};\theta_{k}\right)$ and its factorization once per outer iteration.

In SLC, we additionally compute $\nabla_{\theta}G\left(Y_{k};\theta_{k}\right)$ once. We further compute $a_{k}\equiv-\left(\nabla_{Y}G\left(Y_{k};\theta_{k}\right)\right)^{-1}\left(G\left(Y_{k};\theta_{k}\right)\right)$ and $b_{k}\equiv-\left(\nabla_{Y}G\left(Y_{k};\theta_{k}\right)\right)^{-1}\left(\nabla_{\theta}G\left(Y_{k};\theta_{k}\right)\right)$ once. In the optimization step, the value of $Q\left(\theta,Y_{k}+a_{k}+b_{k}\left(\theta-\theta_{k}\right)\right)$ is evaluated for candidate values of $\theta$. Finally, $Y_{k+1}=Y_{k}+a_{k}+b_{k}\left(\theta_{k+1}-\theta_{k}\right)$ is computed.

In contrast, in EPL, we compute $G_{k}\left(\theta\right)\equiv G\left(Y_{k};\theta\right)$, $\nabla_{\theta}G\left(Y_{k};\theta\right)$, $a_{k}\left(\theta\right)\equiv-J_{k}^{-1}G_{k}\left(\theta\right)$, and $b_{k}\left(\theta\right)\equiv-J_{k}^{-1}\left(\nabla_{\theta}G\left(Y_{k};\theta\right)\right)$ for each candidate $\theta$ in the optimization step. In addition, the value of $Q\left(\theta,Y_{k}+a_{k}\left(\theta\right)\right)$ is evaluated for candidate values of $\theta$ in the optimization step. Finally, $Y_{k+1}=Y_{k}-J_{k}^{-1}G_{k}\left(\theta_{k+1}\right)$ is computed.

This implies that the computational cost of EPL can be higher due to the requirement to evaluate $G\left(Y_{k};\theta\right)$ and its partial derivatives with respect to $\theta$ for each candidate $\theta$ in the optimization step.

\subsubsection*{Numerical derivative-based}

In SLC, we compute $G\left(Y_{k};\theta_{k}\right)$ and $\nabla_{\theta}G\left(Y_{k};\theta_{k}\right)$ once. In addition, we compute $a_{k}\equiv-\left(\nabla_{Y}G\left(Y_{k};\theta_{k}\right)\right)^{-1}\left(G\left(Y_{k};\theta_{k}\right)\right)$ and $b_{k}\equiv-\left(\nabla_{Y}G\left(Y_{k};\theta_{k}\right)\right)^{-1}\left(\nabla_{\theta}G\left(Y_{k};\theta_{k}\right)\right)$ via a Krylov method once. The number of linear systems to be solved by the Krylov method is $n_{\theta}+1$. In the optimization step, the value of $Q\left(\theta,Y_{k}+a_{k}+b_{k}\left(\theta-\theta_{k}\right)\right)$ is evaluated for candidate values of $\theta$. Finally, $Y_{k+1}=Y_{k}+a_{k}+b_{k}\left(\theta_{k+1}-\theta_{k}\right)$ is computed.

In contrast, in EPL, we compute $G\left(Y_{k};\theta\right)$ for each candidate $\theta$ in the optimization step. In addition, we compute $a_{k}\left(\theta\right)\equiv-\left(\nabla_{Y}G\left(Y_{k};\theta_{k}\right)\right)^{-1}\left(G\left(Y_{k};\theta\right)\right)$ via a Krylov method for each candidate $\theta$. Given that the number of function evaluations is roughly $\left(2n_{\theta}+1\right)\cdot n_{opt}$, where $n_{opt}$ denotes the number of optimization iterations, the number of linear systems to be solved by the Krylov method is $(2n_{\theta}+1)\cdot n_{opt}$. In the optimization step, the value of $Q\left(\theta,Y_{k}+a_{k}\left(\theta\right)\right)$ is evaluated for candidate values of $\theta$. Finally, $Y_{k+1}=Y_{k}-J_{k}^{-1}G_{k}\left(\theta_{k+1}\right)$ is computed.

This implies that the computational cost of EPL can be much higher, as the method requires repeatedly solving linear systems for each candidate $\theta$.

\subsubsection*{Numerical results}

Table \ref{tab:Comparison-SLC-EPL} compares the computational costs of SLC and EPL in the dynamic game estimation with unobserved heterogeneity examined in Section \ref{subsec:Dynamic-discrete-game-numerical}. Consistent with the theoretical discussion above, EPL is substantially more computationally demanding than SLC. 
\begin{center}
{} 
\begin{table}[H]
\caption{Comparison between SLC and EPL (Dynamic game with unobserved heterogeneity)}\label{tab:Comparison-SLC-EPL}

\begin{centering}
{\footnotesize{}%
\begin{tabular}{cccccc}
\hline 
\multirow{1}{*}{{\footnotesize Analy. diff.}} & \multirow{1}{*}{{\footnotesize Algorithm}} & \multirow{1}{*}{{\footnotesize Comp. Time (sec)}} & \multirow{1}{*}{{\footnotesize Feval $Q$}} & \multirow{1}{*}{{\footnotesize Feval $G$}} & {\footnotesize\# of}\tabularnewline
 &  &  &  &  & {\footnotesize Main iter.}\tabularnewline
\hline 
\multirow{2}{*}{{\footnotesize Yes}} & {\footnotesize SLC} & {\footnotesize 2.1} & {\footnotesize 8} & {\footnotesize 239} & {\footnotesize 7}\tabularnewline
 & {\footnotesize EPL} & {\footnotesize 6.8} & {\footnotesize 594} & {\footnotesize 816} & {\footnotesize 8}\tabularnewline
\multirow{2}{*}{{\footnotesize No}} & {\footnotesize SLC} & {\footnotesize 7.4} & {\footnotesize 5036} & {\footnotesize 3056} & {\footnotesize 7}\tabularnewline
 & {\footnotesize EPL} & {\footnotesize 168.5} & {\footnotesize 99997} & {\footnotesize 6328} & {\footnotesize 6}\tabularnewline
\hline 
\end{tabular}}{\footnotesize\par}
\par\end{centering}
{\footnotesize Notes. Based on a randomly generated dataset with a random initial starting point. ``Feval'' denotes the number of function evaluations.}{\footnotesize\par}
\end{table}
\par\end{center}

\section{Extensions}\label{sec:Extensions-Supp-Appendix}

This section describes extensions of the algorithms, including the spectral algorithm (Section \ref{subsec:Spectral-algorithm}), inequality constraints (Section \ref{subsec:Inequality-constraints}), and a line-search SLC (Section \ref{subsec:Stabilized-algorithm}).

\subsection{Spectral algorithm for accelerating and stabilizing convergence}\label{subsec:Spectral-algorithm}

Although sequential algorithms based on ZJP exhibit nearly quadratic local convergence in large samples, their convergence may slow down when the sample size is not sufficiently large. Moreover, near-quadratic local convergence does not guarantee that the iteration is contractive when the initial values are far from the solution. To accelerate and stabilize convergence, one can incorporate acceleration methods for fixed\nobreakdash-point iterations, particularly the spectral algorithm.

The spectral algorithm was originally developed as a method for solving nonlinear equations. Its design essentially draws on the idea of Newton's method, particularly in the choice of step size; see also the discussions in \citet{varadhan2008simple}, \citet{aguirregabiria2021imposing}, and \citet{fukasawa2024fast}. Algorithm \ref{alg:SLC-spectral} outlines the procedure for the sequential algorithm incorporating the spectral algorithm.

\begin{algorithm}[H]
\caption{Sequential algorithm with the spectral algorithm}\label{alg:SLC-spectral}

Set initial values $\gamma_{0}\equiv\left(\theta_{0},Y_{0}\right)$. Iterate the following until convergence $(k=0,1,2,\cdots)$:
\begin{enumerate}
\item Compute $\widetilde{\theta_{k+1}}=\arg\min_{\theta}Q\left(\theta,\Upsilon\left(\theta;\gamma_{k}\right)\right)$.
\item Compute $\widetilde{Y_{k+1}}=\Upsilon\left(\widetilde{\theta_{k+1}};\gamma_{k}\right)$. Let $F(\gamma_{k})\equiv H\left(\left(\begin{array}{c}
\theta_{k}\\
Y_{k}
\end{array}\right)\right)-\left(\begin{array}{c}
\theta_{k}\\
Y_{k}
\end{array}\right)=\left(\begin{array}{c}
\widetilde{\theta_{k+1}}-\theta_{k}\\
\widetilde{Y_{k+1}}-Y_{k}
\end{array}\right)$.
\item Compute the spectral coefficient $\sigma_{k}=\frac{\left\Vert s_{k}\right\Vert _{2}}{\left\Vert y_{k}\right\Vert _{2}}$ if $k\geq1$, where $s_{k}\equiv\gamma_{k}-\gamma_{k-1}$ and $y_{k}\equiv F\left(\gamma_{k}\right)-F\left(\gamma_{k-1}\right)$. Let $\sigma_{k}=1$ if $k=0$.
\item Compute $\theta_{k+1}\leftarrow\sigma_{k}\widetilde{\theta_{k+1}}+\left(1-\sigma_{k}\right)\theta_{k}$ and $Y_{k+1}\leftarrow\sigma_{k}\widetilde{Y_{k+1}}+\left(1-\sigma_{k}\right)Y_{k}$.
\end{enumerate}
\end{algorithm}

As shown in \citet{aguirregabiria2021imposing} in the context of the NPL mapping for dynamic discrete games, which may not be a contraction even in a neighborhood of the solution, the use of the spectral algorithm substantially stabilizes and accelerates convergence.

In Step 3 of Algorithm \ref{alg:SLC-spectral}, we use the step size $\sigma_{k}=\frac{\left\Vert s_{k}\right\Vert _{2}}{\left\Vert y_{k}\right\Vert _{2}}\geq0$. As discussed in \citet{fukasawa2024fast},\footnote{See also \citet{varadhan2008simple}.} other choices of step size such as $\sigma_{k}=-\frac{s_{k}^{\prime}s_{k}}{s_{k}^{\prime}y_{k}}$ and $\sigma_{k}=-\frac{s_{k}^{\prime}y_{k}}{y_{k}^{\prime}y_{k}}$ are not desirable when the mapping $H\left(\gamma\right)$ is a contraction, because choosing a step size $\sigma_{k}<0$ may lead to divergence of the iteration in the absence of stabilization strategies.

In principle, one may use other acceleration methods for fixed-point iterations, such as Anderson acceleration or SQUAREM, which have been applied in economics and other fields.\footnote{These acceleration methods have been used to accelerate the BLP contraction mapping in static BLP estimation; see \citet{conlon2020best} and \citet{fukasawa2024fast}. In the case of the inner loops of static and dynamic BLP estimations, Anderson acceleration is faster than the spectral and SQUAREM algorithms, as discussed in \citet{fukasawa2024fast}.} However, even when the SLC algorithm is used, the update direction may not be a descent direction with respect to the merit function $\phi_{1}\left(\theta,Y;\mu\right)\equiv Q\left(\theta,Y\right)+\mu\left\Vert G\left(Y;\theta\right)\right\Vert _{1}$. Although Supplemental Appendix \ref{subsec:Stabilized-algorithm} shows that global convergence of the SLC algorithm combined with the spectral algorithm can be ensured under mild conditions by incorporating line\nobreakdash-search steps, it remains unclear whether similar guarantees apply when the SLC algorithm is combined with acceleration methods other than the spectral algorithm. Hence, when global convergence is of primary concern, the spectral algorithm may be preferable as the acceleration method for the SLC algorithm.

\subsection{Inequality constraints on $\theta$}\label{subsec:Inequality-constraints}

In applications of structural estimation, practitioners sometimes impose inequality constraints on the domain of structural parameters $\theta$ (e.g., nonnegativity constraints). Hence, in this subsection we consider the setting where the domain of $\theta$ is restricted to a convex set $\Omega\subset\mathbb{R}^{n_{\theta}}$. Let $B\left(\theta\right)\geq0$ be the inequality corresponding to the convex set $\Omega$. Then, the constrained optimization problem reduces to solving the following constrained optimization problem with both equality and inequality constraints:

\begin{eqnarray}
\min_{\theta\in\mathbb{R}^{n_{\theta}},Y\in\mathbb{R}^{n_{Y}}} & Q(\theta,Y)\label{eq:constrained_opt_problem_ineq}\\
\text{s.t.} & G(Y;\theta)=0\nonumber \\
 & B\left(\theta\right)\geq0\nonumber 
\end{eqnarray}

Even when we impose this constraint on the domain of $\theta$, we can continue using Algorithm \ref{alg:general_seq} by letting $\theta_{k+1}=\arg\min_{\theta\in\Omega\subset\mathbb{R}^{n_{\theta}}}\widetilde{Q}\left(\theta;\gamma_{k}\right)\equiv Q\left(\theta,\Upsilon\left(\theta;\gamma_{k}\right)\right)$ in Step 1.\footnote{When we additionally introduce the spectral algorithm or other fixed-point iteration algorithms, $\sigma_{\theta,k}\widetilde{\theta_{k+1}}+\left(1-\sigma_{\theta,k}\right)\theta_{k}$ should be projected onto $\Omega$.}

Let $\Gamma_{KKT}$ be the set of $\gamma\equiv\left(\theta,Y\right)$ satisfying the Karush--Kuhn--Tucker (KKT) conditions of problem (\ref{eq:constrained_opt_problem_ineq}). Then, the following proposition shows the relationship between the fixed points of the sequential algorithms and the KKT conditions. 
\begin{prop}
\label{prop:fixed-points-KKT-seq-ineq}Suppose Assumptions \ref{as:existence_constrained_opt_sol}, \ref{as:conti_diffble}, and \ref{as:mapping_cdns} hold. Then,

(a). $\Gamma_{seq}\subset\Gamma_{KKT}$.

(b). $\Gamma_{seq}=\Gamma_{KKT}$ if $\widetilde{Q}\left(\theta;\gamma\right)$ is strictly convex in $\theta$.

(c). $\widehat{\gamma}\in\Gamma_{seq}$ holds if $\nabla_{\theta\theta^{\prime}}\widetilde{Q}\left(\theta;\gamma\right)$ is positive definite. 
\end{prop}

\subsection{Stabilized SLC}\label{subsec:Stabilized-algorithm}

This subsection shows that we can attain global convergence of SLC by introducing a line-search procedure. Note that the following discussion does not rely on any statistical properties. Let $\phi_{1}\left(\gamma;\mu\right)\equiv Q\left(\theta,Y\right)+\mu\left\Vert G\left(Y;\theta\right)\right\Vert _{1}$ be the merit function.\footnote{In the case of SQP for solving constrained optimization problems with inequality constraints, the alternative merit function $\widetilde{\phi_{1}}\left(\gamma;\mu\right)\equiv Q\left(\theta,Y\right)+\mu\left\Vert G\left(Y;\theta\right)\right\Vert _{1}+\mu\sum_{i=1}^{n_{B}}\left[B_{i}\left(\theta\right)\right]^{-}$, where $\left[B_{i}\left(\theta\right)\right]^{-}=\max\left\{ 0,-B_{i}\left(\theta\right)\right\} $, is often used (cf. \citealp{nocedal2006numerical}). However, in our algorithm, $B\left(\theta_{k}\right)\geq0$, $B\left(\widetilde{\theta_{k+1}}\right)\geq0$, and $B\left(\theta_{k}+\sigma_{k}d_{\theta,k}\right)\geq0$ hold by construction, and $\left[B_{i}\left(\theta\right)\right]^{-}\equiv\max\left\{ 0,-B_{i}\left(\theta\right)\right\} =0$ holds for any $\theta=\theta_{k}+\alpha_{k}\sigma_{k}d_{\theta,k}$ because the subset $\Theta$, characterizing the function $B$, is convex. Consequently, we do not introduce the term $\mu\sum_{i=1}^{n_{B}}\left[B_{i}\left(\theta\right)\right]^{-}$ in the merit function $\phi_{1}\left(\gamma;\mu\right)$.} Algorithm \ref{alg:SLC-line-search} outlines the steps:\footnote{The updating rules for $\mu_{k}$ are motivated by those in \citet{nocedal2006numerical} for the SQP algorithm. Because the SLC algorithm without stabilization does not require computing Lagrange multipliers, update schemes for $\mu_{k}$ that rely on multiplier information---commonly used in SQP---are not desirable in our setting.}

\begin{algorithm}[H]
\caption{Line-search SLC (with spectral algorithm)}\label{alg:SLC-line-search}

Set initial values $\gamma_{0}\equiv\left(\theta_{0},Y_{0}\right)$ such that $B\left(\theta_{0}\right)\geq0$, parameter values $\mu_{0}\geq\underline{\mu}>0,\ c>0,\ \zeta\in(0,1),\ \tau_{\alpha}\in(0,1)$, and a sequence $\left\{ \eta_{k}\geq0\right\} _{k=0,1,\dots}$. Iterate the following until convergence $(k=0,1,2,\cdots)$:
\begin{enumerate}
\item Compute $Z_{1,k}\equiv Y_{k}-\left(\nabla_{Y}G\left(Y_{k};\theta_{k}\right)\right)^{-1}\left(G\left(Y_{k};\theta_{k}\right)\right)$ and $Z_{2,k}\equiv-\left(\nabla_{Y}G\left(Y_{k};\theta_{k}\right)\right)^{-1}\left(\nabla_{\theta}G\left(Y_{k};\theta_{k}\right)\right)$.
\item Compute $\widetilde{\theta_{k+1}}=\arg\min_{\theta\in\Theta}Q\left(\theta,Z_{1,k}+Z_{2,k}\left(\theta-\theta_{k}\right)\right)$.
\item Compute $\widetilde{Y_{k+1}}=Z_{1,k}+Z_{2,k}\left(\widetilde{\theta_{k+1}}-\theta_{k}\right)$.
\item Compute $d_{\theta,k}\equiv\widetilde{\theta_{k+1}}-\theta_{k}$ and $d_{Y,k}\equiv\widetilde{Y_{k+1}}-Y_{k}$.
\item (Spectral algorithm) Compute the spectral coefficient $\sigma_{k}$ such that $\theta_{k}+\sigma_{k}d_{\theta,k}\in\Theta$.
\item (For guaranteeing global convergence) Choose $\mu_{k}$ such that

\[
\mu_{k}=\max\left\{ \mu_{k-1},\frac{\left(\nabla_{\gamma}Q\left(\gamma_{k}\right)\right)d_{k}+c\left\Vert d_{k}\right\Vert ^{2}}{\left\Vert G\left(Y_{k};\theta_{k}\right)\right\Vert _{1}+\eta_{k}}\right\} .
\]

\item \label{enu:backtrack-SLC}(Backtracking) Let $\alpha_{k}\leftarrow1$. Iterate the following:
\begin{enumerate}
\item Compute $\Delta_{k}\equiv D\left(\phi_{1}\left(\theta_{k},Y_{k};\mu_{k}\right);d_{k}\right)=\left(\nabla_{\gamma}Q\left(\gamma_{k}\right)\right)d_{k}-\mu_{k}\left\Vert G\left(Y_{k};\theta_{k}\right)\right\Vert _{1}$.
\item If $\phi_{1}\left(\gamma_{k}+\alpha_{k}\sigma_{k}d_{k};\mu_{k}\right)<\phi_{1}\left(\gamma_{k};\mu_{k}\right)+\zeta\alpha_{k}\sigma_{k}\Delta_{k}+(1-\zeta)\alpha_{k}\sigma_{k}\mu_{k}\eta_{k}$, set $\gamma_{k+1}\leftarrow\gamma_{k}+\alpha_{k}\sigma_{k}d_{k}$ and return to Step 1.
\item Reset $\alpha_{k}\leftarrow\tau_{\alpha}\alpha_{k}$ and return to Step \ref{enu:backtrack-SLC}(b).
\end{enumerate}
\end{enumerate}
\end{algorithm}

To show the global convergence of the algorithm, we impose the following two assumptions.

\begin{assumption}\par The solution of $\min_{\theta\in\Theta}\widetilde{Q}\left(\theta;\gamma_{k}\right)=Q\left(\theta,\Upsilon\left(\theta;\gamma_{k}\right)\right)$ is unique.\par \label{as:unique_sol_Q_tilde} \end{assumption}

\begin{assumption}[Boundedness]\par The sequences $\left\{ \gamma_{k}\right\} _{k\in\mathbb{N}}$ and $\left\{ \mu_{k}\right\} _{k\in\mathbb{N}}$ are bounded. In addition, there exists $\overline{K}\in\mathbb{N}$ such that $\mu_{k}=\overline{\mu}\ \forall k\geq\overline{K}$ for some constant $\overline{\mu}$.\par \label{as:boundedness} \end{assumption}

Assumption \ref{as:unique_sol_Q_tilde} is satisfied if $\widetilde{Q}\left(\theta;\gamma\right)$ is strictly convex. The boundedness of the iterates (Assumption \ref{as:boundedness}) is standard in the SQP literature.

The following lemma shows that the backtracking step terminates in a finite number of steps. 
\begin{lem}
\label{lem:finite-terminate-backtracking}There exists $\alpha_{k}\in(0,1]$ such that $\phi_{1}\left(\gamma_{k}+\alpha_{k}\sigma_{k}d_{k};\mu_{k}\right)<\phi_{1}\left(\gamma_{k};\mu_{k}\right)+\zeta\alpha_{k}\sigma_{k}\Delta_{k}+(1-\zeta)\alpha_{k}\sigma_{k}\mu_{k}\eta_{k}$ for any $k\in\mathbb{N}$ and any tuning parameter $\eta_{k}\geq0$. 
\end{lem}
Then, we obtain the following global convergence result: 
\begin{prop}
(Global convergence)\label{prop:Global-convergence}

Under Assumptions \ref{as:existence_constrained_opt_sol}, \ref{as:conti_diffble}, \ref{as:mapping_cdns}, \ref{as:nonsingular_Jacobian}, \ref{as:unique_sol_Q_tilde}, and \ref{as:boundedness}, the sequence $\left\{ \gamma_{k}\equiv\left(\theta_{k},Y_{k}\right)\right\} _{k=1,2,\cdots}$ generated by the line-search SLC algorithm (Algorithm \ref{alg:SLC-line-search}) converges to a point in $\Gamma_{FOC}$. 
\end{prop}
\begin{rem}
In principle, we can set $\eta_{k}=0$. In this case, we can guarantee that $\Delta_{k}<0$ for a sufficiently large $\mu_{k}$. However, when the iterate $\gamma_{k}$ satisfies $G(\gamma_{k})\approx0$ even while $\gamma_{k}$ is far from $\widehat{\gamma}$, the value of $\mu_{k}$, which must be at least $\frac{\left(\nabla_{\gamma}Q\left(\gamma_{k}\right)\right)d_{k}+c\left\Vert d_{k}\right\Vert ^{2}}{\left\Vert G\left(Y_{k};\theta_{k}\right)\right\Vert _{1}+\eta_{k}}$, can become very large. It is known that a large penalty parameter $\mu_{k}$ slows down the convergence of iterations in SQP (cf. \citealp{nocedal2006numerical}), which is undesirable. Consequently, the current study introduces the term $\eta_{k}$ so that $\mu_{k}$ does not become excessively large.

Note that setting $\eta_{k}>0$ also acts as a nonmonotone line search. The presence of the term $(1-\zeta)\alpha_{k}\sigma_{k}\mu_{k}\eta_{k}$ in Step 7(b) means that a step size $\alpha_{k}$ is accepted even if the merit function value increases slightly. In the numerical optimization literature, it is well known that line-search SQP methods with a monotone line-search procedure may suffer from the Maratos effect, whereby the line-search step inhibits fast local convergence. Nonmonotone line-search strategies have been applied to mitigate this issue (cf. \citealp{nocedal2006numerical}), and the present study adopts such an approach. 
\end{rem}

\section{Details of the numerical experiments}

\label{sec:Details-numerical-Supp-Appendix}

\subsection{Dynamic game with time-varying unobserved heterogeneity}

\label{subsec:Dynamic-game-details}

We partition the state $\omega$ into two components: the observed state $x$ and the unobserved state $s$. Let $v$ be the choice-specific value functions, and let $\overline{p_{1}}\equiv\left(\mathcal{L}\left(x_{1},s_{1}\right)\right)_{(x_{1},s_{1})}$ be the initial distribution.

The maximum likelihood estimator maximizes the following likelihood function:

\begin{eqnarray*}
 &  & \mathcal{L}\left(a,x|x_{1};\theta_{u},\pi,v,\overline{p_{1}}\right)\\
 & = & \sum_{s_{1}\in\mathcal{S}}\sum_{s_{2}\in\mathcal{S}}\cdots\sum_{s_{T}\in\mathcal{S}}\mathcal{L}\left(a_{1},\cdots,a_{T},x_{1},\cdots,x_{T},s_{1},\cdots,s_{T}|x_{1};\theta_{u},\pi,v,\overline{p_{1}}\right)\\
 & = & \sum_{s_{1}\in\mathcal{S}}\sum_{s_{2}\in\mathcal{S}}\cdots\sum_{s_{T}\in\mathcal{S}}\mathcal{L}\left(s_{1}|x_{1};\overline{p_{1}}\right)\mathcal{L}\left(a_{1},x_{2}|x_{1},s_{1};\theta_{u},v\right)\cdot\\
 &  & \mathcal{L}\left(s_{2}|s_{1};\pi\right)\mathcal{L}\left(a_{2},x_{3}|x_{2},s_{2};\theta_{u},v\right)\cdots\\
 &  & \mathcal{L}\left(s_{T}|s_{T-1};\pi\right)\mathcal{L}\left(a_{T}|x_{T-1},s_{T-1};\theta_{u},v\right)\\
 & = & \sum_{s_{T}\in\mathcal{S}}\mathcal{L}\left(a_{T}|x_{T},s_{T};\theta_{u},v\right)\sum_{s_{T-1}\in\mathcal{S}}\mathcal{L}\left(s_{T}|s_{T-1};\pi\right)\sum_{s_{T-2}\in\mathcal{S}}\mathcal{L}\left(a_{T-1},x_{T-1}|x_{T-2},s_{T-2};\theta_{u},v\right)\cdot\\
 &  & \cdots\cdot\sum_{s_{1}\in\mathcal{S}}\mathcal{L}\left(s_{2}|s_{1};\pi\right)\mathcal{L}\left(a_{1},x_{2}|x_{1},s_{1};\theta_{u},v\right)\cdot\mathcal{L}\left(s_{1}|x_{1};\overline{p_{1}}\right)
\end{eqnarray*}
where $v$ satisfies $v=\Phi_{v}(v;\theta_{u},\pi)$. Under the assumption of stationary distribution in the initial period, $\overline{p_{1}}$ satisfies $\overline{p_{1}}=\Phi_{\overline{p_{1}}}\left(\overline{p_{1}};v,\theta_{u},\pi\right)$, where $\overline{p_{1}}\left(x^{\prime},s^{\prime}\right)=\sum_{x,s}\mathcal{L}\left(x^{\prime},s^{\prime}|x,s;\theta_{u},\pi,v\right)\overline{p_{1}}\left(x,s\right)$. Note that $\mathcal{L}\left(s_{1}|x_{1};\overline{p_{1}}\right)=\frac{\overline{p_{1}}\left(x_{1},s_{1}\right)}{\sum_{x_{1}}\overline{p_{1}}\left(x_{1},s_{1}\right)}$ holds.

By defining $h_{1}\left(s_{1};\theta_{u},\pi,v,\overline{p_{1}}\right)\equiv\mathcal{L}\left(s_{1}|x_{1};\theta_{u},\pi,v,\overline{p_{1}}\right)$ and $h_{t}\left(s_{t};\theta_{u},\pi,v,\overline{p_{1}}\right)\equiv\sum_{s_{t-1}\in\mathcal{S}}\mathcal{L}\left(s_{t}|s_{t-1};\pi\right)\mathcal{L}\left(d_{t-1},x_{t}|x_{t-1},s_{t-1};\theta_{u},v\right)h_{t-1}\left(s_{t-1};\theta_{u},\pi,v,\overline{p_{1}}\right)$ for $t=2,\cdots,T$,

\begin{eqnarray*}
\mathcal{L}\left(a,x|x_{1};\theta_{u},\pi,v,\overline{p_{1}}\right) & = & \mathcal{L}\left(s_{T}|s_{T-1};\pi\right)\mathcal{L}\left(a_{T}|x_{T},s_{T};\theta_{u},\pi,v\right)h_{T}\left(s_{T};\theta_{u},\pi,v,\overline{p_{1}}\right)
\end{eqnarray*}
holds. Hence, given the values of $\theta_{u},\pi,v,\overline{p_{1}}$, we can compute the likelihood $\mathcal{L}\left(a,x|x_{1};\theta_{u},\pi,v,\overline{p_{1}}\right)$.

For estimation, the initial parameter values are drawn randomly as follows. For parameters other than $\theta_{RS_{2}}$ and $\pi$, we first estimate them using the first-step NPL estimator under the assumption of no unobserved heterogeneity, and multiply these estimates by random draws from $U\left[0.5,1.5\right]$. Parameter $\theta_{RS_{2}}$ is drawn from $U[0,2]$, and $\pi$ is drawn from $U[0.1,0.99]$. The initial choice-specific value functions $v$ are those estimated using the first\nobreakdash-step NPL under the assumption of no unobserved heterogeneity. For $\overline{p_{1}}$, we set the initial values to $\frac{1}{|\Omega|}$. To rule out unusually large or small values, we impose additional inequality constraints requiring the parameter values to lie between $-$10 and 10.

\subsection{Dynamic BLP model}

\label{subsec:Dynamic-BLP-details}

As in \citet{sun2019computationally}, let $X_{jt}=[1,\chi_{jt},-p_{jt}]$. Product characteristics $\chi_{jt}$ and $\xi_{jt}$ are generated as $\chi_{jt}\equiv\left[\begin{array}{c}
\chi_{1jt}\\
\chi_{2jt}\\
\chi_{3jt}
\end{array}\right]\sim N\left(\left[\begin{array}{c}
0\\
0\\
0
\end{array}\right],\left[\begin{array}{ccc}
0.5^{2}\\
 & 0.5^{2}\\
 &  & 0.5^{2}
\end{array}\right]\right)$ and $\xi_{jt}\sim N(0,1)$. The price $p_{jt}$ is generated from $p_{jt}=\gamma_{0}+\gamma_{x}^{\prime}\chi_{jt}+\gamma_{z}z_{jt}+\gamma_{w}w_{jt}+\gamma_{\xi}\xi_{jt}-\gamma_{p}^{\prime}\sum_{k\neq j}\chi_{kt}+u_{jt}$, where $z_{jt}=\rho_{0}+\rho_{1}z_{jt-1}+\eta_{jt}$, $\eta_{jt}\sim N(0,0.1^{2})$, $w_{jt}\sim N(0,1)$, $u_{jt}\sim N(0,0.01^{2})$, $[\gamma_{0},\gamma_{X1},\gamma_{X2},\gamma_{X3},\gamma_{z},\gamma_{w},\gamma_{\xi}]=[1,0.2,0.2,0.1,1,0.2,0.7],\gamma_{p}=[0.1,0.1,0.1]$, $z_{j0}=8$, $[\rho_{0},\rho_{z}]=[0.1,0.95]$.

For the demand parameters $\theta_{i}$, let $\theta_{i}=[\theta_{i}^{X0},\theta_{i}^{X1},\theta_{i}^{X2},\theta_{i}^{X3},\theta_{i}^{p}]$ and $\left[\begin{array}{c}
\theta_{i}^{X1}\\
\theta_{i}^{X2}\\
\theta_{i}^{p}
\end{array}\right]\sim N\left(\left[\begin{array}{c}
1\\
1\\
2
\end{array}\right],\left[\begin{array}{ccc}
0.5^{2}\\
 & 0.5^{2}\\
 &  & 0.25^{2}
\end{array}\right]\right)$, $\theta_{i}^{X0}=6,\text{\ensuremath{\theta_{i}^{X3}=0.5}}$. As instruments, we use polynomial expansions of the baseline instruments $\left[\chi_{jt},z_{jt},w_{jt},\sum_{k\neq j}\chi_{kt}\right]$ (up to second order) as in \citet{sun2019computationally} and differentiation IVs(\citet{gandhi2019measuring}). As the weighting matrix for GMM estimation, we use $W=\left(Z^{\prime}Z\right)^{-1}$, where $Z$ denotes the matrix of instrumental variables.

For estimation, the initial values of the standard deviations of the random coefficients $\sigma_{\chi_{1}}$ and $\sigma_{\chi_{2}}$ are randomly drawn as the true value multiplied by a uniform random variable on $U[0,2]$. The initial value of $\sigma_{p}$ is obtained by multiplying its true value by a random draw from $U[0,1.5]$.\footnote{Because the domain of prices is large, a large initial value of $\sigma_{p}$ can sometimes cause the NFXP algorithm to diverge. Therefore, we set the initial value of $\sigma_{p}$\LyXZeroWidthSpace{} relatively small.} To exclude unusually large or small values, we impose inequality constraints requiring the parameter values to lie between $-$5 and 5. The initial value functions $V$ are chosen to satisfy the constraint at the initial parameter values.

We do not impose nonnegativity constraints on the standard deviations of the random coefficients. For these parameters, \citet{monardo2025comparing} demonstrated by numerical experiments on static BLP models that imposing nonnegativity constraints on normal distribution standard deviation parameters may increase the risk of convergence to local rather than global minima in the NFXP algorithm. The present study finds a similar pattern for the SLC algorithm.

\subsubsection*{Local minima}

In 4 of the 20 randomly generated datasets, each evaluated using 5 randomly chosen initial parameter values, the SLC-Spectral and NFXP estimators yield different estimates,\footnote{We assume they differ if the sup norm of the distance is larger than 1E-6.} although the largest difference is less than 0.08. In this numerical experiment, the standard deviation parameters of the random coefficients are not restricted to be positive, which increases the number of potential local minima.\footnote{Because the normal distribution is symmetric, $-\widehat{\sigma_{p}}$ can be the solution if $\widehat{\sigma_{p}}$ is the solution for instance, in the limit where an infinite number of simulation draws is used for the normal distribution. However, in practice the number of simulation draws is finite, and the solution may instead be $-\widetilde{\sigma_{p}}$, which differs slightly from $-\widehat{\sigma_{p}}$.} This study finds that this discrepancy is driven by the choice of initial parameter values: when the SLC is initialized at slightly perturbed versions of the NFXP solution, it converges to the NFXP solution. Moreover, in 1 of the 4 datasets where the two estimators differ, the SLC attains parameter values with a smaller objective value. This suggests that SLC-Spectral performs competitively in locating global optima.

\section{Additional results}

\label{sec:Additional-results-Supp-Appendix}

\subsection{EVFI for Dynamic game estimation}

\label{subsec:EVFI-additional}

This section presents numerical results showing how the Endogenous Value Function Iteration (EVFI), proposed by \citet{bray2019markov} can be used with both NFXP and SLC in the dynamic game model considered in Section \ref{subsec:Dynamic-discrete-game-numerical}. \citet{bray2019markov} showed that EVFI attains faster convergence than the traditional Value Function Iteration (VFI) and the Relative Value Function Iteration (\citealp{bray2019strong}) for solving dynamic discrete choice models.

The key idea behind EVFI is that the relevant object is the shape of the value function over endogenous states for each exogenous state, rather than its levels across all exogenous and endogenous states. Here, rather than the original constraint $v=\Phi_{v}\left(v;\theta\right)$ concerning choice-specific value functions $v$, we consider the following alternative constraint concerning an alternative variable $v_{E}$:

\[
v_{E}=\Lambda\left(\Phi_{v}\left(v_{E};\theta\right)\right)
\]
Here, $\Lambda$ is a mapping defined by:

\[
\Lambda\widetilde{v}\left(\omega\equiv\left(z,y\right),a\right)=\widetilde{v}\left(\omega,a\right)-\frac{1}{|\mathcal{Y}|}\sum_{y^{\prime}\in\mathcal{Y}}\widetilde{v}\left(z,y^{\prime},a_{0}\right).
\]
We let $\mathcal{Y}$ be the set of endogenous states. We assume that the state $\omega$ can be divided into exogenous states $z$ and endogenous states $y$. Unlike exogenous states, endogenous states depend on agents' previous actions. We also choose an arbitrary action $a_{0}$.

By letting $\Pi$ be a mapping from choice-specific value functions $v$ to CCPs, we can obtain $\Pi v=\Pi\Lambda v$.\footnote{Let $C(z)\equiv\frac{1}{|\mathcal{Y}|}\sum_{y^{\prime}\in\mathcal{Y}}\widetilde{v}\left(z,y^{\prime},a_{0}\right)$. Then, CCP of action $a$ given state $\omega$ is given by: 
\begin{eqnarray*}
Pr\left(a|\omega;v\right) & = & Pr\left(v\left(\omega,a\right)+\epsilon\left(\omega,a\right)\geq v\left(\omega,a^{\prime}\right)+\epsilon\left(\omega,a^{\prime}\right)\ \forall a^{\prime}\neq a\right)\\
 & = & Pr\left(v\left(\omega,a\right)-C(z)+\epsilon\left(\omega,a\right)>v\left(\omega,a^{\prime}\right)-C(z)+\epsilon\left(\omega,a^{\prime}\right)\ \forall a^{\prime}\neq a\right)\\
 & = & Pr\left(a|\omega;\Lambda v\right)
\end{eqnarray*}
} This suggests that, instead of solving the original problem subject to the constraints $v=\Phi_{v}\left(v;\theta\right)$ and $\overline{p_{1}}=\Phi_{\overline{p_{1}}}\left(v,\overline{p_{1}};\theta\right)$, it suffices to obtain the estimate $\theta$ by solving the following alternative problem:

\begin{eqnarray*}
 & \min_{\theta,v_{E},\overline{p_{1}}} & Q\left(\theta,v_{E},\overline{p_{1}}\right)\\
 & s.t. & v_{E}=\Lambda\Phi_{v}\left(v_{E};\theta\right)\\
 &  & \overline{p_{1}}=\Phi_{\overline{p_{1}}}\left(v_{E},\overline{p_{1}};\theta\right)
\end{eqnarray*}

Table \ref{tab:EVFI-based} shows the results. The table suggests that using the idea of the EVFI leads to faster convergence of the NFXP and (numerical derivative-based) SLC. Moreover, the VFI-based SLC remains faster than the EVFI-based NFXP. 
\begin{center}
{} 
\begin{table}
\caption{\foreignlanguage{japanese}{Results of numerical experiments (Dynamic game with time-varying unobserved heterogeneity; EVFI-based)}}
\label{tab:EVFI-based}
\begin{centering}
{\footnotesize{}}{\footnotesize{}%
\begin{tabular}{cccccccc}
\hline 
{\footnotesize Analy.} & \multirow{1}{*}{{\footnotesize Algorithm}} & \multicolumn{2}{c}{{\footnotesize Comp. Time (sec)}} & {\footnotesize Feval $Q$} & {\footnotesize Feval $G$} & {\footnotesize\# of} & {\footnotesize Difference with}\tabularnewline
\hline 
{\footnotesize Derivatives} &  & {\footnotesize Mean} & {\footnotesize Std.} & {\footnotesize (mean)} & {\footnotesize (mean)} & {\footnotesize Main iter.} & {\footnotesize SLC-VFI (Analy)}\tabularnewline
\hline 
\multirow{4}{*}{{\footnotesize Yes}} & {\footnotesize SLC-VFI} & {\footnotesize 8.3} & {\footnotesize 1} & {\footnotesize 311.9} & {\footnotesize 9.3} & {\footnotesize 9.3} & {\footnotesize -}\tabularnewline
 & {\footnotesize NFXP-VFI} & {\footnotesize 26.1} & {\footnotesize 1.9} & {\footnotesize 42.8} & {\footnotesize 3999.3} & {\footnotesize 33.4} & {\footnotesize 2.2E-06}\tabularnewline
 & {\footnotesize SLC-EVFI} & {\footnotesize 8.1} & {\footnotesize 1} & {\footnotesize 311.9} & {\footnotesize 9.3} & {\footnotesize 9.3} & {\footnotesize 0}\tabularnewline
 & {\footnotesize NFXP-EVFI} & {\footnotesize 17.6} & {\footnotesize 1.8} & {\footnotesize 42.8} & {\footnotesize 1967.4} & {\footnotesize 33.4} & {\footnotesize 2.2E-06}\tabularnewline
\hline 
\multirow{4}{*}{{\footnotesize No}} & {\footnotesize SLC-VFI} & {\footnotesize 41.1} & {\footnotesize 5} & {\footnotesize 3817.7} & {\footnotesize 5492.3} & {\footnotesize 9.2} & {\footnotesize 1.67E-06}\tabularnewline
 & {\footnotesize NFXP-VFI} & {\footnotesize 164.3} & {\footnotesize 8.1} & {\footnotesize 10074.2} & {\footnotesize 40779.9} & {\footnotesize 33.4} & {\footnotesize 2.20E-06}\tabularnewline
 & {\footnotesize SLC-EVFI} & {\footnotesize 31.8} & {\footnotesize 4.4} & {\footnotesize 3821.7} & {\footnotesize 4288.5} & {\footnotesize 9.1} & {\footnotesize 2.29E-06}\tabularnewline
 & {\footnotesize NFXP-EVFI} & {\footnotesize 91.3} & {\footnotesize 6.7} & {\footnotesize 10071.1} & {\footnotesize 22502} & {\footnotesize 33.4} & {\footnotesize 2.18E-06}\tabularnewline
\hline 
\end{tabular}}{\footnotesize\par}
\par\end{centering}
{\footnotesize Notes. Based on 20 randomly generated datasets with five random initial starting points. $N=640$. ``Feval'' denotes the mean number of function evaluations.}{\footnotesize\par}
\end{table}
\par\end{center}

The speedup of the Jacobian-free SLC achieved by using EVFI can be explained as follows. In particular, using a fixed\nobreakdash-point mapping with a smaller modulus for constructing the constraint function $G$ is generally preferable for achieving faster convergence of iterative methods to solve the associated linear equations.

In general, Krylov subspace methods converge faster for linear systems whose coefficient matrices have smaller condition numbers. For example, for a symmetric positive definite coefficient matrix, the convergence rate of the GMRES is characterized by the condition number of the matrix (cf. Chapter 6 of \citealp{saad2003iterative}).

In the limiting case in which the modulus of the fixed\nobreakdash-point mapping is zero---that is, when $\Phi\left(Y;\theta\right)=\exists h\left(\theta\right)$---the Jacobian $\nabla_{Y}G\left(Y;\theta\right)=\nabla_{Y}\left(Y-\Phi\left(Y;\theta\right)\right)$ becomes the identity matrix. In this case, the condition number is equal to 1, which is the theoretical lower bound.

When the modulus is strictly positive, the Jacobian generally deviates from the identity matrix, and its condition number increases. As the modulus increases, the associated linear system generally becomes more poorly conditioned, slowing the convergence of Krylov methods.

\section{Proof of additional results}

\label{sec:Proof-of-additional-Supp-Appendix}

\subsection{Convexity of $\widetilde{Q}\left(\theta;\gamma\right)$}

\subsubsection*{Proof of Lemma \ref{lem:pd_Q_tilde_Hessian}}

If $\Upsilon(\theta,\gamma)$ is linear concerning $\theta$, $\nabla_{\theta}\Upsilon\left(\theta,\gamma\right)$ does not depend on $\theta$. Hence, if $Q\left(\theta,Y\right)$ does not depend on $\theta$, $\nabla_{\theta}Q\left(\Upsilon(\theta,\gamma)\right)=\left(\nabla_{Y}Q\left(\Upsilon(\theta,\gamma)\right)\right)\left(\nabla_{\theta}\Upsilon\left(\theta,\gamma\right)\right)$ holds, which implies:

\[
\nabla_{\theta\theta^{\prime}}Q\left(\Upsilon(\theta,\gamma)\right)=\left(\nabla_{\theta}\Upsilon(\theta,\gamma)\right)^{\prime}\left(\nabla_{YY^{\prime}}Q\left(\Upsilon(\theta,\gamma)\right)\right)\left(\nabla_{\theta}\Upsilon(\theta,\gamma)\right).
\]

Then, for any $d\in\mathbb{R}^{n_{\theta}}-\{0\}$,

\begin{eqnarray*}
d^{T}\left(\nabla_{\theta\theta^{\prime}}Q\left(\Upsilon(\theta,\gamma)\right)\right)d & = & d^{T}\left(\nabla_{\theta}\Upsilon(\theta,\gamma)\right)^{\prime}\left(\nabla_{YY^{\prime}}Q\left(\Upsilon(\theta,\gamma)\right)\right)\left(\nabla_{\theta}\Upsilon(\theta,\gamma)\right)d\\
 & = & \left(\left(\nabla_{\theta}\Upsilon(\theta,\gamma)\right)d\right)^{\prime}\left(\nabla_{YY^{\prime}}Q\left(\Upsilon(\theta,\gamma)\right)\right)\left(\left(\nabla_{\theta}\Upsilon(\theta,\gamma)\right)d\right)\\
 & \geq & 0\ \left(\because\nabla_{YY^{\prime}}Q\left(\Upsilon(\theta,\gamma)\right)\text{\ is\ positive definite}\right).
\end{eqnarray*}

If $\Upsilon\left(\theta;\gamma_{k}\right)\equiv Y_{k}-\left(\nabla_{Y}G\left(Y_{k};\theta_{k}\right)\right)^{-1}\left(G\left(Y_{k};\theta_{k}\right)\right)-\left(\nabla_{Y}G\left(Y_{k};\theta_{k}\right)\right)^{-1}\left(\nabla_{\theta}G\left(Y_{k};\theta_{k}\right)\right)\left(\theta-\theta_{k}\right)$,

linear independence of the columns of $\nabla_{\theta}G\left(\gamma\right)$ implies

$d\neq0\Rightarrow\nabla_{\theta}G\left(\gamma\right)d\neq0$. If we assume $\left(\nabla_{Y}G\left(\gamma\right)\right)^{-1}\nabla_{\theta}G\left(\gamma\right)d=0$, $\nabla_{\theta}G\left(\gamma\right)d=0$ holds, which contradicts $\nabla_{\theta}G\left(\gamma\right)d\neq0$. Hence, $d\neq0\Rightarrow\left(\nabla_{\theta}\Upsilon(\theta,\gamma)\right)d=\left(\nabla_{Y}G\left(\gamma\right)\right)^{-1}\nabla_{\theta}G\left(\gamma\right)d\neq0$ holds, and

\begin{eqnarray*}
d^{T}\left(\nabla_{\theta\theta^{\prime}}Q\left(\Upsilon(\theta,\gamma)\right)\right)d & = & \left(\left(\nabla_{\theta}\Upsilon(\theta,\gamma)\right)d\right)^{\prime}\left(\nabla_{YY^{\prime}}Q\left(\Upsilon(\theta,\gamma)\right)\right)\left(\left(\nabla_{\theta}\Upsilon(\theta,\gamma)\right)d\right)\\
 & > & 0\ \left(\because\nabla_{YY^{\prime}}Q\left(\Upsilon(\theta,\gamma)\right)\text{\ is\ positive definite}\right).
\end{eqnarray*}

\subsection{Propositions related to the KKT conditions}

\subsubsection*{Proof of Propositions \ref{prop:fixed-points-KKT-seq-ineq}}

(a). Let $\widetilde{\gamma}\equiv\left(\widetilde{\theta},\widetilde{Y}\right)\in\Gamma_{seq}$ be a fixed point of Algorithm \ref{alg:general_seq}. By $\widetilde{\theta}=\arg\min_{\theta\in\Theta}\widetilde{Q}\left(\theta;\widetilde{\gamma}\right)$, there exists $\widetilde{\lambda_{B}}$ such that

\begin{eqnarray*}
 &  & \nabla_{\theta}\widetilde{Q}\left(\widetilde{\theta};\widetilde{\gamma}\right)-\widetilde{\lambda_{B}}^{T}\nabla_{\theta}B\left(\widetilde{\theta}\right)=0\\
 &  & B\left(\widetilde{\theta}\right)\geq0\\
 &  & \widetilde{\lambda_{B}}\geq0\\
 &  & \widetilde{\lambda_{B_{i}}}B_{i}\left(\widetilde{\theta}\right)=0\ \forall i=1,\cdots n_{B}
\end{eqnarray*}

Then, $\nabla_{\theta}\Upsilon\left(\theta;\gamma=\left(\theta,Y\right)\right)=-\left(\nabla_{Y}G\left(Y;\theta\right)\right)^{-1}\left(\nabla_{\theta}G\left(Y;\theta\right)\right)$ by Assumption \ref{as:mapping_cdns} (c) and $\nabla_{\theta}\widetilde{Q}\left(\widetilde{\theta};\widetilde{\gamma}\right)=\nabla_{\theta}Q\left(\widetilde{\theta},\widetilde{Y}\right)+\left(\nabla_{Y}Q\left(\widetilde{\theta},\widetilde{Y}\right)\right)\left(\nabla_{\theta}\Upsilon\left(\widetilde{\theta};\widetilde{\gamma}\right)\right)$ implies:

\[
\nabla_{\theta}Q\left(\widetilde{\theta},\widetilde{Y}\right)-\left(\nabla_{Y}Q\left(\widetilde{\theta},\widetilde{Y}\right)\right)\left(\nabla_{Y}G\left(\widetilde{Y};\widetilde{\theta}\right)\right)^{-1}\left(\nabla_{\theta}G\left(\widetilde{Y};\widetilde{\theta}\right)\right)-\widetilde{\lambda_{B}}^{T}\left(\nabla_{\theta}B\left(\widetilde{\theta}\right)\right)=0
\]
In addition, $\Upsilon\left(\widetilde{\theta};\widetilde{\gamma}\right)=\widetilde{Y}$ and Assumption \ref{as:mapping_cdns} (a) implies $G\left(\widetilde{Y};\widetilde{\theta}\right)=0$.

Therefore, the discussion above implies that there exist $\lambda_{G}^{T}=\left(\nabla_{Y}Q\left(\widetilde{\gamma}\right)\right)\left(\nabla_{Y}G\left(\widetilde{\gamma}\right)\right)^{-1}$ and $\lambda_{B}=\widetilde{\lambda_{B}}$ such that

\begin{eqnarray*}
 &  & \nabla_{\theta}Q\left(\widetilde{\gamma}\right)-\lambda_{G}^{T}\nabla_{\theta}G\left(\widetilde{\gamma}\right)-\lambda_{B}^{T}\nabla_{\theta}B\left(\widetilde{\theta}\right)=0\\
 &  & \nabla_{Y}Q\left(\widetilde{\gamma}\right)-\lambda_{G}^{T}\nabla_{Y}G\left(\widetilde{\gamma}\right)=0\\
 &  & G\left(\widehat{\gamma}\right)=0\\
 &  & B\left(\widetilde{\theta}\right)\geq0\\
 &  & \lambda_{B}\geq0\\
 &  & \lambda_{B_{i}}B_{i}\left(\widetilde{\theta}\right)=0\ \forall i=1,\cdots n_{B}
\end{eqnarray*}

Hence, $\widetilde{\gamma}\equiv\left(\widetilde{\theta},\widetilde{Y}\right)\in\Gamma_{KKT}$ holds.

(b). Let $\widetilde{\gamma}\equiv\left(\widetilde{\theta},\widetilde{Y}\right)\in\Gamma_{KKT}$. Then, there exist $\lambda_{G}$ and $\lambda_{B}$ such that:

\begin{eqnarray*}
 &  & \nabla_{\theta}Q\left(\widetilde{\gamma}\right)-\lambda_{G}^{T}\nabla_{\theta}G\left(\widetilde{\gamma}\right)-\lambda_{B}^{T}\nabla_{\theta}B\left(\widetilde{\theta}\right)=0\\
 &  & \nabla_{Y}Q\left(\widetilde{\gamma}\right)-\lambda_{G}^{T}\nabla_{Y}G\left(\widetilde{\gamma}\right)=0\\
 &  & G\left(\widetilde{\gamma}\right)=0\\
 &  & B\left(\widetilde{\theta}\right)\geq0\\
 &  & \lambda_{B}\geq0\\
 &  & \lambda_{B_{i}}B_{i}\left(\widetilde{\theta}\right)=0\ \forall i=1,\cdots n_{B}
\end{eqnarray*}

Then,

\begin{eqnarray*}
0 & = & \nabla_{\theta}Q\left(\widetilde{\gamma}\right)-\left(\nabla_{Y}Q\left(\widetilde{\theta},\widetilde{Y}\right)\right)\left(\nabla_{Y}G\left(\widetilde{Y};\widetilde{\theta}\right)\right)^{-1}\left(\nabla_{\theta}G\left(\widetilde{Y};\widetilde{\theta}\right)\right)-\lambda_{B}^{T}\left(\nabla_{\theta}B\left(\widetilde{\theta}\right)\right)\\
 & = & \nabla_{\theta}\widetilde{Q}\left(\widetilde{\theta},\widetilde{Y}\right)-\lambda_{B}^{T}\left(\nabla_{\theta}B\left(\widetilde{\theta}\right)\right)\ \left(\because\text{Assumption \ref{as:mapping_cdns} (c)}\right)
\end{eqnarray*}

Because there exists $\lambda_{B}$ such that

\begin{eqnarray*}
 &  & \nabla_{\theta}\widetilde{Q}\left(\widetilde{\theta};\widetilde{\gamma}\right)-\lambda_{B}^{T}\nabla_{\theta}B\left(\widetilde{\theta}\right)=0\\
 &  & B\left(\widetilde{\theta}\right)\geq0\\
 &  & \lambda_{B}\geq0\\
 &  & \lambda_{B_{i}}B_{i}\left(\widetilde{\theta}\right)=0\ \forall i=1,\cdots n_{B},
\end{eqnarray*}
and $\nabla_{\theta\theta^{\prime}}\widetilde{Q}\left(\widetilde{\theta};\widetilde{\gamma}\right)$ is positive definite, $\widetilde{\theta}=\arg\min_{\theta\in\Theta}\widetilde{Q}\left(\theta;\widetilde{\gamma}\right)$ holds.

Because $G\left(\widetilde{\gamma}\right)=0$ implies $\widetilde{Y}=\Upsilon\left(\widetilde{\theta};\widetilde{\gamma}\right)$ by Assumption \ref{as:mapping_cdns}(a), $\widetilde{\gamma}\equiv\left(\widetilde{\theta},\widetilde{Y}\right)\in\Gamma_{KKT}$ is a fixed point of Algorithm \ref{alg:general_seq}. Hence, $\widetilde{\gamma}\in\Gamma_{seq}$ holds.

(c). By setting $\widetilde{\gamma}=\widehat{\gamma}\in\Gamma_{KKT}$ in the proof of (b), an analogous argument applies, yielding $\widehat{\gamma}=\widetilde{\gamma}\in\Gamma_{seq}$.

\subsection{Termination of backtracking (Proof of Lemma \ref{lem:finite-terminate-backtracking})}

In what follows, let $\Delta_{k}\equiv D\left(\phi_{1}\left(\theta_{k},Y_{k};\mu_{k}\right);d_{k}\right)$. To show Lemma \ref{lem:finite-terminate-backtracking}, we use the following lemma.

\label{lem:Delta_k_upperbound}If $\mu_{k}=\max\left\{ \mu_{k-1},\frac{\left(\nabla_{\gamma}Q\left(\gamma_{k}\right)\right)d_{k}+c\left\Vert d_{k}\right\Vert ^{2}}{(1-\kappa)\left\Vert G\left(\gamma_{k}\right)\right\Vert _{1}+\eta_{k}}\right\} $, $\Delta_{k}=\left(\nabla_{\gamma}Q\left(\gamma_{k}\right)\right)d_{k}-\mu_{k}\left\Vert G\left(\gamma_{k}\right)\right\Vert _{1}\leq-c\left\Vert d_{k}\right\Vert ^{2}+\mu_{k}\eta_{k}\leq\mu_{k}\eta_{k}$. Because $\mu_{k}\geq\frac{\left(\nabla_{\gamma}Q\left(\gamma_{k}\right)\right)d_{k}+c\left\Vert d_{k}\right\Vert ^{2}}{(1-\kappa)\left\Vert G\left(\gamma_{k}\right)\right\Vert _{1}+\eta_{k}}$,

\[
\left(1-\kappa\right)\mu_{k}\left\Vert G\left(\gamma_{k}\right)\right\Vert _{1}+\mu_{k}\eta_{k}\geq\left(\nabla_{\gamma}Q\left(\gamma_{k}\right)\right)d_{k}+c\left\Vert d_{k}\right\Vert ^{2}
\]
holds. Hence, we obtain:

\begin{eqnarray*}
\Delta_{k}=\left(\nabla_{\gamma}Q\left(\gamma_{k}\right)\right)d_{k}-\mu_{k}\left\Vert G\left(\gamma_{k}\right)\right\Vert _{1} & \leq & -\kappa\mu_{k}\left\Vert G\left(\gamma_{k}\right)\right\Vert _{1}-c\left\Vert d_{k}\right\Vert ^{2}+\mu_{k}\eta_{k}\\
 & \leq & -c\left\Vert d_{k}\right\Vert ^{2}+\mu_{k}\eta_{k}\\
 & \leq & \mu_{k}\eta_{k}
\end{eqnarray*}

\subsubsection*{Proof of Lemma \ref{lem:finite-terminate-backtracking}}

Because $\Delta_{k}=\lim_{\alpha\rightarrow0}\frac{\phi_{1}\left(\gamma_{k}+\alpha d_{k};\mu\right)-\phi_{1}\left(\gamma_{k};\mu\right)}{\alpha}$, $\exists\overline{\alpha_{k}}>0$ s.t. $\frac{\phi_{1}\left(\gamma_{k}+\overline{\alpha_{k}}d_{k};\mu_{k}\right)-\phi_{1}\left(\gamma_{k};\mu_{k}\right)}{\overline{\alpha_{k}}}<\Delta_{k}+(1-\zeta)c\left\Vert d_{k}\right\Vert ^{2}$. By letting $\alpha_{k}=\frac{\overline{\alpha_{k}}}{\sigma_{k}}$, we obtain 
\begin{eqnarray*}
 &  & \phi_{1}\left(\gamma_{k}+\alpha_{k}\sigma_{k}d_{k};\mu_{k}\right)-\phi_{1}\left(\gamma_{k};\mu_{k}\right)\\
 & < & \alpha_{k}\sigma_{k}\left(\Delta_{k}+(1-\zeta)c\left\Vert d_{k}\right\Vert ^{2}\right)\\
 & = & \zeta\alpha_{k}\sigma_{k}\Delta_{k}+\left(1-\zeta\right)\alpha_{k}\sigma_{k}\left(\Delta_{k}+c\left\Vert d_{k}\right\Vert ^{2}\right)\\
 & \leq & \zeta\alpha_{k}\sigma_{k}\Delta_{k}+\left(1-\zeta\right)\alpha_{k}\sigma_{k}\mu_{k}\eta_{k}\ \left(\because\Delta_{k}\leq-c\left\Vert d_{k}\right\Vert ^{2}+\mu_{k}\eta_{k}\text{\ by\ Lemma \ref{lem:Delta_k_upperbound}}\right)
\end{eqnarray*}

\subsection{Global convergence (Proof of Proposition \ref{prop:Global-convergence})}

\subsubsection*{Proof of Proposition \ref{prop:Global-convergence}}

The following proposition is motivated by the proof of the global convergence of SQP algorithm for solving constrained optimization problems (cf. Section 17 of \citealp{bonnans2006numerical}). Because $\phi_{1}\left(\gamma_{k+1};\overline{\mu}\right)<\phi_{1}\left(\gamma_{k};\overline{\mu}\right)+\zeta\alpha_{k}\sigma_{k}\Delta_{k}+\left(1-\zeta\right)\alpha_{k}\sigma_{k}\overline{\mu}\eta_{k}$ for $k\geq\overline{K}$, the following holds for $j\in\mathbb{N}$ and $k\geq\overline{K}$:

\begin{eqnarray*}
\phi_{1}\left(\gamma_{k+j+1};\overline{\mu}\right) & < & \phi_{1}\left(\gamma_{k};\overline{\mu}\right)+\zeta\sum_{m=0}^{j}\alpha_{k+m}\sigma_{k+m}\Delta_{k+m}+\left(1-\zeta\right)\sum_{m=0}^{j}\alpha_{k+m}\sigma_{k+m}\overline{\mu}\eta_{k+m}\\
 & \leq & \phi_{1}\left(\gamma_{k};\overline{\mu}\right)+\zeta\overline{\sigma}\sum_{m=0}^{j}\alpha_{k+m}\Delta_{k+m}+\left(1-\zeta\right)\overline{\sigma}\overline{\mu}\sum_{m=0}^{j}\eta_{k+m}
\end{eqnarray*}

If the merit function is bounded, $\sum_{m=0}^{j}\eta_{k+m}<\infty$ implies $-\infty<\sum_{m=0}^{j}\alpha_{k+m}\Delta_{k+m}<\infty$, namely, $\lim_{k\rightarrow\infty}\alpha_{k}\Delta_{k}=0$. If $\alpha_{k}\geq\exists\underline{\alpha}>0\ \forall k$, $\lim_{k\rightarrow\infty}\Delta_{k}\rightarrow0$ holds. Because $\Delta_{k}\leq-c\left\Vert d_{k}\right\Vert ^{2}+\mu_{k}\eta_{k}$ by Lemma \ref{lem:Delta_k_upperbound} and $\lim_{k\rightarrow\infty}\eta_{k}\rightarrow0,\mu_{k}\leq\overline{\mu}$, $\lim_{k\rightarrow\infty}d_{k}\rightarrow0$ holds.

Hence, it suffices to show $\lim_{k\rightarrow\infty}\alpha_{k}>0$. Suppose $\lim_{k\rightarrow\infty}\alpha_{k}=0$ and $\left\Vert d_{k}\right\Vert >0\ \forall k\in\mathbb{N}$. Then, by $\left|\sigma_{k}\right|\leq\overline{\sigma}$, $\exists\overline{K_{1}}\in\mathbb{N}$ s.t. $k\geq\overline{K_{1}}\Rightarrow0<\alpha_{k}\sigma_{k}<1$ and $0<\alpha_{k}\leq\frac{1}{\tau_{\alpha}}<1$. Hence, for $k\geq\overline{K_{1}}$, 
\begin{equation}
\phi_{1}\left(\gamma_{k}+\frac{1}{\tau_{\alpha}}\alpha_{k}\sigma_{k}d_{k};\mu_{k}\right)>\phi_{1}\left(\gamma_{k};\mu_{k}\right)+\zeta\frac{1}{\tau_{\alpha}}\alpha_{k}\sigma_{k}\Delta_{k}+\left(1-\zeta\right)\frac{1}{\tau_{\alpha}}\alpha_{k}\sigma_{k}\mu_{k}\eta_{k}\label{eq:ineq_for_global_conv_1}
\end{equation}
holds.

By Lemma \ref{lem: ineq_merit_func}, 
\begin{equation}
\phi_{1}\left(\gamma_{k}+\frac{1}{\tau_{\alpha}}\alpha_{k}\sigma_{k}d_{k};\mu_{k}\right)\leq\phi_{1}\left(\gamma_{k};\mu_{k}\right)+\frac{1}{\tau_{\alpha}}\alpha_{k}\sigma_{k}\Delta_{k}+\exists C_{1}\left(\frac{1}{\tau_{\alpha}}\alpha_{k}\sigma_{k}\right)^{2}\left\Vert d_{k}\right\Vert ^{2}\label{eq:ineq_for_global_conv_2}
\end{equation}
holds, where $C_{1}>0$.

Then, for $k\geq\overline{K_{1}}$, (\ref{eq:ineq_for_global_conv_1}) and (\ref{eq:ineq_for_global_conv_2}) imply

\[
\phi_{1}\left(\gamma_{k};\mu_{k}\right)+\zeta\frac{1}{\tau_{\alpha}}\alpha_{k}\sigma_{k}\Delta_{k}+\left(1-\zeta\right)\frac{1}{\tau_{\alpha}}\alpha_{k}\sigma_{k}\mu_{k}\eta_{k}<\phi_{1}\left(\gamma_{k};\mu_{k}\right)+\frac{1}{\tau_{\alpha}}\alpha_{k}\sigma_{k}\Delta_{k}+C_{1}\left(\frac{1}{\tau_{\alpha}}\alpha_{k}\sigma_{k}\right)^{2}\left\Vert d_{k}\right\Vert ^{2},
\]
namely,

\[
\left(1-\zeta\right)\frac{1}{\tau_{\alpha}}\alpha_{k}\sigma_{k}\left(\mu_{k}\eta_{k}-\Delta_{k}\right)<C_{1}\left(\frac{1}{\tau_{\alpha}}\alpha_{k}\sigma_{k}\right)^{2}\left\Vert d_{k}\right\Vert ^{2}.
\]

Because $\mu_{k}$ is chosen so that $\Delta_{k}\leq\mu_{k}\eta_{k}-c\left\Vert d_{k}\right\Vert ^{2}$ holds by Lemma \ref{lem:Delta_k_upperbound}, 
\[
c\left(1-\zeta\right)\frac{1}{\tau_{\alpha}}\alpha_{k}\sigma_{k}\left\Vert d_{k}\right\Vert ^{2}<C_{1}\left(\frac{1}{\tau_{\alpha}}\alpha_{k}\sigma_{k}\right)^{2}\left\Vert d_{k}\right\Vert ^{2},
\]
namely, 
\begin{equation}
0<\frac{1}{\tau_{\alpha}}\alpha_{k}\sigma_{k}C_{1}\left(\frac{1}{\tau_{\alpha}}\alpha_{k}\sigma_{k}-\frac{c}{C_{1}}\left(1-\zeta\right)\right)\left\Vert d_{k}\right\Vert ^{2}\label{eq:ineq_global_conv_contradiction}
\end{equation}
holds for $k\geq\overline{K_{1}}$.

By $\lim_{k\rightarrow\infty}\alpha_{k}\sigma_{k}=0$, there exists $\overline{K_{2}}>\overline{K_{1}}$ such that $k\geq\overline{K_{2}}\Rightarrow\frac{1}{\tau_{\alpha}}\alpha_{k}\sigma_{k}-\frac{c}{C_{1}}\left(1-\zeta\right)<0$. Then, $\frac{1}{\tau_{\alpha}}\alpha_{k}\sigma_{k}-\frac{c}{C_{1}}\left(1-\zeta\right)<0$ holds for $k\geq\overline{K_{2}}>\overline{K_{1}}$. This implies that the right hand side of (\ref{eq:ineq_global_conv_contradiction}) takes a negative value, which is a contradiction. Therefore, $\lim_{k\rightarrow\infty}\alpha_{k}>0$.
\end{document}